\documentclass[a4paper,english]{article}
\usepackage[T1]{fontenc}
\usepackage[latin9]{inputenc}
\usepackage[active]{srcltx}
\usepackage{babel}
\usepackage{mathtools}
\usepackage{amsthm}
\usepackage{amsmath}
\usepackage{amssymb}
\usepackage{graphicx,epstopdf}
\usepackage{setspace}
\usepackage{esint}
\usepackage[numbers]{natbib}
\usepackage[unicode=true,pdfusetitle,
 bookmarks=false,
 breaklinks=false,pdfborder={0 0 1},backref=false,colorlinks=false]
 {hyperref}

\makeatletter

\providecommand{\tabularnewline}{\\}

\newcommand{\lyxaddress}[1]{
\par {\raggedright #1
\vspace{1.4em}
\noindent\par}
}
\theoremstyle{plain}
\newtheorem{thm}{\protect\theoremname}
  \theoremstyle{plain}
  \newtheorem{assumption}[thm]{\protect\assumptionname}
  \theoremstyle{definition}
  \newtheorem{defn}[thm]{\protect\definitionname}
  \theoremstyle{definition}
  \newtheorem{eigenproblem}[thm]{\protect\eigenproblemname}
  \theoremstyle{remark}
  \newtheorem{rem}[thm]{\protect\remarkname}
  \theoremstyle{plain}
  \newtheorem{lem}[thm]{\protect\lemmaname}
  \theoremstyle{plain}
  \newtheorem{prop}[thm]{\protect\propositionname}
  \theoremstyle{plain}
  \newtheorem{cor}[thm]{\protect\corollaryname}

\@ifundefined{date}{}{\date{}}
\usepackage{a4wide,diagbox}

\renewcommand{\hat}{\widehat}
\renewcommand{\tilde}{\widetilde}

\allowdisplaybreaks

\makeatother

  \providecommand{\assumptionname}{Assumption}
  \providecommand{\corollaryname}{Corollary}
  \providecommand{\definitionname}{Definition}
  \providecommand{\eigenproblemname}{Eigenproblem}
  \providecommand{\lemmaname}{Lemma}
  \providecommand{\propositionname}{Proposition}
  \providecommand{\remarkname}{Remark}
\providecommand{\theoremname}{Theorem}

\begin{document}
\global\long\def\P{\mathbb{P}_{\sigma,b}}

\global\long\def\E{\mathbb{E}_{\sigma,b}}

\global\long\def\Var{\operatorname{Var}_{\sigma,b}}

\global\long\def\L{\mathcal{L}_{v}}

\global\long\def\RG{\mathcal{R}}

\global\long\def\TG{\mathcal{T}}

\global\long\def\F{\mathcal{F}}

\global\long\def\I{\mathbf{1}}

\global\long\def\N{\mathbb{N}}

\global\long\def\Z{\mathbb{Z}}

\global\long\def\R{\mathbb{R}}

\global\long\def\mc{\omega(\Delta)}

\global\long\def\JC{\maltese}

\title{Nonparametric volatility estimation in scalar diffusions:\linebreak{}
Optimality across observation frequencies.}

\author{Jakub Chorowski%
\thanks{Supported by the Deutsche Forschungsgemeinschaft (DFG) RTG 1845 ''Stochastic
Analysis with Applications in Biology, Finance and Physics\textquotedbl{}.%
}}

\maketitle

\lyxaddress{\begin{center}
Institute of Mathematics\linebreak{}
 Humboldt-Universität zu Berlin\linebreak{}
 chorowsj@math.hu-berlin.de
\par\end{center}}
\begin{abstract}
The nonparametric volatility estimation problem of a scalar diffusion
process observed at equidistant time points is addressed. Using the
spectral representation of the volatility in terms of the invariant
density and an eigenpair of the infinitesimal generator the first
known estimator that attains the minimax optimal convergence rates
for both high and low-frequency observations is constructed. The proofs
are based on a posteriori error bounds for generalized eigenvalue
problems as well as the path properties of scalar diffusions and stochastic
analysis. The finite sample performance is illustrated by a numerical
example.
\end{abstract}
\begin{doublespace}
\textbf{MSC2010 subject classification:} Primary 62M05; Secondary
62G99, 62M15, 60J60. 
\end{doublespace}


\section{Introduction}

Consider the problem of estimating the volatility of a diffusion process
$(X_{t},t\geq0).$ The statistical properties depend, essentially,
on the observation scheme. It is natural to assume discrete observations:
\[
X_{0},X_{\Delta},...,X_{N\Delta},\quad\Delta>0,\quad T=N\Delta.
\]

The quality of an estimator is typically assessed by its asymptotic
properties when the sample size $N$ tends to infinity. The usual
assumptions are either $\Delta\to0$ or $T\to\infty,$ which corresponds
to high and low-frequency regimes, respectively. Different frequency
assumptions require very different methods. Since the frequency regimes
are a theoretical construct, for any given sample, we need to choose
among high and low-frequency estimators. Therefore, it is of crucial
interest to develop universal methods that will perform at optimal
level regardless of the sampling frequency. In this paper, the first
nonparametric estimator of the volatility that attains minimax optimal
rates in both high and low-frequency regimes is introduced. In the
parametric setting, the problem of the universal scale estimation
was first raised in Jacobsen \citep{Jacobsen:2001,Jacobsen:2002}.
The constructed estimators were consistent and asymptotically Gaussian
for all values of $\Delta,$ but nearly efficient for small values
of $\Delta$ only. The estimation method, which relied on the use
of the estimating functions, is different from the one applied in
this paper.

It is a well-known consequence of the Girsanov theorem that when $T$
is fixed, the drift coefficient is not identifiable. Since we are
interested in a universal scale method, we focus on the volatility
estimation and, henceforth, treat drift as a nuisance parameter.

The existing high-frequency estimators (see \citet{Florens-Zmirou:1993,Hoffmann:1999A,Jacod:2000,BandiPhillips:2003})
are based on the interpretation of the squared volatility as the instantaneous
conditional variance of the process. Consequently, the assumption
$\Delta\to0$ is crucial for the consistency of these estimators,
see \citep{Florens-Zmirou:1989} and \citep[Section 3]{Sorensen:2004}.
On the other hand, it has been conjectured that the minimax optimal
low-frequency estimator introduced by Gobet, Hoffmann and Reiß (GHR)
\citep{GobetHoffmannReiss:2004} also performs well in the high-frequency
regime. This conjecture is based on the observation that the spectral
representation of the volatility in terms of an eigenpair of the infinitesimal
generator can be generalized by replacing the invariant density with
the occupation density of the path $(X_{t},t\leq T).$ While this
generalization might be sufficient to obtain the consistency of the
GHR estimator when applied to the high-frequency data, the numerical
study reveals that the convergence rates are not optimal. The reason
for this is that when the time horizon of the sample is fixed, the
estimator inherits the poor regularity of the occupation density,
which, contrary to the invariant density, is not linked to the regularity
of the diffusion coefficients. As we show below, this difficulty can
be solved with the appropriate averaging of the spectral estimator,
which is the main motivation behind the Definition \ref{def:definition of spectral estimator}
of the universally optimal estimator. For more details, refer to section
\ref{sub:Connection to GHR estimator}.

Based on the spectral method, the low-frequency analysis of the universally
optimal estimator is similar to \citep{GobetHoffmannReiss:2004,ChorowskiTrabs:2015}.
The real difficulty is in the high-frequency analysis, where the universal
estimator is compared to the benchmark high-frequency estimator introduced
by \citet{Florens-Zmirou:1993} (see Section \ref{sub:Connection to FZ estimator}).
In particular, we develop the perturbation theory for bilinear coercive
forms with Hölder regular coefficients (see Appendix \ref{sec:Eigenvalue problem for form a}),
which may be of independent interest.

In the next Sections, we present the construction of the universal
scale estimator and state the high and low-frequency convergence rates.
In Section \ref{sec:Discussion} we discuss the relation of the proposed
estimator to the high and low-frequency benchmark estimators. Finite
sample behaviour of the new estimator compared with the Florens-Zmirou
and GHR estimators is illustrated in Section \ref{sub:Numerical Example}.
In Section \ref{sub:Limitations and conclusions} we discuss the assumptions
and possible extensions of the model. The proofs of the high and low-frequency
convergence rates are shown in Sections \ref{sec:Proof in HF} and
\ref{sec:proof in LF}, respectively.

\subsection{Construction of the estimator\label{sub:Construction estimator}}

We follow the low-frequency literature \citep{GobetHoffmannReiss:2004,NicklSohl:2015,ChorowskiTrabs:2015}
and consider a diffusion model on $[0,1]$ with boundary reflection
(see Section \ref{sub:Limitations and conclusions} for a discussion
of the model). Let $\|\cdot\|_{\infty}$ denote the supremum norm
on space $B([0,1])$ of bounded measurable functions on $[0,1]$.
Finally, denote by 
\[
H^{i}=\left\{ f\in L^{2}([0,1]):\: f\text{ has }i\text{ weak derivatives with }f^{(j)}\in L^{2}([0,1]),\, j\leq i\right\} 
\]
 the $L^{2}-$Sobolev spaces on $[0,1]$ of order $i=1,2$. $H^{i}$
is a Hilbert space with the norm
\[
\|f\|_{H^{i}}=\sum_{j\leq i}\|f^{(i)}\|_{L^{2}}.
\]

\begin{assumption}
\label{assu: Set Theta}For given constants $0<d<D$ suppose $(\sigma,b)\in\Theta$,
where 
\[
\Theta:=\Theta(d,D)=\{(\sigma,b)\in H^{1}([0,1])\times B([0,1]):\|b\|_{\infty}\vee\|\sigma^{2}\|_{H^{1}}<D,\inf_{x\in[0,1]}\sigma^{2}(x)\ge d\}.\qquad
\]

\end{assumption}
Let the process $(X_{t},t\geq0)$ be given by the following Skorokhod
type stochastic differential equation:
\begin{eqnarray}
dX_{t} & = & b(X_{t})dt+\sigma(X_{t})dW_{t}+dK_{t},\label{eq:SDE for X}\\
X_{t} & \in & [0,1]\text{ for every }t\geq0,\nonumber 
\end{eqnarray}
where $(W_{t},t\geq0)$ is a standard Brownian motion and $(K_{t},t\geq0)$
is an adapted continuous process with finite variation, starting from
$0$, such that for every $t\geq0$ we have $\int_{0}^{t}\I_{(0,1)}(X_{s})dK_{s}=0.$
The Sobolev regularity of $\sigma$ ensures that the SDE (\ref{eq:SDE for X})
has a unique strong solution, see \citep[Theorem 4]{Veretennikov:1979}.
As shown in \citep{GobetHoffmannReiss:2004}, $X$ admits an invariant
measure with Lebesgue density.
\begin{assumption}
\label{assu: Initial condition}The initial condition $x_{0}$ is
distributed with respect to the invariant measure $\mu$ on $[0,1]$,
independently of the driving Brownian motion $W$.
\end{assumption}
Under Assumption \ref{assu: Initial condition}, the diffusion $X$
is stationary and ergodic. We denote with $\P$ the law of $X$ on
the canonical space $\Omega$ of continuous functions over the positive
axis with values in $[0,1]$, equipped with the topology of the uniform
convergence on compact sets and endowed with its $\sigma-$field $\mathcal{F}$.
We denote with $\E$ the corresponding expectation operator.
\begin{defn}
\label{def:Definition of mu_N}Denote by $\hat{\mu}_{N}$ the empirical
measure associated to the observed sample: 
\[
\hat{\mu}_{N}=\frac{1}{2N}\delta_{\{X_{0}\}}+\frac{1}{N}\sum_{n=1}^{N-1}\delta_{\{X_{n\Delta}\}}+\frac{1}{2N}\delta_{\{X_{N\Delta}\}}.
\]

\end{defn}
The underweighting of the first and the last observations is asymptotically
negligible, but has meaningful finite sample interpretation both in
the low and high-frequency regimes (see remarks before the equation
(\ref{eq:form l as a linear combination of g and p}) and after Definition
\ref{def:Florens-Zmirou estimator}). By ergodicity, when the time
horizon $T$ of the observed sample grows to infinity, the empirical
measure $\hat{\mu}_{N}(dx)$ converges weakly to the stationary distribution
$\mu(dx)$. When $T$ is fixed, but the observation frequency increases,
the empirical measure tends to the occupation measure $\mu_{T}$ of
the path $(X_{t},0\leq t\leq T)$ (see Definition \ref{def:Occupation density}). 
\begin{defn}
For $J\in\N_{+},\: j=1,...,J$, let $\I_{j}(x)=\I(\frac{j-1}{J}\leq x<\frac{j}{J})$
be the indicator function of the $j^{\text{th}}$ sub-interval and
\begin{eqnarray*}
\psi_{j}(x) & = & \int_{0}^{x}\I_{j}(y)dy,\text{ for }j=1,...,J,\\
\psi_{0}(x) & = & 1.
\end{eqnarray*}
Let $V_{J}=\operatorname{span}\{\psi_{j}:j=0,...,J\}$ be the space
of linear splines with knots at $\{0,\frac{1}{J},\frac{2}{J},...,\frac{J-1}{J},1\}$
and $V_{J}^{0}=\{v\in V_{J}:\int_{0}^{1}v(x)\hat{\mu}_{N}(dx)=0\}$
be the subspace of functions $L^{2}(\hat{\mu}_{N})-$orthogonal to
constants.
\end{defn}
Consider the generalized symmetric eigenproblem: 
\begin{eigenproblem}
\label{EigenProb:  l_hat g_hat}Find $(\hat{\gamma},\hat{u})\in\R\times V_{J}$
with $\hat{u}\neq0$, such that 
\[
\hat{l}(\hat{u},v)=\hat{\gamma}\hat{g}(\hat{u},v),\text{ for all }v\in V_{J},
\]
where $\hat{g},\hat{l}:V_{J}\times V_{J}\to\R$ are symmetric, bilinear
forms defined by:
\begin{eqnarray*}
\hat{g}(u,v) & = & \int_{0}^{1}u(x)v(x)\hat{\mu}_{N}(dx),\\
\hat{l}(u,v) & = & \frac{1}{2T}\sum_{n=0}^{N-1}\big(u(X_{(n+1)\Delta})-u(X_{n\Delta})\big)\big(v(X_{(n+1)\Delta})-v(X_{n\Delta})\big).
\end{eqnarray*}

\end{eigenproblem}
When the observed sample visits at least twice every interval $[\frac{j-1}{J},\frac{j}{J}),$
the form $\hat{g}$ is positive definite on $V_{J}$, while $\hat{l}$
is positive semi-definite on $V_{J}$ and positive definite on $V_{J}^{0}$.
In such a case, Eigenproblem \ref{EigenProb:  l_hat g_hat} has $\text{dim}(V_{J})=J+1$
solutions $(\hat{\gamma}_{j},\hat{u}_{j})_{j=0,...,J}$, with non-negative
eigenvalues $0\leq\hat{\gamma}_{0}\leq\hat{\gamma}_{1}\leq...\leq\hat{\gamma}_{J}$
and $\hat{g}-$orthogonal eigenfunctions. It is easy to check that
$\hat{\gamma}_{0}=0$ is an eigenvalue which corresponds to the constant
function. Since the eigenfunctions are $\hat{g}-$orthogonal, it follows
that $\hat{u}_{j}\in V_{J}^{0}$ for $1\leq j\leq J$. Consequently,
$\hat{\gamma}_{1}>0$. 
\begin{defn}
\label{def:definition of spectral estimator} Let
\[
\hat{\zeta}_{1}=\frac{\log(1-\Delta\hat{\gamma}_{1})}{\Delta}\I(\Delta\hat{\gamma}_{1}<1)\quad\text{and}\quad\hat{u}_{1}(x)=\sum_{j=0}^{J}\hat{u}_{1,j}\psi_{j}(x).
\]
When $\hat{u}_{1,j}\neq0$ we define the spectral estimator by
\begin{eqnarray*}
\hat{\sigma}_{S,j}^{2} & = & \frac{-2\hat{\zeta}_{1}\int_{0}^{1}\psi_{j}(x)\hat{u}_{1}(x)\hat{\mu}_{N}(dx)}{\int_{0}^{1}\psi_{j}'(x)\hat{u}_{1,j}\hat{\mu}_{N}(dx)},\\
\hat{\sigma}_{S}^{2}(x) & = & \sum_{j=1}^{J}\hat{\sigma}_{S,j}^{2}\I_{j}(x).
\end{eqnarray*}

\end{defn}
The condition $\I(\Delta\hat{\gamma}_{1}<1)$ is a technical assumption
which ensures that the estimator $\hat{\zeta}_{1}$ is well defined.
As explained in Section \ref{sub:Connection to GHR estimator}, $1-\Delta\hat{\gamma}_{1}$
is the estimator of the largest nontrivial eigenvalue of the transition
operator. When $\Delta\hat{\gamma}_{1}\geq1$, the estimated transition
operator is negative definite on $V_{J}^{0}$, thus the spectral approach
will not provide a reliable output. Proposition \ref{prop:Properties of the eigenfunction u}
and inequality (\ref{eq:eigenpair error}) ensure that $\Delta\hat{\gamma}_{1}<1$
with high probability, both in high and low-frequency regimes.

\subsection{High-frequency convergence rate\label{sub:High-frequency-convergence-rate}}

The estimation of volatility at point $x$ is possible only when the
process spends enough time around $x.$
\begin{defn}
\label{def:Occupation density}Set $T>0$. Define the occupation density
\begin{equation}
\mu_{T}=\frac{L_{T}}{T\sigma^{2}},\label{eq:occupation density as normalized local time}
\end{equation}
where $L_{T}$ is the semimartingale local time of the path $(X_{t}:0\leq t\leq T)$.
\end{defn}
For any bounded Borel measurable function $f,$ the following occupation
formula holds: 
\begin{equation}
\frac{1}{T}\int_{0}^{T}f(X_{s})ds=\int_{0}^{1}f(x)\mu_{T}(x)dx.\label{eq:Occupation formula}
\end{equation}

In order to obtain the global rates of convergence, we must assume
that the occupation density of the observed path is bounded from below.
Therefore, for a given level $v$, we study the risk of the estimator
conditioned to the event 
\[
\L=\left\{ \inf_{x\in[0,1]}\mu_{T}(x)\geq v\right\} .
\]

\begin{thm}
\label{thm:High Frequency Error}Grant Assumptions \ref{assu: Set Theta}
and \ref{assu: Initial condition}. Fix $T>0$, $0<a<b<1$ and $v>0.$
Choose $J\sim\Delta^{-1/3}$. For every $\epsilon\in(0,1)$ and $\Delta>0$
sufficiently small, there exists an event $\RG_{\epsilon}$, of probability
larger than $1-\epsilon$, and a positive constant $C_{\epsilon}$,
such that 
\[
\sup_{(\sigma,b)\in\Theta(d,D)}\E\left[\I_{\RG_{\epsilon}\cap\L}\cdot\|\hat{\sigma}_{S}^{2}-\sigma^{2}\|_{L^{1}([a,b])}\right]\leq C_{\epsilon}\Delta^{\frac{1}{3}}.
\]

\end{thm}
\citet[Proposition 2]{Hoffmann:2001} shows that the rate $\Delta^{1/3}$
is optimal in the minimax sense even in the class of diffusions with
Lipschitz volatility. To prove Theorem \ref{thm:High Frequency Error},
we compare $\hat{\sigma}_{S}^{2}$ with the benchmark Florens-Zmirou
estimator, see Section \ref{sub:Connection to FZ estimator}. While
the consistency of the spectral estimator can be obtained using the
well known path properties of diffusion processes, the proof of the
exact convergence rate is rather demanding. As explained in Section
\ref{sub:Sketch of the high freq proof}, it is necessary to show
the regularity properties of the estimated eigenfunction $\hat{u}_{1},$
which requires rather sophisticated arguments from the perturbation
theory of differential operators with non-smooth coefficients.

\subsection{Low-frequency convergence rate\label{sub:Low frequency convergence rate}}

In the low-frequency regime, we need to threshold the estimator in
order to ensure integrability and stability against large stochastic
errors. As expected, $\hat{\sigma}_{S}^{2}$ achieves the same mean
$L^{2}$ rate as the original Gobet-Hoffmann and Reiß estimator. Furthermore,
for $\sigma\in H^{1},$ this rate is minimax optimal, which can be
obtained by the same proof as \citep[Theorem 2.5]{GobetHoffmannReiss:2004}.
\begin{thm}
\label{thm:Low Frequency Error}Grant Assumptions \ref{assu: Set Theta}
and \ref{assu: Initial condition}. Fix $\Delta>0$ and $0<a<b<1$.
Choosing $J\sim N^{\frac{1}{5}}$, it holds
\[
\sup_{(\sigma,b)\in\Theta(d,D)}\E\left[\big\|\hat{\sigma}_{S}^{2}\wedge D-\sigma^{2}\big\|_{L^{2}([a,b])}^{2}\right]^{\frac{1}{2}}\lesssim N^{-\frac{1}{5}}.
\]

\end{thm}
The general idea of the proof is the same as in \citet{GobetHoffmannReiss:2004}
or \citep{ChorowskiTrabs:2015}. We use the mixing property of the
process $X$ to control the approximation error of the stationary
measure $\mu$ by the empirical measure $\hat{\mu}_{N}$, see Corollary
\ref{cor:Uniform bounds on the estimator of invariant measure}. Then,
as discussed in Section \ref{sub:Connection to GHR estimator}, we
bound the estimation error of $(\kappa_{1},u_{1})$ - the first nontrivial
eigenpair of the transition operator $P_{\Delta}$, obtaining
\[
|\hat{\kappa}_{1}-\kappa_{1}|+\|\hat{u}_{1}-u_{1}\|_{H^{1}}=O_{\mathbb{P}}(N^{-1/5}).
\]
Finally, we bound the plug-in error of the spectral estimator $\hat{\sigma}_{S}$.
A tenuous point is in that the estimator $\hat{u}_{1}$ converges
to the eigenfunction $u_{1}$ in the sense of mean $H^{1}$ norm only,
hence we can not postulate a uniform positive lower bound on $\inf_{x\in[a,b]}\hat{u}_{1}'(x)$.
Following \citet{ChorowskiTrabs:2015}, we are able to overcome this
difficulty by applying the threshold $\hat{\sigma}_{S}^{2}\wedge D$.

\section{Discussion\label{sec:Discussion}}

\subsection{\label{sub:Connection to GHR estimator}Connection to the GHR low-frequency
estimator}

In this section, we explain the relation between the defined estimator
$\hat{\sigma}_{S}$ above and the original spectral estimator introduced
in \citep[Section 3.2]{GobetHoffmannReiss:2004}. First, let us review
the construction of the GHR estimator.
\begin{defn}
\label{def:p_hat}As in \citet[Eq. 3.8]{GobetHoffmannReiss:2004}
for $u,v\in V_{J}$ let 
\[
\hat{p}(u,v)=\frac{1}{2N}\sum_{n=0}^{N-1}\big(u(X_{n\Delta})v(X_{(n+1)\Delta})+v(X_{n\Delta})u(X_{(n+1)\Delta})\big).
\]

\end{defn}
A crucial observation is that, due to the appropriate weighting of
the empirical measure, $\hat{p}$ becomes a linear combination of
$\hat{l}$ and $\hat{g}$. Indeed, using the summation by parts formula,
we obtain 
\begin{equation}
\hat{l}=\frac{1}{\Delta}(\hat{g}-\hat{p}).\label{eq:form l as a linear combination of g and p}
\end{equation}
Hence, for $(\hat{\gamma}_{i},\hat{u}_{i})$- any solution of the
Eigenproblem \ref{EigenProb:  l_hat g_hat}, we have 
\begin{equation}
\hat{p}(\hat{u}_{i},v)=(1-\Delta\hat{\gamma}_{i})\hat{g}(\hat{u}_{i},v)\text{ for every }v\in V_{J}.\label{eq:Eigenproblem for p_hat}
\end{equation}
Denote 
\begin{equation}
\hat{\kappa}_{i}=(1-\Delta\hat{\gamma}_{i}).\label{eq:Definition of kappa_hat}
\end{equation}
We conclude that the eigenpair $(\hat{\kappa}_{1},\hat{u}_{1})$ is
equal to the estimator of the eigenpair of the transition operator
which is defined in \citep[Eq 3.11]{GobetHoffmannReiss:2004}. Taking
into account that functions $(\psi_{j})$ are not orthonormal, following
\citep[Eq. 3.12 and Eq. 3.7]{GobetHoffmannReiss:2004}, we define
the GHR estimator as: 
\begin{defn}
\label{def:GHR estimator} 
\[
\hat{\sigma}_{GHR}^{2}(x)=\frac{2\hat{\zeta}_{1}\int_{0}^{x}\hat{u}_{1}(y)\hat{\mu}_{N}(dy)}{\hat{u}_{1}'(x)\hat{\mu}(x)},
\]
where 
\[
\hat{\mu}=\sum_{j=0}^{J}\hat{\mu}_{j}\psi_{j}\,\text{with}\,(\hat{\mu}_{j})_{j}=\Big(\Big[\int_{0}^{1}\psi_{i}(y)\psi_{j}(y)dy\Big]_{i,j}\Big)^{-1}\Big(\int_{0}^{1}\psi_{i}(x)\hat{\mu}_{N}(dx)\Big)_{i},
\]
is an estimator of the stationary density. 
\end{defn}
Note that estimator $\hat{\sigma}_{S}$ can be seen as a local average
of $\hat{\sigma}_{GHR}^{2}.$ Indeed, since $\I_{j}=\psi_{j}'$, integrating
by parts gives us 
\begin{equation}
\hat{\sigma}_{S,j}^{2}=\frac{2\hat{\zeta}_{1}\int_{0}^{1}\psi'_{j}(x)\Big(\int_{0}^{x}\hat{u}_{1}(y)\hat{\mu}_{N}(dy)\Big)dx}{\int_{\frac{j-1}{J}}^{\frac{j}{J}}\hat{u}_{1}'(x)\hat{\mu}_{N}(dx)}=\frac{\int_{\frac{j-1}{J}}^{\frac{j}{J}}\hat{\sigma}_{GHR}^{2}(x)\hat{u}_{1}'(x)\hat{\mu}(x)dx}{\int_{\frac{j-1}{J}}^{\frac{j}{J}}\hat{u}_{1}'(x)\hat{\mu}_{N}(dx)}.\label{eq:relation to GHR estimator}
\end{equation}
Since we focus on volatility functions in $H^{1},$ the above averaging
has no effect on the low-frequency convergence rate. On the other
hand, there are multiple reasons why it is beneficial for optimality
in the high-frequency regime. Firstly, since $\hat{u}_{1}'$ is constant
on every interval $\big[\frac{j-1}{J},\frac{j}{J}\big]$, after averaging
we do not have to estimate the density of the occupation measure (which
is not regular in the high-frequency setting), but the occupation
measure of the intervals $\big[\frac{j-1}{J},\frac{j}{J}\big]$. Furthermore,
averaging reduces the variance of the estimator, which can be clearly
seen in Figure \ref{fig:Vol Plots}. The intuitive explanation of
this phenomenon is that while the original estimator $\hat{\sigma}_{GHR}^{2}$
inherits the rough behaviour of the occupation density (via the inverse
of the derivative of the eigenfunction $u_{1}$ which has the same
smoothness as the design density) this irregularity is removed by
multiplication with $\hat{u}_{1}'\hat{\mu}.$

\subsection{\label{sub:Connection to FZ estimator}Connection to the Florens-Zmirou
estimator}

The general idea of the proof of the high-frequency convergence rate
is to compare estimator $\hat{\sigma}_{S}$ with the minimax optimal
(see \citep[Proposition 2]{Hoffmann:2001}) high-frequency estimator
introduced in \citet{Florens-Zmirou:1993}. In this section, we recall
the definition of the Florens-Zmirou estimator and discuss its relation
to $\hat{\sigma}_{S}.$
\begin{defn}
\label{def:Florens-Zmirou estimator}Define the time-symmetric version
of the well known Nadaraya-Watson type estimator of the squared volatility
coefficient, introduced in \citet{Florens-Zmirou:1993}, by
\begin{align*}
\hat{\sigma}_{FZ,j}^{2} & =\frac{\sum_{n=0}^{N-1}(\I_{j}(X_{n\Delta})+\I_{j}(X_{(n+1)\Delta}))(X_{(n+1)\Delta}-X_{n\Delta})^{2}}{\Delta\sum_{n=0}^{N-1}(\I_{j}(X_{n\Delta})+\I_{j}(X_{(n+1)\Delta}))},\\
\hat{\sigma}_{FZ}^{2}(x) & =\sum_{j=1}^{J}\hat{\sigma}_{FZ,j}^{2}\I_{j}(x).
\end{align*}
Note that the underweighting of the first and last observation in
the denominator of $\hat{\sigma}_{FZ,j}^{2}$ appears naturally as
an artifact of the time symmetry.\end{defn}
\begin{rem}
\label{rem:Time Symmetric FZ}We call $\hat{\sigma}_{FZ}^{2}$ a time-symmetrized
version of the Florens-Zmirou estimator, since it is an average of
the standard Florens-Zmirou estimators (c.f. \citep[Eq. (1.1)]{Florens-Zmirou:1993})
constructed for the process $(X_{t},0\leq t\leq T)$ and the time
reversed process $Y_{t}=X_{T-t}$. Indeed, let
\begin{equation}
\hat{\sigma}_{j}^{2}(X_{0},X_{\Delta},...,X_{N\Delta})=\frac{\sum_{n=0}^{N-1}\I_{j}(X_{n\Delta})(X_{(n+1)\Delta}-X_{n\Delta})^{2}}{\Delta\big(\frac{1}{2}\I_{j}(X_{0})+\sum_{n=1}^{N-1}\I_{j}(X_{n\Delta})+\frac{1}{2}\I_{j}(X_{N\Delta})\big)}.\label{eq:forward Florens-Zmirou}
\end{equation}
Then
\[
\hat{\sigma}_{FZ,j}^{2}=\frac{\hat{\sigma}_{j}^{2}(X_{0},X_{\Delta},...,X_{N\Delta})+\hat{\sigma}_{j}^{2}(Y_{0},Y_{\Delta},...,Y_{N\Delta})}{2}.
\]
Since stationary scalar diffusions are reversible, under the Assumption
\ref{assu: Initial condition}, the process $(Y_{t},0\leq t\leq T)$
is identical in law to $(X_{t},0\leq t\leq T)$. Hence, the statistical
properties of estimator $\hat{\sigma}_{FZ}^{2}$ are the same as those
of the classical Florens-Zmirou estimator. 
\end{rem}
Recall that $(\hat{\gamma}_{1},\hat{u}_{1})$ is an eigenpair of the
Eigenproblem \ref{EigenProb:  l_hat g_hat}. From Definition \ref{def:definition of spectral estimator}
of the spectral estimator, it follows that
\begin{equation}
\hat{\sigma}_{S,j}^{2}=\frac{-\hat{\zeta}_{1}}{\hat{\gamma}_{1}}\frac{2\hat{l}(\hat{u}_{1},\psi_{j})}{\hat{u}_{1,j}\int_{\frac{j-1}{J}}^{\frac{j}{J}}\hat{\mu}_{N}(dx)}.\label{eq:Spectral estimator in high frequency form}
\end{equation}
A similar representation formula can be established for the time symmetric
Florens-Zmirou estimator $\hat{\sigma}_{FZ}^{2}$. 
\begin{defn}
\label{def:Form f_hat}Define a bilinear form $\hat{f}:V_{J}\times V_{J}\to\R$
by 
\[
\hat{f}(u,v)=\frac{1}{2}\int_{0}^{1}u'(x)v'(x)\hat{\sigma}_{FZ}^{2}(x)\hat{\mu}_{N}(dx).
\]

\end{defn}
Consider vector $(v_{j})_{j=1,...,J}$ such that $v_{j}\neq0$ for
every $j=1,...,J$ and the associated function $v\in V_{J}^{0}$.
We have 
\begin{equation}
\hat{\sigma}_{FZ,j}^{2}=\frac{2\hat{f(}v,\psi_{j})}{v_{j}\int_{\frac{j-1}{J}}^{\frac{j}{J}}\hat{\mu}_{N}(dx)}.\label{eq:Florens-Zmirou estimator in spectral form}
\end{equation}
As will be thoroughly explained in Section \ref{sub:Sketch of the high freq proof},
when $\Delta\to0$, the eigenvalue ratio $-\hat{\zeta}_{1}/\hat{\gamma}_{1}$
in (\ref{eq:Spectral estimator in high frequency form}) tends to
$1.$ Consequently, the difference between estimators $\hat{\sigma}_{S}^{2}$
and $\hat{\sigma}_{FZ}^{2}$ is controlled by
\begin{equation}
\frac{2|\hat{l}(\hat{u}_{1},\psi_{j})-\hat{f(}\hat{u}_{1},\psi_{j})|}{\hat{u}_{1,j}\int_{\frac{j-1}{J}}^{\frac{j}{J}}\hat{\mu}_{N}(dx)}.\label{eq:difference of sigma_S and sigma_FZ}
\end{equation}
The main observation is that in the high-frequency analysis, we do
not have to control the estimation error of the derivative $\hat{u}_{1}'$.
Indeed, to bound (\ref{eq:difference of sigma_S and sigma_FZ}), we
need only to show a uniform lower bound for $\hat{u}_{1,j}$ and an
upper bound for the difference $|\hat{l}(\hat{u}_{1},\psi_{j})-\hat{f(}\hat{u}_{1},\psi_{j})|.$
Unfortunately, $|\hat{l}(v,\psi_{j})-\hat{f(}v,\psi_{j})|$ is not
small enough for any bounded function $v$. To achieve the required
upper bound for the estimated eigenfunction, we need to first obtain
some regularity properties of $\hat{u}_{1},$ which is the most difficult
part of the high-frequency analysis.

\subsection{\label{sub:Numerical Example}A Numerical Example}

In this section, we present the numerical results for the volatility
estimation across different observation time scales. We compare three
estimation methods: the time symmetric Florens-Zmirou estimator $\hat{\sigma}_{FZ}^{2}$
(see Definition \ref{def:Florens-Zmirou estimator}), the spectral
estimator $\hat{\sigma}_{GHR}^{2}$ (see Definition \ref{def:GHR estimator},
c.f. \citet[Section 3.2]{GobetHoffmannReiss:2004}) with approximation
space $V_{J}$ of linear splines with equidistant knots, and finally,
the locally averaged spectral estimator $\hat{\sigma}_{S}^{2}$. We
apply an oracle choice of the projection level $J$, minimizing the
risk.

We compare the locally averaged spectral estimator $\hat{\sigma}_{S}^{2}$
with benchmark estimators $\hat{\sigma}_{FZ}^{2}$ and $\hat{\sigma}_{GHR}^{2}$
in both high and low-frequency regimes. Following \citet[Section 5]{ChorowskiTrabs:2015}
we consider diffusion process $X$ with mean reverting drift $b(x)=0.2-0.4x$,
quadratic squared volatility function $\sigma^{2}(x)=0.4-(x-0.5)^{2},$
and two reflecting barriers at 0 and 1. This choice of diffusion coefficients
is supposed to minimize the reflection effect alongside with some
variability in the volatility function. Nevertheless, the depicted
behaviour is typical for other diffusion processes. The sample paths
were generated using the Euler-Maruyama scheme with time step size
$\Delta/100\wedge0.001$ with reflection after each step. All simulated
paths were conditioned to have an occupation time density greater
than $v=0.2$. Table \ref{tab:HF mean L1 error} presents the oracle
mean $L^{1}([0.1,0.9])$ estimation error of $\sigma^{2}$, obtained
by a Monte Carlo simulation with 1000 iterations, in high ($T=5,\Delta\to0$)
and low $(\Delta=0.25,T\to\infty)$ frequency regimes, respectively.
The estimated volatility functions for 20 independent paths are depicted
in Figure \ref{fig:Vol Plots}. 
\begin{figure}[t]
\noindent \centering{}\includegraphics[width=0.95\textwidth]{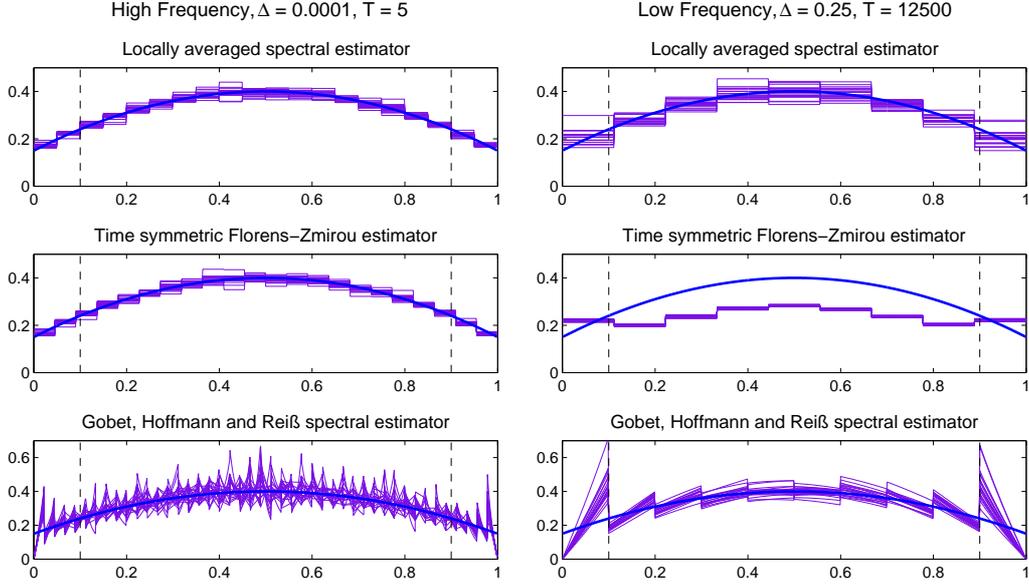}\protect\caption{\label{fig:Vol Plots}Estimated volatility functions for 20 independent
trajectories.}
\end{figure}

In the case of high-frequency observations, $\hat{\sigma}_{S}^{2}$
performs similarly to the benchmark estimator $\hat{\sigma}_{FZ}^{2}$.
Relative to $\|\sigma^{2}\|_{L^{1}([0.1,0.9])}\approx0.28$, the error
decreases from approximately 6\% for $\Delta=10^{-3}$ to 3\% for
$\Delta=10^{-4}.$ The estimation error of spectral estimator $\hat{\sigma}_{GHR}^{2}$
is almost twice as large, although the quality of the estimation improves
when $\Delta$ decreases. It is important to note that the oracle
values of space parameter $J$ for $\hat{\sigma}_{GHR}^{2}$ are much
bigger than those for other estimation methods. When $\Delta$ is
small, the eigenfunctions inherit the regularity of the local time;
the increase in dimension compensates for the projection error. Due
to local averaging, this irregularity problem does not appear for
$\hat{\sigma}_{S}^{2}$, compare with Figure \ref{fig:Vol Plots},
where estimator $\hat{\sigma}_{GHR}^{2}$ oscillates heavily. Furthermore,
there is no visible boundary effect, suggesting that the error rate
of the spectral estimator does not deteriorate outside the fixed interval
$[0.1,0.9]$.

In the low-frequency regime, $\hat{\sigma}_{S}^{2}$ performs slightly
better than the original spectral estimator $\hat{\sigma}_{GHR}^{2}$.
The boundary problem is visible, especially for $\hat{\sigma}_{GHR}^{2}$.
The relative error decreases from 12\% for $T$=1000 to 5\% for $T$=30
000. The Florens-Zmirou estimator $\hat{\sigma}_{FZ}^{2}$ underestimates
the volatility and commits a relative error of 30\% . This is expected
and due mostly to the boundary reflection, which, for low-frequency
observations, is not negligible in the interior of the state space.
As found by unreported simulations, in the case of low-frequency observations,
the locally averaged spectral estimator $\hat{\sigma}_{S}^{2}$ will
outperform the Florens-Zmirou estimator in the case of a highly varying
volatility function $\sigma^{2}$, even when the sampling frequency
is big enough to ignore the reflection effect.

\begin{table}[t]
\begin{centering}
\begin{tabular}{|c|c|c|c|c|c|c|}
\multicolumn{1}{c}{} & \multicolumn{6}{c}{High-Frequency Regime: $T=5$}\tabularnewline
\cline{2-7} 
\multicolumn{1}{c|}{} & {\tiny{}$\Delta=0.001$}  & {\tiny{}$\Delta=0.00075$}  & {\tiny{}$\Delta=0.0005$}  & {\tiny{}$\Delta=0.00035$}  & {\tiny{}$\Delta=0.0002$}  & {\tiny{}$\Delta=0.0001$}\tabularnewline
\hline 
{\tiny{}$\hat{\sigma}_{GHR}^{2}$}  & {\tiny{}0.0388$_{(18)}$}  & {\tiny{}0.0353$_{(23)}$}  & {\tiny{}0.0322$_{(24)}$}  & {\tiny{}0.0292$_{(32)}$}  & {\tiny{}0.0258$_{(36)}$}  & {\tiny{}0.0220$_{(49)}$}\tabularnewline
\hline 
{\tiny{}$\hat{\sigma}_{S}^{2}$}  & {\tiny{}0.0195$_{(9)}$}  & {\tiny{}0.0174$_{(10)}$}  & {\tiny{}0.0149$_{(10)}$}  & {\tiny{}0.0131$_{(12)}$}  & {\tiny{}0.0108$_{(13)}$}  & {\tiny{}0.0088$_{(18)}$}\tabularnewline
\hline 
{\tiny{}$\hat{\sigma}_{FZ}^{2}$}  & {\tiny{}0.0169$_{(10)}$}  & {\tiny{}0.0153$_{(11)}$}  & {\tiny{}0.0133$_{(12)}$}  & {\tiny{}0.0119$_{(12)}$}  & {\tiny{}0.0100$_{(13)}$}  & {\tiny{}0.0080$_{(20)}$}\tabularnewline
\hline 
\multicolumn{1}{c}{} & \multicolumn{1}{c}{} & \multicolumn{1}{c}{} & \multicolumn{1}{c}{} & \multicolumn{1}{c}{} & \multicolumn{1}{c}{} & \multicolumn{1}{c}{}\tabularnewline
\multicolumn{1}{c}{} & \multicolumn{6}{c}{Low-Frequency Regime:$\Delta=0.25$}\tabularnewline
\cline{2-7} 
\multicolumn{1}{c|}{} & {\tiny{}T=1k}  & {\tiny{}T=3k}  & {\tiny{}T=7k}  & {\tiny{}T=10k}  & {\tiny{}T=15k}  & {\tiny{}T=20k}\tabularnewline
\hline 
{\tiny{}$\hat{\sigma}_{GHR}^{2}$}  & {\tiny{}0.0386$_{(5)}$}  & {\tiny{}0.0333$_{(6)}$}  & {\tiny{}0.0256$_{(11)}$}  & {\tiny{}0.0226$_{(11)}$}  & {\tiny{}0.0198$_{(11)}$}  & {\tiny{}0.0178$_{(11)}$}\tabularnewline
\hline 
{\tiny{}$\hat{\sigma}_{S}^{2}$}  & {\tiny{}0.0310$_{(4)}$}  & {\tiny{}0.0245$_{(6)}$}  & {\tiny{}0.0200$_{(7)}$}  & {\tiny{}0.0182$_{(8)}$}  & {\tiny{}0.0166$_{(8)}$}  & {\tiny{}0.0155$_{(9)}$}\tabularnewline
\hline 
{\tiny{}$\hat{\sigma}_{FZ}^{2}$}  & {\tiny{}0.0821$_{(5)}$}  & {\tiny{}0.0823$_{(5)}$}  & {\tiny{}0.0823$_{(5)}$}  & {\tiny{}0.0822$_{(5)}$}  & {\tiny{}0.0823$_{(5)}$}  & {\tiny{}0.0824$_{(5)}$}\tabularnewline
\hline 
\end{tabular}
\par\end{centering}

\centering{}\protect\caption{\label{tab:HF mean L1 error}Monte Carlo estimation errors in high
and low-frequency regimes. The value of parameter $J$ is given in
the subscript.}
\end{table}

\subsection{Extensions and limitations\label{sub:Limitations and conclusions}}

\paragraph{\textit{Stationarity of process $X$.}}

In the high-frequency analysis the stationarity assumption ensures
that process $X$ is time reversible. General initial distributions
could be considered, but in order to preserve the performance of the
estimation for the time reversed process, the coefficients of the
backward process must belong to the nonparametric family $\Theta.$

Due to the spectral gap of the generator, process $X$ is geometrically
ergodic. In particular, as $t\to\infty$, the one dimensional distributions
of $X_{t}$ converge exponentially fast to the invariant measure $\mu$.
It follows that, in the low-frequency regime, the assumption of stationarity
can be made without loss of generality for asymptotic results.

\paragraph{\textit{Estimation at the boundaries.}}

In the high-frequency regime, we prove the error bound in the interior
of the state space. Restriction to the interval $(a,b)$ allows us
to obtain uniform lower bounds on the derivative of eigenfunction
$\hat{u}_{1}$, which, due to boundary conditions, are not valid in
the entire state space. This restriction could be omitted by obtaining
uniform bounds on the ratio of derivatives $\hat{u}_{1,j\pm1}/\hat{u}_{1,j}$.
Unfortunately, since our proof relies on a posteriori error bounds
on solutions for perturbed eigenvalue problems, we do not have the
sufficient tools to control the pointwise relative error of the eigenfunctions.
Nevertheless, the numerical results suggest that the spectral estimation
procedure also behaves well at the boundaries of the state space.

In the low-frequency regime, the spectral estimator is unstable at
the boundary due to Neumann boundary conditions for the eigenfunctions
of the infinitesimal generator. Refer to \citep[Section 3.3.8]{GobetHoffmannReiss:2004}
for a discussion of the boundary problem.

\paragraph{\textit{Boundary reflection.}}

Following previous works on the spectral estimation in the low frequency
setting, e.g. \citep{GobetHoffmannReiss:2004,NicklSohl:2015,ChorowskiTrabs:2015},
we consider an Itô diffusion model on the state space $[0,1]$ with
instantaneous reflection at the boundaries. The assumption of a compact
state space makes the construction of the estimator easier and facilitates
error analysis in the low-frequency setting, c.f. \citet{Reiss:2006}.
We point out, here, that the reflection assumption is not restrictive
in the high-frequency setting. Indeed, consider diffusion $X$ defined
on the entire real line with drift $b$ and volatility $\sigma$.
Let
\[
A(t)=\int_{0}^{t}\I_{[0,1]}(X_{s})ds
\]
be the occupation time of interval $[0,1]$. Assume that $\lim_{t\to\infty}A(t)=\infty$
and define the right-continuous inverse
\[
C(t)=\inf\{s>0|A(s)>t\}.
\]
Process $Y_{t}=X_{C(t)}$ follows the law of a reflected diffusion
on $[0,1]$ with drift $b$ and volatility $\sigma.$ Assume now the
given observations $X_{0},X_{\Delta},...,X_{N\Delta}$. The sub-sequence
of the values that lie in $[0,1]$ forms a chain of observation of
$Y$. The sampling frequency is random (and depends on the path),
but when $\Delta$ shrinks, it becomes close to equidistant. The difficulty
in handling irregularities at the boundaries is similar to these found
when considering the reflection effect. Unfortunately, while this
reduction can be used under the assumption that $\Delta$ is small,
it can't be applied in the low frequency setting, hence it is not
practical in the context of scale invariant estimation.

\paragraph{\emph{Linear spline basis.}}

The use of the linear spline basis is very convenient, as functions
$\psi_{j}$ appear naturally after applying integration by parts to
the locally averaged GHR estimator, see (\ref{eq:relation to GHR estimator}).
Nevertheless, unreported simulations suggest that the spectral estimation
method performs as well with other bases. The Fourier cosines basis
in $[0,1]$ is especially efficient, consisting of the eigenfunctions
of the reflected Brownian motion process.

\paragraph*{Adaptivity.}

An important decision in the spectral estimation is the choice of
the basis dimension $J.$ The general problem is twofold: dimension
$J$ should adapt to the smoothness of the coefficients and simultaneously
to the observation frequency. In \citep{ChorowskiTrabs:2015} the
authors applied Lepski\textquoteright s method to construct a data-driven
version of the GHR estimator that adapts to the smoothness of the
volatility. In the case of the low frequency data, the same selection
rule can be applied for the universal estimator $\hat{\sigma}_{S}$.
The precise construction of a method that will adapt to the observation
frequency remains open.

\medskip{}

The numerical study shows that the proposed estimator $\hat{\sigma}_{S}$
smoothly interpolates between the high and low-frequency estimators.
The optimal convergence rates in both frequency regimes leave out
the question of the paradigm to use when one has to consider data.
The different convergence rates in high and low frequency regimes
raise the question of bivariate asymptotics with respect to both $\Delta$
and $T.$ Nevertheless, because of the structural differences of the
high and low-frequency data, we believe that such an analysis would
be particularly challenging.

\section{\label{sec:Proof in HF}High-frequency analysis}

We will write $f\lesssim g$ (resp. $g\gtrsim f$) when $f\leq C\cdot g$
for some universal constant $C>0$. $f\sim g$ is equivalent to $f\lesssim g$
and $g\lesssim f$.

The proof of Theorem \ref{thm:High Frequency Error} is presented
in Section \ref{sub:Proof of HF error bound} and is accomplished
in several steps. In Section \ref{sub:Proof of Florens zmirou error}
we prove the convergence rate of the time-symmetric Florens-Zmirou
estimator. Section \ref{sub:Properties of the eigenfunction u} is
devoted to the proof of Proposition \ref{prop:Properties of the eigenfunction u}
- the uniform bounds on the estimated eigenpair $(\hat{\gamma}_{1},\hat{u}_{1}).$
In Section \ref{sub:Mean crossing bounds} we prove some technical
results on the crossing intensity of the diffusion processes.

\subsection{Preliminaries}

From now on we take the Assumptions \ref{assu: Set Theta} and \ref{assu: Initial condition}
as granted. Fix $0<a<b<1$ and the level $v>0$. For simplicity, set
$T=1$. Let $J\sim\Delta^{-1/3}$.

Sobolev regularity of the volatility implies that it is $1/2-$Hölder
continuous. Indeed, by the Cauchy-Schwarz inequality it holds
\begin{equation}
\sup_{x,y\in[0,1]}\frac{|\sigma(x)-\sigma(y)|}{|x-y|^{1/2}}=\sup_{x,y\in[0,1]}\frac{\left|\int_{x}^{y}\sigma'(z)dz\right|}{|x-y|^{1/2}}\leq\|\sigma\|_{H^{1}}.\label{eq:Holder regularity of sigma}
\end{equation}
 Recall Definition \ref{def:Occupation density} of the occupation
density $\mu_{T}$. Formula (\ref{eq:occupation density as normalized local time}),
together with (\ref{eq:Holder regularity of sigma}), imply that $\mu_{T}$
inherits the regularity properties of the local time. In particular
\begin{thm}
\label{thm:properties of the occupation density}The function $\mu_{1}$
is almost surely Hölder continuous of order $\alpha$ for every $\alpha<1/2$.
Moreover, for every $p\geq1$, we have
\begin{eqnarray}
\sup_{(\sigma,b)\in\Theta}\E\Big[\sup_{x\in[0,1]}\mu_{1}^{p}(x)\Big] & < & \infty.\label{eq:bound on moments of mu_T}\\
\sup_{(\sigma,b)\in\Theta}\E\big[|\mu_{1}(x)-\mu_{1}(y)|^{2p}\big] & \leq & C_{p}|x-y|^{p}.\label{eq:bound on incremens of mu_T}
\end{eqnarray}
\end{thm}
\begin{proof}
Since $\sigma$ is uniformly bounded and $1/2-$Hölder continuous,
the claim of the theorem can be deduced from the well known properties
of the family of the local times $(L_{t},t\geq0)$ of the semimartingale
$X,$ see the proof of \citep[Chapter VI, Theorem 1.7]{RevuzYor:1999}
and the subsequent remark.\end{proof}
\begin{defn}
\label{def:Modulus of continuity}Denote by $\omega$ the modulus
of continuity of the path $(X_{t},0\leq t\leq1),$ i.e. 
\[
\omega(\delta)=\sup_{\begin{subarray}{c}
0\leq s,t\leq1,\,|t-s|\leq\delta\end{subarray}}|X_{t}-X_{s}|.
\]

\end{defn}
Because of the ellipticity assumption $\sigma>0$ the path $(X_{t},0\leq t\leq1)$
shares the properties of Brownian paths. In particular, we can apply
the Brownian upper bounds (see \citet{FischerNappo:2010}) on the
moments of $\omega$: 
\begin{thm}
\label{thm:Moments of the modulus of continuity}For every $p\geq1$
there exists a constant $C_{p}>0$ such that 
\begin{equation}
\sup_{(\sigma,b)\in\Theta}\E[\omega^{p}(\Delta)]\leq C_{p}\Delta^{p/2}ln^{p}\left(\Delta^{-1}\right).\label{eq:Moments of the modulus of continuity}
\end{equation}

\end{thm}
The proof of Theorem \ref{thm:Moments of the modulus of continuity}
is postponed to Section \ref{sec:Construction and properties of X}.
Using (\ref{eq:Moments of the modulus of continuity}) we can show
that on $\mathcal{L}_{v}$ the occupation measure $\hat{\mu}_{N}$
is spread uniformly on $[0,1]$ with high probability.
\begin{lem}
\label{lem:mass of empirical measure on Ik}Let
\begin{equation}
\RG_{1}=\L\cap\{\mc\|\mu_{1}\|_{\infty}\leq\Delta^{5/11}v\}.\label{eq:Def of event R1}
\end{equation}
For $\Delta$ sufficiently small we have 
\[
\P(\L\setminus\RG_{1})\lesssim\Delta^{2/3}.
\]
Furthermore, on the event $\RG_{1}$, for every $1\le j\leq J$, we
have 
\[
v\lesssim J\int_{\frac{j-1}{J}}^{\frac{j}{J}}\hat{\mu}_{N}(dx)\lesssim\|\mu_{1}\|_{\infty}.
\]

\end{lem}
The proof of Lemma \ref{lem:mass of empirical measure on Ik} is postponed
to Section \ref{sub:Proof of Florens zmirou error}. 

As mentioned in Section \ref{sub:High-frequency-convergence-rate}
we want to compare the spectral estimator $\hat{\sigma}_{S}^{2}$
with the benchmark high-frequency estimator $\hat{\sigma}_{FZ}^{2}$.
Before that, we have to prove a uniform upper bound on the mean $L^{2}$
error of the time symmetric Florens-Zmirou estimator. The result below
is a generalization of \citep[Proposition 2]{Hoffmann:2001}, where
the same rate was obtained under the assumptions of smooth drift and
Lipschitz volatility. As proved in \citep[Proposition 2]{Hoffmann:2001}
the rate $\Delta^{1/3}$ is optimal in the minimax sense even on the
class of diffusions with Lipschitz volatility.
\begin{thm}
\label{thm:FlorensZmirou L2 error}Grant Assumptions \ref{assu: Set Theta}
and \ref{assu: Initial condition}. Fix $T>0$ and choose $J\sim\Delta^{-\frac{1}{3}}$.
We have
\begin{equation}
\sup_{(\sigma,b)\in\Theta(d,D)}\E\big[\I_{\RG_{1}}\cdot\|\hat{\sigma}_{FZ}^{2}-\sigma^{2}\|_{L^{2}[1/J,1-1/J]}^{2}\big]^{\frac{1}{2}}\lesssim\Delta^{1/3}.\label{eq:FZ interior bound}
\end{equation}
Because of the reflection, the rate deteriorates at the boundary.
For $x\in[0,1/J]\cup[1-1/J,1]$ 
\begin{equation}
\sup_{(\sigma,b)\in\Theta(d,D)}\E\big[\I_{\RG_{1}}\cdot|\hat{\sigma}_{FZ}^{2}(x)-\sigma^{2}(x)|^{2}\big]^{\frac{1}{2}}\lesssim\Delta^{1/33}.\label{eq:FZ boundary bound}
\end{equation}

\end{thm}
The proof of Theorem \ref{thm:FlorensZmirou L2 error} is postponed
to Section \ref{sub:Proof of Florens zmirou error}. The main idea
is the decomposition of the error into a martingale and deterministic
approximation parts as in \citep[Proposition 2]{Hoffmann:2001}. As
expected, under the high-frequency assumption, the reflection has
an effect only at the boundary. Inequalities (\ref{eq:FZ interior bound})
and (\ref{eq:FZ boundary bound}) imply
\[
\sup_{(\sigma,b)\in\Theta(d,D)}\E\big[\I_{\RG_{1}}\cdot\|\hat{\sigma}_{FZ}^{2}(x)-\sigma^{2}(x)\|_{L^{1}[0,1]}\big]\lesssim\Delta^{1/3}.
\]

\subsection{\label{sub:Sketch of the high freq proof}Outline of the proof of
the high-frequency convergence rate}

Since by Theorem \ref{thm:FlorensZmirou L2 error} the estimator $\hat{\sigma}_{FZ}^{2}$
attains the optimal rate $\Delta^{1/3}$, to prove Theorem \ref{thm:High Frequency Error}
it is enough to upper bound the mean $L^{1}[0,1]$ error between $\hat{\sigma}_{FZ}^{2}$
and $\hat{\sigma}_{S}^{2}$. Using representations (\ref{eq:Spectral estimator in high frequency form})
and (\ref{eq:Florens-Zmirou estimator in spectral form}) $\hat{\sigma}_{FZ}^{2}-\hat{\sigma}_{S}^{2}$
can be reduced to the difference of the forms $\hat{f}$ and $\hat{l}$
(c.f. Lemma \ref{lem: sigma_S minus sigma  bounded by forms a and f}).
First, we need however to list the properties of the eigenpair $(\hat{\zeta}_{1},\hat{u}_{1}).$
The proof of the next Proposition is postponed to Section \ref{sub:Properties of the eigenfunction u}. 
\begin{prop}
\label{prop:Properties of the eigenfunction u}Let $0<a<b<1$ be fixed.
For every $\epsilon>0$ there exists an event $\RG_{2}=\RG_{2}(\epsilon)$,
with $\P(\L\setminus\RG_{2})\leq\epsilon$, and a constant $C=C(\epsilon)$,
such that, for $\Delta$ sufficiently small we have 
\begin{equation}
\I_{\RG_{2}}\cdot|\hat{\gamma}_{1}|\lesssim C.\label{eq:Uniform bound on gamma_hat}
\end{equation}
Furthermore, the eigenfunction $\hat{u}_{1}$ can be chosen such that
on $\RG_{2}$ 
\[
\sum_{j=1}^{J}\hat{u}_{1,j}^{2}=J\text{ and }\hat{u}_{1,j}\sim1,\,\text{and}\,\sum_{j=1}^{J}\hat{u}_{1,j}^{2}\I(\hat{u}_{1,j}<0)\lesssim1
\]
hold for any $j=\left\lfloor aJ\right\rfloor -1,...,\left\lceil bJ\right\rceil +1$.\end{prop}
\begin{rem}
The normalization $\sum_{j=1}^{J}\hat{u}_{1,j}^{2}=J$ is natural,
as it is equivalent to $\|\hat{u}_{1}'\|_{L^{2}}=1.$ In short, Proposition
\ref{prop:Properties of the eigenfunction u} states the existence
of uniform bounds on $\hat{u}_{1}'\mid_{[a,b]}$. Because of the Neumann
boundary conditions on the generator, the separation from the boundary
is necessary for the existence of the lower bound.
\end{rem}

\begin{rem}
\label{rem:Spectral estimator without eigenvalue ratio}From the general
inequality
\[
|1+\log(1-x)/x|\leq x,\quad0<x<1/2,
\]
together with the uniform bound (\ref{eq:Uniform bound on gamma_hat})
on the eigenvalue $\hat{\gamma}_{1},$ we deduce that, on the high
probability event $\RG_{2}$, $|1+\hat{\zeta}_{1}/\hat{\gamma}_{1}|\lesssim\Delta$
holds. Consequently, the eigenvalue ratio $-\hat{\zeta}_{1}/\hat{\gamma}_{1}$
in (\ref{eq:Spectral estimator in high frequency form}) is of no
importance in the high-frequency analysis. \end{rem}
\begin{defn}
\label{def:Spectral estimator without eigenvalue ratio}Define
\begin{eqnarray}
\tilde{\sigma}_{S,j}^{2} & = & \frac{2\hat{l(}\hat{u}_{1},\psi_{j})}{\hat{u}_{1,j}\int_{\frac{j-1}{J}}^{\frac{j}{J}}\hat{\mu}_{N}(dx)},\label{eq:Spectral estimator in HF form without eigenvalue}\\
\tilde{\sigma}_{S}^{2}(x) & = & \sum_{j=1}^{J}\tilde{\sigma}_{S,j}^{2}\I_{j}(x).\nonumber 
\end{eqnarray}

\end{defn}
For simplicity, we will refer from now on to $\tilde{\sigma}_{S}$
as to the spectral estimator. Comparing the representations (\ref{eq:Spectral estimator in HF form without eigenvalue})
and (\ref{eq:Florens-Zmirou estimator in spectral form}) we obtain
\[
|\tilde{\sigma}_{S,j}^{2}-\hat{\sigma}_{FZ,j}^{2}|=\frac{2|\hat{l(}\hat{u}_{1},\psi_{j})-\hat{f}(\hat{u}_{1},\psi_{j})|}{\hat{u}_{1,j}\int_{\frac{j-1}{J}}^{\frac{j}{J}}\hat{\mu}_{N}(dx)}.
\]
Since by Proposition \ref{prop:Properties of the eigenfunction u}
the derivative $\hat{u}_{1,j}$ has a uniform lower bound, Lemma \ref{lem:mass of empirical measure on Ik}
implies that to show the convergence rate $\Delta^{1/3}$ we have
to prove that 
\[
|\hat{l(}\hat{u}_{1},\psi_{j})-\hat{f(}\hat{u}_{1},\psi_{j})|=O_{p}(\Delta^{2/3}).
\]
As argued in Proposition \ref{prop:Mean error between f and a on a general function},
for any function $v\in V_{J}$ with bounded derivative, it holds 
\[
|\hat{l(}v,\psi_{j})-\hat{f(}v,\psi_{j})|=O_{p}(\Delta^{1/2}),
\]
which leads to a suboptimal rate $\Delta^{1/6}$. In order to achieve
the optimal rate $\Delta^{1/3}$ we need to use the regularity of
the first nontrivial eigenfunction $\hat{u}_{1}$. By the means of
the Perron-Frobenius theory, in Proposition \ref{prop:Suboptimal rate for spectral and error between u and unit},
we prove that for some high probability event $\RG_{3}$ 
\[
\E\Big[\I_{\RG_{3}}\cdot\Big|\frac{\hat{u}_{1}'(\frac{j}{J}\pm\frac{1}{J})-\hat{u}_{1}'(\frac{j}{J})}{J^{-1/2}}\Big|^{2}\Big]^{\frac{1}{2}}\lesssim1
\]
holds, which can be interpreted as the almost $1/2-$Hölder regularity
of $\hat{u}_{1}'$ (see Remark \ref{rem:Regularity of the eigenfunction}).
This regularity of the eigenfunction allows us to reduce the estimation
error to an approximation problem of the occupation time, see decomposition
(\ref{eq:Proof of HF error bounds}) and Lemma \ref{lem:Error of l and f on the unit vector}.

\subsection{\label{sub:Proof of Florens zmirou error}Proof of Theorem \ref{thm:FlorensZmirou L2 error}}

We begin with the proof of Lemma \ref{lem:mass of empirical measure on Ik}. 
\begin{proof}[Proof of Lemma \ref{lem:mass of empirical measure on Ik}]
Note first that, on the event $\L$, we have
\[
v\leq J\int_{\frac{j-1}{J}}^{\frac{j}{J}}\mu_{1}(dx)\leq\|\mu_{1}\|_{\infty}.
\]
Using the occupation formula (\ref{eq:Occupation formula}), we obtain
that 
\begin{align*}
 & \Big|\frac{1}{N}\sum_{n=0}^{N-1}\I_{j}(X_{n\Delta})-\int_{\frac{j-1}{J}}^{\frac{j}{J}}\mu_{1}(x)dx\Big|\leq\sum_{n=0}^{N-1}\int_{n\Delta}^{(n+1)\Delta}|\I_{j}(X_{n\Delta})-\I_{j}(X_{s})|ds\leq\\
 & \qquad\leq\sum_{n=0}^{N-1}\int_{n\Delta}^{(n+1)\Delta}\Big(\I\big(\big|X_{s}-{\textstyle \frac{j-1}{J}}\big|<\mc\big)ds+\I\big(\big|X_{s}-{\textstyle \frac{j}{J}}\big|<\mc\big)\Big)ds\\
 & \qquad=\int_{\frac{j-1}{J}-\mc}^{\frac{j-1}{J}+\mc}\mu_{1}(x)dx+\int_{\frac{j}{J}-\mc}^{\frac{j}{J}+\mc}\mu_{1}(x)dx\leq4\mc\|\mu_{1}\|_{\infty}.
\end{align*}
Hence, and since $J\sim\Delta^{-1/3}$, on the event $\RG_{1}$ 
\[
\Delta^{\frac{1}{3}}v\lesssim\Delta^{\frac{1}{3}}v-4\mc\|\mu_{1}\|_{\infty}\lesssim\int_{\frac{j-1}{J}}^{\frac{j}{J}}\hat{\mu}_{N}(dx)\lesssim(\Delta^{\frac{1}{3}}+4\mc)\|\mu_{1}\|_{\infty}\lesssim\Delta^{\frac{1}{3}}\|\mu_{1}\|_{\infty},
\]
holds for any $\Delta<1$. Finally, to prove that $\RG_{1}$ is a
high probability event, note that for any $p\geq1$, Theorem \ref{thm:Moments of the modulus of continuity}
together with the inequality (\ref{eq:bound on moments of mu_T})
imply 
\begin{eqnarray*}
\P(\L\setminus\RG_{1}) & \lesssim & \Delta^{-5p/11}\E[\mc^{p}\|\mu_{1}\|_{\infty}^{p}]\\
 & \lesssim & \Delta^{-5p/11}\E[\mc^{2p}]^{1/2}\E[\|\mu_{1}\|_{\infty}^{2p}]^{1/2}\\
 & \lesssim & \Delta^{-5p/11}\Delta^{p/2}\ln^{p/2}(\Delta^{-1}).
\end{eqnarray*}
We obtain the claim by choosing $p\geq15$. 
\end{proof}
Now, we are ready to prove Theorem \ref{thm:FlorensZmirou L2 error}.
The main ideas are as in \citep[Proposition 2]{Hoffmann:2001}. The
novelty consists on the direct treatment of the drift term and the
analysis of the boundary behaviour, which is an artifact of the reflection.
\begin{proof}[Proof of Theorem \ref{thm:FlorensZmirou L2 error}]
Set $\RG=\RG_{1}.$ Recall the definition (\ref{eq:forward Florens-Zmirou})
and the discussion thereafter. It follows, that it is sufficient to
prove the claim for $\hat{\sigma}_{j}^{2}(X_{0},\allowbreak X_{\Delta},\allowbreak...,\allowbreak X_{N\Delta})$.

Since
\[
\Delta\Big(\frac{1}{2}\I_{j}(X_{0})+\sum_{n=1}^{N-1}\I_{j}(X_{n\Delta})+\frac{1}{2}\I_{j}(X_{N\Delta})\Big)=\int_{\frac{j-1}{J}}^{\frac{j}{J}}\hat{\mu}_{N}(dx),
\]
by Lemma \ref{lem:mass of empirical measure on Ik}, on the event
$\RG_{1}$, the denominator of $\hat{\sigma}_{j}^{2}(X_{0},X_{\Delta},...,X_{N\Delta})$
has a uniform lower bound of order $\Delta^{1/3}$. Hence, in order
to prove (\ref{eq:FZ interior bound}), we have to show that, for
any $j=2,...,J-1$ and $x\in[\frac{j-1}{J},\frac{j}{J}]$, we have
\begin{equation}
\E\Big[\I_{\RG_{1}}\cdot\big|\sum_{n=0}^{N-1}\I_{j}(X_{n\Delta})\big((X_{(n+1)\Delta}-X_{n\Delta})^{2}-\Delta\sigma^{2}(x)\big)\big|^{2}\Big]^{\frac{1}{2}}\lesssim\Delta^{1/2}\Big(\int_{\frac{j-2}{J}}^{\frac{j+1}{J}}[(\sigma^{2})'(y)]^{2}dy\Big)^{\frac{1}{2}}+\Delta^{2/3}.\label{eq:interior bound}
\end{equation}
Indeed, (\ref{eq:interior bound}) implies
\[
\E\Big[\I_{\RG_{1}}\cdot\|\hat{\sigma}_{FZ}^{2}-\sigma^{2}\|_{L^{2}[1/J,1-1/J]}^{2}\Big]\lesssim\sum_{j=2}^{J-1}\frac{1}{J\Delta^{\frac{2}{3}}}\Big(\Delta\int_{\frac{j-2}{J}}^{\frac{j+1}{J}}[(\sigma^{2})'(y)]^{2}dy+\Delta^{\frac{4}{3}}\Big)=\Delta^{\frac{2}{3}}(\|\sigma^{2}\|_{H^{1}}^{2}+1).
\]

\emph{Step 1. Error bound in the interior.} Fix $2\leq j\leq J-1$
and $x\in[\frac{j-1}{J},\frac{j}{J}]$. Note that on the event $\RG_{1}$
the condition $\I_{j}(X_{n\Delta})=1$ implies that no reflection
occurs for $t\in[n\Delta,(n+1)\Delta]$. Using Itô formula we can
decompose 
\[
\sum_{n=0}^{N-1}\I_{j}(X_{n\Delta})\big((X_{(n+1)\Delta}-X_{n\Delta})^{2}-\Delta\sigma^{2}(x)\big):=A_{1}+A_{2}+A_{3}+A_{4},
\]
where 
\begin{eqnarray*}
A_{1} & = & \sum_{n=0}^{N-1}\I_{j}(X_{n\Delta})\Big[\Big(\int_{n\Delta}^{(n+1)\Delta}\sigma(X_{s})dW_{s}\Big)^{2}-\int_{n\Delta}^{(n+1)\Delta}\sigma^{2}(X_{s})ds\Big],\\
A_{2} & = & \sum_{n=0}^{N-1}\I_{j}(X_{n\Delta})\int_{n\Delta}^{(n+1)\Delta}(\sigma^{2}(X_{s})-\sigma^{2}(x))ds,\\
A_{3} & = & \sum_{n=0}^{N-1}\I_{j}(X_{n\Delta})\Big(\int_{n\Delta}^{(n+1)\Delta}b(X_{s})ds\Big)^{2},\\
A_{4} & = & -2\sum_{n=0}^{N-1}\I_{j}(X_{n\Delta})\int_{n\Delta}^{(n+1)\Delta}\sigma(X_{s})dW_{s}\int_{n\Delta}^{(n+1)\Delta}b(X_{s})ds.
\end{eqnarray*}
We will bound the second moment of each of the terms $A_{1},...,A_{4}$.
First, note that arguing as in the proof of Lemma \ref{lem:mass of empirical measure on Ik},
we obtain 
\begin{equation}
\frac{1}{N}\sum_{n=0}^{N-1}\I_{j}(X_{n\Delta})\leq(\Delta^{\frac{1}{3}}+4\mc)\|\mu_{1}\|_{\infty}.\label{eq:as bound on mean int mass}
\end{equation}
Consequently, from the Cauchy-Schwarz inequality, together with Theorem
\ref{thm:Moments of the modulus of continuity} and the inequality
(\ref{eq:bound on moments of mu_T}) follows that 
\begin{equation}
\E\Big[\Big(\frac{1}{N}\sum_{n=0}^{N-1}\I_{j}(X_{n\Delta})\Big)^{2}\Big]^{\frac{1}{2}}\lesssim\Delta^{\frac{1}{3}}.\label{eq:Bound on mean int mass}
\end{equation}

Denote by $\F_{n}$ the $\sigma-$field generated by $\{X_{m\Delta}:0\leq m\leq n\}.$
Let 
\[
\eta_{n}=\Big(\int_{n\Delta}^{(n+1)\Delta}\sigma(X_{s})dW_{s}\Big)^{2}-\int_{n\Delta}^{(n+1)\Delta}\sigma^{2}(X_{s})ds.
\]
Since $(\eta_{n})_{n}$ are $(\F_{n})-$martingale increments, they
are conditionally uncorrelated. Using the Burkholder-Davies-Gundy
inequality we obtain that $\E[\eta_{n}^{2}|\F_{n}]\lesssim\Delta^{2}.$
Consequently, 
\[
\E[A_{1}^{2}]^{\frac{1}{2}}=\Big(\sum_{n=0}^{N-1}\E[\I_{j}(X_{n\Delta})\eta_{n}^{2}]\Big)^{\frac{1}{2}}\lesssim\Delta^{\frac{1}{2}}\E\Big[\frac{1}{N}\sum_{n=0}^{N-1}\I_{j}(X_{n\Delta})\Big]^{\frac{1}{2}}\lesssim\Delta^{\frac{2}{3}},
\]
where we used (\ref{eq:Bound on mean int mass}) to obtain the last
inequality. On the event $\RG_{1}$, when $\I_{j}(X_{n\Delta})=1$,
we have 
\[
|X_{s}-x|\leq|X_{s}-X_{n\Delta}|+|X_{n\Delta}-x|\leq\mc+\Delta^{1/3}\lesssim\Delta^{1/3}.
\]
Using the Cauchy-Schwarz inequality we obtain that 
\[
A_{2}\leq\sum_{n=0}^{N-1}\I_{j}(X_{n\Delta})\int_{n\Delta}^{(n+1)\Delta}\Big|\int_{x}^{X_{s}}(\sigma^{2})'(y)dy\Big|ds\leq\frac{1}{N}\sum_{n=0}^{N-1}\I_{j}(X_{n\Delta})\Delta^{1/6}\Big(\int_{\frac{j-2}{J}}^{\frac{j+1}{J}}[(\sigma^{2})'(y)]^{2}dy\Big)^{1/2}.
\]
Hence
\[
\E[\I_{\RG_{1}}\cdot A_{2}^{2}]^{\frac{1}{2}}\leq\Delta^{1/2}\Big(\int_{\frac{j-2}{J}}^{\frac{j+1}{J}}[(\sigma^{2})'(y)]^{2}dy\Big)^{1/2}.
\]
 The drift function $b$ is uniformly bounded, hence $|A_{3}|\lesssim\Delta.$
Denote 
\[
Y_{t}=\int_{0}^{t}\sigma(X_{s})dW_{s}\:\text{and}\:\omega_{Y}(\Delta)=\sup_{\begin{subarray}{c}
0\leq s,t\leq1\\
|t-s|\leq\delta
\end{subarray}}|Y_{t}-Y_{s}|.
\]
The uniform bound on $b$, together with $|Y_{(n+1)\Delta}-Y_{n\Delta}|\leq\omega_{Y}(\Delta)$,
and the inequality (\ref{eq:as bound on mean int mass}) imply 
\begin{align*}
\E[\I_{\RG_{1}}\cdot A_{4}^{2}]^{\frac{1}{2}} & \lesssim\E\Big[\I_{\RG_{1}}\cdot\Big(\frac{1}{N}\sum_{n=0}^{N-1}\I_{j}(X_{n\Delta})\omega_{Y}(\Delta)\Big)^{2}\Big]^{\frac{1}{2}}\\
 & \lesssim\E\Big[\big(\Delta^{\frac{1}{3}}\|\mu_{1}\|_{\infty}\omega_{Y}(\Delta)\big)^{2}\Big]^{\frac{1}{2}}\lesssim\Delta^{\frac{2}{3}},
\end{align*}
where we used uniform bounds on the moments of modulus of continuity
of semimartingales with bounded coefficients (see Theorem \ref{thm:Moments of the modulus of continuity}).

\emph{Step 2. Error bound at the boundaries.} Set $j=1$ (the case
$j=J$ follows analogously) and $x\in[0,1/J]$. On $\RG_{1}$, whenever
$X_{n\Delta}\geq\Delta^{5/11}$, no reflection occurs for $t\in[n\Delta,(n+1)\Delta]$.
Denote 
\[
\I_{1}(x)=\I(x<\Delta^{5/11})+\I(\Delta^{5/11}\leq x<J^{-1}):=\I_{1,0}(x)+\I_{1,1}(x).
\]
We decompose 
\[
\sum_{n=0}^{N-1}\I_{1}(X_{n\Delta})\big((X_{(n+1)\Delta}-X_{n\Delta})^{2}-\Delta\sigma^{2}(x)\big):=E_{1}+E_{2},
\]
with
\begin{eqnarray*}
E_{1} & = & \sum_{n=0}^{N-1}\I_{1,0}(X_{n\Delta})\big((X_{(n+1)\Delta}-X_{n\Delta})^{2}-\Delta\sigma^{2}(x)\big),\\
E_{2} & = & \sum_{n=0}^{N-1}\I_{1,1}(X_{n\Delta})\big((X_{(n+1)\Delta}-X_{n\Delta})^{2}-\Delta\sigma^{2}(x)\big).
\end{eqnarray*}
On $\RG_{1}$ holds $|(X_{(n+1)\Delta}-X_{n\Delta})^{2}-\Delta\sigma^{2}(x)|\lesssim\Delta^{10/11}$.
Hence, arguing as in the proof of Lemma \ref{lem:mass of empirical measure on Ik},
we obtain that 
\begin{align*}
\E\Big[\I_{\RG_{1}}\cdot E_{1}^{2}\Big]^{\frac{1}{2}} & \lesssim\Delta^{-\frac{1}{11}}\E\Big[\I_{\RG_{1}}\cdot\Big(\frac{1}{N}\sum_{n=0}^{N-1}\I_{1,0}(X_{n\Delta})\Big)^{2}\Big]^{\frac{1}{2}}\\
 & \lesssim\Delta^{-\frac{1}{11}}\E\Big[\I_{\RG_{1}}\cdot\Big(\int_{0}^{1}\I_{1,0}(x)\mu_{1}(x)dx+4\mc\|\mu_{1}\|_{\infty}\Big)^{2}\Big]^{\frac{1}{2}}\lesssim\Delta^{\frac{4}{11}}.
\end{align*}
To bound the second moment of $E_{2}$, note that when $\I_{1,1}(X_{n\Delta})=1$
no reflection occurs for $t\in[n\Delta,(n+1)\Delta]$. Consequently,
we can proceed as in Step 1, obtaining 
\[
\E\Big[\I_{\RG_{1}}\cdot E_{2}^{2}\Big]^{\frac{1}{2}}\lesssim\Delta^{1/2}\Big(\int_{0}^{\frac{2}{J}}[(\sigma^{2})'(y)]^{2}dy\Big)^{\frac{1}{2}}+\Delta^{2/3}\lesssim\Delta^{1/2}\|\sigma^{2}\|_{H^{1}}.
\]
We conclude that
\[
\E\Big[\I_{\RG_{1}}\cdot|\hat{\sigma}_{FZ}^{2}(x)-\sigma^{2}(x)|^{2}\Big]^{\frac{1}{2}}\lesssim\Delta^{-1/3}\E\Big[\I_{\RG_{1}}\cdot(E_{1}+E_{2})^{2}\Big]^{\frac{1}{2}}\lesssim\Delta^{1/33}.
\]
\end{proof}
\begin{cor}
\label{cor:Uniform FZ error}For every $\epsilon>0$ and $\Delta$
sufficiently small, there exists an event $\RG=\RG(\epsilon)\subseteq\RG_{1},$
with $\P(\L\setminus\RG)\leq\epsilon$, such that on $\RG$ 
\begin{equation}
\hat{\sigma}_{FZ}^{2}(x)\sim1\quad\text{for every }x\in[0,1].\label{eq:Uniform error for sigma_FZ}
\end{equation}
\end{cor}
\begin{proof}
From Theorem \ref{thm:FlorensZmirou L2 error} follows that 
\[
\E\big[\I_{\RG_{1}}\cdot\|\hat{\sigma}_{FZ}^{2}-\sigma^{2}\|_{L^{2}[0,1]}^{2}\big]\lesssim\Delta^{1/3+2/33}.
\]
Set $\epsilon>0$. Let 
\[
\RG=\RG_{1}\cap\big\{\|\hat{\sigma}_{FZ}^{2}-\sigma^{2}\|_{L^{2}}^{2}\leq(2J)^{-1}\inf_{x\in[0,1]}\sigma^{4}(x)\big\}.
\]
From Markov's inequality, together with the lower bound on the probability
of the event $\RG_{1},$ follows that 
\[
\P(\L\setminus\RG)\lesssim\Delta^{2/33}.
\]
Hence, for $\Delta$ sufficiently small, we have $\P(\L\setminus\RG)\geq1-\epsilon$.
Since $\|\hat{\sigma}_{FZ}^{2}-\sigma^{2}\|_{\infty}^{2}\leq J\|\hat{\sigma}_{FZ}^{2}-\sigma^{2}\|_{L^{2}}^{2}$
we conclude that on $\RG$ holds $\hat{\sigma}_{FZ}^{2}\sim1$. 
\end{proof}

\subsection{Properties of the eigenpair $(\hat{\gamma}_{1},\hat{u}_{1})$\label{sub:Properties of the eigenfunction u}}

In this section we want to prove Proposition \ref{prop:Properties of the eigenfunction u}.
Because of the tridiagonal structure of the form $\hat{l}$, the direct
analysis of the eigenfunction $\hat{u}_{1}$ is difficult. Instead,
we consider the generalized eigenvalue problem for forms $\hat{f}$
(recall Definition \ref{def:Form f_hat}) and $\hat{g}$: 
\begin{eigenproblem}
\label{EigenProb: f_hat g_hat}Find $(\hat{\lambda},\hat{w})\in\R\times V_{J}^{0},$
with $\hat{w}\neq0$, such that 
\[
\hat{f}(\hat{w},v)=\hat{\lambda}\hat{g}(\hat{w},v)\text{ for every function }v\in V_{J}^{0}.
\]

\end{eigenproblem}
On the high probability event $\RG_{2}\subset\RG_{1}$ such that $\hat{\sigma}_{FZ}^{2}\sim1$
(see Corollary \ref{cor:Uniform FZ error}), the form $\hat{f}$ is
positive-definite and symmetric. Consequently, on $\RG_{2},$ the
Eigenproblem \ref{EigenProb: f_hat g_hat} has $J$ solutions $(\hat{\lambda}_{j},\hat{w}_{j})_{j=1,..,J}$
with $0<\hat{\lambda}_{1}\leq\hat{\lambda}_{2}\leq...\leq\hat{\lambda}_{J}$. 
\begin{defn}
\label{def:Basis psi_0  and matrices F_hat, M_hat}For $j=1,...,J$
define $\psi_{j}^{0}=\psi_{j}-\int_{0}^{1}\psi_{j}(x)\hat{\mu}_{N}(dx)\in V_{J}^{0}.$
Let 
\[
\hat{F}_{i,j}:=\hat{f}(\psi_{i}^{0},\psi_{j}^{0})=\hat{f}(\psi_{i},\psi_{j})\quad\text{and}\quad\hat{M}_{i,j}=\hat{g}(\psi_{i}^{0},\psi_{j}^{0})
\]
be the matrix representations of forms $\hat{f}$ and $\hat{g}$ on
$V_{J}^{0}\times V_{J}^{0}$ with respect to the algebraic basis $(\psi_{j}^{0})_{j}$. 
\end{defn}
Arguing as in \citet[Lemma 6.1]{GobetHoffmannReiss:2004} we obtain
that 
\begin{equation}
\hat{M}_{i,j}=\int_{\frac{i-1}{J}}^{\frac{i}{J}}\int_{\frac{j-1}{J}}^{\frac{j}{J}}\int_{0}^{y\wedge z}\hat{\mu}_{N}(dx)\int_{y\vee z}^{1}\hat{\mu}_{N}(dx)dydz.\label{eq:Formula for M}
\end{equation}
$\hat{F}$ is a diagonal matrix with strictly positive diagonal entries,
hence it is invertible. Eigenproblem \ref{EigenProb: f_hat g_hat}
is equivalent to 
\[
\hat{F}^{-1}\hat{M}(\hat{w}_{i,j})_{j}=\hat{\lambda}_{i}^{-1}(\hat{w}_{i,j})_{j},
\]
where $(\hat{w}_{i,j})_{j=1,...,J}$ indicates the coefficient vector
associated to the eigenfunction $\hat{w}_{i}$, i.e. 
\[
\hat{w}_{i}=\sum_{j=1}^{J}\hat{w}_{i,j}\psi_{j}^{N}=\sum_{j=1}^{J}\hat{w}_{i,j}\psi_{j}+\hat{w}_{i,0},
\]
with some $\hat{w}_{i,0}$ such that $\hat{w}_{i}\in V_{J}^{0}$.
Since the matrix $\hat{M}$ has all entries strictly positive, the
matrix $\hat{F}^{-1}\hat{M}$ satisfies the conditions of the Perron\textendash Frobenius
theorem. Consequently, the eigenvector $(\hat{w}_{1,j})_{j}$ can
be chosen strictly positive, which corresponds to the monotonicity
property of the eigenfunction $\hat{w}_{1}$. In what follows, we
will show that the Eigenproblem \ref{EigenProb: f_hat g_hat} is an
approximation of the Eigenproblem \ref{EigenProb:  l_hat g_hat} for
forms $\hat{l}$ and $\hat{g},$ and deduce that the eigenfunction
$\hat{u}_{1}$ inherits the properties of $\hat{w}_{1}$. Let $\|\cdot\|_{l^{2}}$
denote the standard Euclidean norm on $\R^{J}$. 
\begin{defn}
\label{def:Event R_alpha}Set $0<\alpha<1/42$. Denote by $\RG_{\alpha}$
the set of paths contained in $\L$ such that
\begin{enumerate}
\item $\mc\leq\Delta^{1/2-\alpha}$\label{enu:Event R_alpha - bounds on omega_Delta}
\item for every $x\in(0,1)$ holds $\hat{\sigma}_{FZ}^{2}(x)\sim1$ \label{enu:Event R_alpha - bounds on sigma_FZ}
\item \label{enu:Event R_alpha - L2 bound on sigma_FZ error}$\|\hat{\sigma}_{FZ}^{2}-\sigma^{2}\|_{L^{2}([a,b])}\leq\Delta^{1/3-\alpha}$ 
\item \label{enu:Event R_alpha - regularity of local time}occupation density
$\mu_{1}$ is $1/2-\alpha$ Hölder continuous with Hölder norm bounded
by $\alpha^{-1}$.
\end{enumerate}
\end{defn}
\begin{rem}
\label{rem:Properties of the event R_alphaM}By Theorem \ref{thm:Moments of the modulus of continuity},
Corollary \ref{cor:Uniform FZ error} and the regularity properties
of the occupation density $\mu_{1}$, for every $\epsilon>0$, there
exists $\alpha=\alpha(\epsilon)$ such that, for $\Delta$ sufficiently
small, $\P\big(\L\setminus\RG_{\alpha}\big)<\epsilon$ holds. The
assumption \ref{enu:Event R_alpha - L2 bound on sigma_FZ error} and
Hölder regularity of $\sigma^{2}$ (\ref{eq:Holder regularity of sigma})
imply that on the event $\RG_{\alpha}$ 
\begin{equation}
|\hat{\sigma}_{FZ}^{2}(x)-\sigma^{2}(x)|\lesssim\Delta^{1/6-\alpha}\text{ for all }x\in[a,b].\label{eq:Event R_alpha - uniform bound on sigma_FZ error}
\end{equation}
By the assumption \ref{enu:Event R_alpha - regularity of local time}
we have $\|\mu_{1}\|_{\infty}\lesssim1$. Furthermore, arguing as
in the proof of Lemma \ref{lem:mass of empirical measure on Ik},
we obtain 
\begin{equation}
\Big|\int_{\frac{j-1}{J}}^{\frac{j}{J}}\hat{\mu}_{N}(dx)-\int_{\frac{j-1}{J}}^{\frac{j}{J}}\mu_{1}(dx)\Big|\lesssim\mc\|\mu_{1}\|_{\infty}\lesssim\Delta^{1/2-\alpha}.\label{eq:Event R_alpha - a.s. bound on approx. error of the occupation measure}
\end{equation}
In particular, on $\RG_{\alpha}$ 
\begin{equation}
\int_{\frac{j-1}{J}}^{\frac{j}{J}}\hat{\mu}_{N}(dx)\sim\Delta^{1/3}\text{ holds for every }j=1,...,J.\label{eq:Event R_alpha - bounds on the empirical measure}
\end{equation}

\end{rem}
To bound the error between the solutions of the Eigenproblems \ref{EigenProb:  l_hat g_hat}
and \ref{EigenProb: f_hat g_hat} we need to establish uniform bounds
on the spectral gap of the Eigenproblem \ref{EigenProb: f_hat g_hat}. 
\begin{lem}
\label{lem:Uniform spectral gap of B_f}On the event $\RG_{\alpha}$
the eigenvalue $\hat{\lambda}_{1}$ is uniformly bounded. Furthermore,
the Eigenproblem \ref{EigenProb: f_hat g_hat} has a uniform spectral
gap, i.e. $\hat{\lambda}_{1}^{-1}-\hat{\lambda}_{2}^{-1}\gtrsim1.$\end{lem}
\begin{proof}
Consider the generalized eigenvalue problem: 
\begin{eigenproblem}
\label{EigenProb: true vol and local time}Find $(\lambda,w)\in\R\times V_{J}$
with $w\neq0$ and $\int_{0}^{1}w(x)\mu_{1}(x)dx=0$ such that 
\begin{align*}
\int_{0}^{1}w'(x)v'(x)\sigma^{2}(x)\mu_{1}(x)dx & =\lambda\int_{0}^{1}w(x)v(x)\mu_{1}(x)dx,\\
 & \text{ for all }v\in\{v\in V_{J}:\int_{0}^{1}v(x)\mu_{1}(x)dx=0\}.
\end{align*}

\end{eigenproblem}
Eigenproblem \ref{EigenProb: true vol and local time} has $J$ solutions,
denoted by $(\lambda_{j},w_{j})$ with $0<\lambda_{1}<\lambda_{2}\leq...\leq\lambda_{J}$.
By Proposition \ref{prop:Properties of eigenvalue problem for sigma mu on V_J}
we have $\lambda_{1}\sim1$ and $\lambda_{1}^{-1}-\lambda_{2}^{-1}\gtrsim1$.
Let $M,F$ be $J\times J$ matrices corresponding to the Eigenproblem
\ref{EigenProb: true vol and local time} tested with functions $(\psi_{j}^{1})_{j=1,...,J}$,
where $\psi_{j}^{1}=\psi_{J}-\int_{0}^{1}\psi_{j}\mu_{1}(x)dx$. As
in the case of the data driven Eigenproblem \ref{EigenProb: f_hat g_hat},
we have
\begin{eqnarray*}
M_{i,j} & = & \int_{\frac{i-1}{J}}^{\frac{i}{J}}\int_{\frac{j-1}{J}}^{\frac{j}{J}}\int_{0}^{y\wedge z}\mu_{1}(x)dx\int_{y\vee z}^{1}\mu_{1}(x)dxdydz,\\
F_{i,j} & = & \begin{cases}
0 & :i\neq j\\
\int_{\frac{i-1}{J}}^{\frac{i}{J}}\sigma^{2}(x)\mu_{1}(x)dx & :i=j
\end{cases}
\end{eqnarray*}
and
\[
F^{-1}M(w_{i,j})_{j}=\lambda_{i}^{-1}(w_{i,j})_{j}.
\]
From Weyl's theorem for symmetric eigenvalue problems follows that
\begin{equation}
|\hat{\lambda}_{i}^{-1}-\lambda_{i}^{-1}|\leq\|F^{-1}M-\hat{F}^{-1}\hat{M}\|_{l^{2}}.\label{eq:Weyl's theorem}
\end{equation}
We will show that $\|F^{-1}M-\hat{F}^{-1}\hat{M}\|_{l^{2}}\lesssim\Delta^{1/6-\alpha}$.
Then, the uniform bound on the eigenvalue $\hat{\lambda}_{1}$ and
the lower bound on the spectral gap will follow from the properties
of the Eigenproblem \ref{EigenProb: true vol and local time}.

First, let us observe that by (\ref{eq:Event R_alpha - a.s. bound on approx. error of the occupation measure})
and (\ref{eq:Event R_alpha - uniform bound on sigma_FZ error}), on
$\RG_{\alpha},$ for any $j=1,..,J$ we have
\begin{align*}
 & |\hat{F}_{j,j}-F_{j,j}|=\big|\int_{\frac{j-1}{J}}^{\frac{j}{J}}\hat{\sigma}_{FZ,j}^{2}\hat{\mu}_{N}(dx)-\int_{\frac{j-1}{J}}^{\frac{j}{J}}\sigma^{2}(x)\mu_{1}(x)dx\big|\lesssim\\
 & \qquad\lesssim\hat{\sigma}_{FZ,j}^{2}\big|\int_{\frac{j-1}{J}}^{\frac{j}{J}}\hat{\mu}_{N}(dx)-\int_{\frac{j-1}{J}}^{\frac{j}{J}}\mu_{1}(x)dx\big|+\int_{\frac{j-1}{J}}^{\frac{j}{J}}|\hat{\sigma}_{FZ}^{2}(x)-\sigma^{2}(x)|\mu_{1}(x)dx\lesssim\Delta^{1/2-\alpha}.
\end{align*}
In particular $\hat{F}_{j,j},F_{j,j}\sim\Delta^{1/3}.$ Arguing as
in the proof of Lemma \ref{lem:mass of empirical measure on Ik},
for any $i,j=1,...,J$, we obtain
\[
|M_{i,j}-\hat{M}_{i,j}|\lesssim J^{-2}\mc\|\mu_{1}\|_{\infty}\lesssim\Delta^{7/6-\alpha}.
\]
Furthermore $M_{i,j},\hat{M}_{i,j}\lesssim\Delta^{2/3}$. Since $F$
and $\hat{F}$ are diagonal matrices, it follows that
\[
|(F^{-1}M-\hat{F}^{-1}\hat{M})_{i,j}|\lesssim\Delta^{1/2-\alpha}.
\]
Hence, $\|F^{-1}M-\hat{F}^{-1}\hat{M}\|_{l^{2}}^{2}\leq\sum_{i,j=1}^{J}(F^{-1}M-\hat{F}^{-1}\hat{M})_{i,j}^{2}\lesssim\Delta^{1/3-2\alpha}$.\end{proof}
\begin{prop}
\label{prop:Bounds on the vector w_hat}Choose the eigenfunction $\hat{w}_{1}$
increasing and normalized so that $\|(\hat{w}_{1,j})_{j}\|_{l^{2}}=J^{1/2}$
(i.e. $\|\hat{w}_{1}'\|_{L^{2}}=1$). On the event $\RG_{\alpha}$,
for any $\left\lfloor aJ\right\rfloor -1\leq j\leq\left\lceil bJ\right\rceil +1$
and any $i=1,...,J$ we have
\begin{equation}
1\vee\hat{w}_{1,i}\lesssim\hat{w}_{1,j}\wedge1,\label{eq:Bound on ratios of w}
\end{equation}
Furthermore for $j$ s.t. $J^{1/2}\leq j\leq J-J^{1/2}$ 
\begin{equation}
\Big|\frac{\hat{w}_{1,j\pm1}}{\hat{w}_{1,j}}-1\Big|\lesssim\Delta^{1/6-\alpha}.\label{eq:error between w and unit vector}
\end{equation}
\end{prop}
\begin{proof}
In the proof we will use standard techniques from the Perron-Frobenius
theory of nonnegative matrices (cf. \citet[Chapter II]{Minc:1988}).
In particular, we shall repeatedly use the following inequality \citet[Chapter II, Section 2.1, Eq. (7)]{Minc:1988}:
for any $q_{1},q_{2},...,q_{n}>0$ and $p_{1},p_{2},...,p_{n}\in\R$
\begin{equation}
\min_{i=1,...,n}\frac{p_{i}}{q_{i}}\leq\frac{p_{1}+p_{2}+...+p_{n}}{q_{1}+q_{2}+...+q_{n}}\leq\max_{i=1,...,n}\frac{p_{i}}{q_{i}}.\label{eq:Frobenius bounds inequality}
\end{equation}

\emph{Step 1: ($\hat{w}_{1,i}\lesssim1$).} Fix $1\leq i\leq J$.
By Definition \ref{def:Event R_alpha}.\ref{enu:Event R_alpha - bounds on sigma_FZ},
relation (\ref{eq:Event R_alpha - bounds on the empirical measure})
and Lemma \ref{lem:Uniform spectral gap of B_f}, on the event $\RG_{\alpha}$,
we have
\begin{equation}
J^{-1}\hat{w}_{1,i}\sim\hat{w}_{1,i}\hat{\sigma}_{FZ,i}^{2}\int_{\frac{i-1}{J}}^{\frac{i}{J}}\mu_{N}(x)=\hat{f}(\hat{w}_{1},\psi_{i})=\hat{\lambda}_{1}\hat{g}(\hat{w}_{1},\psi_{i})\sim\hat{g}(\hat{w}_{1},\psi_{i})=\sum_{m=1}^{J}\hat{M}_{i,m}\hat{w}_{1,m}.\label{eq:Matrix representation of w_1,i}
\end{equation}
Hence, by the Cauchy-Schwarz inequality
\[
\hat{w}_{1,i}\lesssim J\Big(\sum_{m=1}^{J}M_{i,m}^{2}\Big)^{1/2}\Big(\sum_{m=1}^{J}\hat{w}_{1,m}^{2}\Big)^{1/2}\lesssim1,
\]
where we used $M_{i,m}\leq J^{-2}$ and the normalization of $(\hat{w}_{1,j}).$

\emph{Step 2: ($\hat{w}_{1,i}\lesssim\hat{w}_{1,j}$).} Fix $\left\lfloor aJ\right\rfloor -1\leq j\leq\left\lceil bJ\right\rceil +1$.
On the event $\RG_{\alpha},$ for any $1\leq i\leq J$ the relation
(\ref{eq:Matrix representation of w_1,i}) together with the inequality
(\ref{eq:Frobenius bounds inequality}) imply
\begin{equation}
\frac{\hat{w}_{1,i}}{\hat{w}_{1,j}}\sim\frac{\sum_{m=1}^{J}\hat{M}_{i,m}\hat{w}_{1,m}}{\sum_{m=1}^{J}\hat{M}_{j,m}\hat{w}_{1,m}}\lesssim\max_{m=1,...,J}\frac{\hat{M}_{i,m}}{\hat{M}_{j,m}}.\label{eq:ratio of w similar to ratio of M}
\end{equation}
We need to show that for arbitrary $m$ $\hat{M}_{i,m}/\hat{M}_{j,m}\lesssim1$
holds. Consider first the case $i<j$. Then by (\ref{eq:Formula for M})
\begin{eqnarray*}
\frac{\hat{M}_{i,m}}{\hat{M}_{j,m}} & = & \frac{\int_{\frac{i-1}{J}}^{\frac{i}{J}}\int_{\frac{m-1}{J}}^{\frac{m}{J}}\int_{0}^{y\wedge z}\hat{\mu}_{N}(dx)\int_{y\vee z}^{1}\hat{\mu}_{N}(dx)dydz}{\int_{\frac{j-1}{J}}^{\frac{j}{J}}\int_{\frac{m-1}{J}}^{\frac{m}{J}}\int_{0}^{y\wedge z}\hat{\mu}_{N}(dx)\int_{y\vee z}^{1}\hat{\mu}_{N}(dx)dydz}\\
 & \leq & \frac{\int_{\frac{i-1}{J}}^{\frac{i}{J}}\int_{\frac{m-1}{J}}^{\frac{m}{J}}\int_{0}^{y\wedge\frac{j-1}{J}}\hat{\mu}_{N}(dx)\int_{y\vee\frac{i-1}{J}}^{1}\hat{\mu}_{N}(dx)dydz}{\int_{\frac{j-1}{J}}^{\frac{j}{J}}\int_{\frac{m-1}{J}}^{\frac{m}{J}}\int_{0}^{y\wedge\frac{j-1}{J}}\hat{\mu}_{N}(dx)\int_{y\vee\frac{j}{J}}^{1}\hat{\mu}_{N}(dx)dydz}\\
 & = & \frac{\int_{\frac{m-1}{J}}^{\frac{m}{J}}f(y)\int_{y\vee\frac{i-1}{J}}^{1}\hat{\mu}_{N}(dx)dy}{\int_{\frac{m-1}{J}}^{\frac{m}{J}}f(y)\int_{y\vee\frac{j}{J}}^{1}\hat{\mu}_{N}(dx)dy},
\end{eqnarray*}
where $f(y)=\int_{0}^{y\wedge\frac{j-1}{J}}\hat{\mu}_{N}(dx).$ Consider
$m>j$. For $y\in[\frac{m-1}{J},\frac{m}{J}]$ holds $y=y\vee\frac{j}{J}=y\vee\frac{i-1}{J}$,
hence the numerator and denominator are equal. Consider $m\leq j$.
For $y\in[\frac{m-1}{J},\frac{m}{J}]$ holds $y\vee\frac{j}{J}=\frac{j}{J}.$
Hence, using (\ref{eq:Event R_alpha - bounds on the empirical measure}),
we obtain
\begin{alignat*}{1}
\frac{\hat{M}_{i,m}}{\hat{M}_{j,m}} & \leq\frac{\int_{\frac{m-1}{J}}^{\frac{m}{J}}f(y)\int_{y\vee\frac{i-1}{J}}^{1}\hat{\mu}_{N}(dx)dy}{\int_{\frac{m-1}{J}}^{\frac{m}{J}}f(y)\int_{y\vee\frac{j}{J}}^{1}\hat{\mu}_{N}(dx)dy}\leq\frac{\int_{\frac{m-1}{J}}^{\frac{m}{J}}f(y)dy}{\int_{\frac{m-1}{J}}^{\frac{m}{J}}f(y)dy\int_{\frac{j}{J}}^{1}\hat{\mu}_{N}(dx)}\\
 & =\Big(\int_{\frac{j}{J}}^{1}\hat{\mu}_{N}(dx)\Big)^{-1}\sim\Big(1-\frac{j}{J}\Big)^{-1}\lesssim1.
\end{alignat*}
We conclude that for $i<j$ and arbitrary $m$ bound $\hat{M}_{i,m}/\hat{M}_{j,m}\lesssim1$
holds. Proceeding analogously, we obtain the same claim for $i>j$.
From (\ref{eq:ratio of w similar to ratio of M}) follows that on
the event $\RG_{\alpha}$, for $\left\lfloor aJ\right\rfloor -1\leq j\leq\left\lceil bJ\right\rceil +1$
and any $1\leq i\leq J$ , we have
\begin{equation}
\hat{w}_{1,i}\lesssim\hat{w}_{1,j}.\label{eq:w_i smaller than w_j}
\end{equation}

\emph{Step 3: ($1\lesssim\hat{w}_{1,j}$).} Let $\hat{w}_{1,j_{0}}=\min_{\left\lfloor aJ\right\rfloor -1\leq j\leq\left\lceil bJ\right\rceil +1}\hat{w}_{1,j}$.
Inequality (\ref{eq:w_i smaller than w_j}) implies
\[
1=\frac{1}{J}\sum_{i=1}^{J}\hat{w}_{1,i}^{2}\lesssim\hat{w}_{1,j_{0}}^{2}.
\]

\emph{Step 4: (proof of (\ref{eq:error between w and unit vector})).}
We will only show $\big|\frac{\hat{w}{}_{1,j+1}}{\hat{w}{}_{1,j}}-1\big|\lesssim\Delta^{1/6-\alpha}$,
the other bound can be obtained by a symmetric argument. First, note
that from Definition \ref{def:Event R_alpha}.\ref{enu:Event R_alpha - regularity of local time}
together with the inequality (\ref{eq:Event R_alpha - a.s. bound on approx. error of the occupation measure})
follows that
\begin{alignat*}{1}
 & \Big|\int_{\frac{j}{J}}^{\frac{j+1}{J}}\hat{\mu}_{N}(dx)-\int_{\frac{j-1}{J}}^{\frac{j}{J}}\hat{\mu}_{N}(dx)\Big|\leq\Big|\int_{\frac{j}{J}}^{\frac{j+1}{J}}\hat{\mu}_{N}(dx)-\int_{\frac{j}{J}}^{\frac{j+1}{J}}\mu_{1}(x)(dx)\Big|+\\
 & +\Big|\int_{\frac{j}{J}}^{\frac{j+1}{J}}\mu_{1}(x)(dx)-\int_{\frac{j-1}{J}}^{\frac{j}{J}}\mu_{1}(x)(dx)\Big|+\Big|\int_{\frac{j-1}{J}}^{\frac{j}{J}}\mu_{1}(x)(dx)-\int_{\frac{j-1}{J}}^{\frac{j}{J}}\hat{\mu}_{N}(dx)\Big|\\
 & \lesssim\Delta^{1/2-\alpha}+\Delta^{1/3}\Delta^{(1/2-\alpha)/3}+\Delta^{1/2-\alpha}\lesssim\Delta^{1/2-\alpha}.
\end{alignat*}
Hence, by (\ref{eq:Event R_alpha - bounds on the empirical measure})
\[
\Big|\frac{\int_{\frac{j}{J}}^{\frac{j+1}{J}}\hat{\mu}_{N}(dx)}{\int_{\frac{j-1}{J}}^{\frac{j}{J}}\hat{\mu}_{N}(dx)}-1\Big|\lesssim\Delta^{1/6-\alpha}.
\]
Similarly, by the $1/2-$Hölder regularity of $\sigma^{2}$ and Definition
\ref{def:Event R_alpha}.\ref{enu:Event R_alpha - bounds on sigma_FZ}
together with (\ref{eq:Event R_alpha - uniform bound on sigma_FZ error})
we have
\[
\Big|\frac{\hat{\sigma}_{FZ,j+1}^{2}}{\hat{\sigma}_{FZ,j}^{2}}-1\Big|\lesssim\Delta^{1/6-\alpha}.
\]
Consequently, instead of $\frac{\hat{w}_{1,j+1}}{\hat{w}_{1,j}}$
we may consider
\[
\frac{\hat{w}_{1,j+1}\hat{\sigma}_{FZ,j+1}^{2}\int_{\frac{j}{J}}^{\frac{j+1}{J}}\hat{\mu}_{N}(dx)}{\hat{w}_{1,j}\hat{\sigma}_{FZ,j}^{2}\int_{\frac{j-1}{J}}^{\frac{j}{J}}\hat{\mu}_{N}(dx)}=\frac{\sum_{m=1}^{J}\hat{M}_{j+1,m}\hat{w}_{1,m}}{\sum_{m=1}^{J}\hat{M}_{j,m}\hat{w}_{1,m}}.
\]
By the inequality (\ref{eq:Frobenius bounds inequality}) 
\[
\min_{m=1,...,J}\frac{\hat{M}_{j+1,m}}{\hat{M}_{j,m}}\leq\frac{\sum_{m=1}^{J}\hat{M}_{j+1,m}\hat{w}_{1,m}}{\sum_{m=1}^{J}\hat{M}_{j,m}\hat{w}_{1,m}}\leq\max_{m=1,...,J}\frac{\hat{M}_{j+1,m}}{\hat{M}_{j,m}}.
\]
Thus, it is enough to show, that for any $m=1,...,J$ bound $\Big|\frac{\hat{M}_{j+1,m}}{\hat{M}_{j,m}}-1\Big|\lesssim\Delta^{1/6}$
holds.
\begin{eqnarray*}
\frac{\hat{M}_{j+1,m}}{\hat{M}_{j,m}} & = & \frac{\int_{\frac{j}{J}}^{\frac{j+1}{J}}\int_{\frac{m-1}{J}}^{\frac{m}{J}}\int_{0}^{y\wedge z}\hat{\mu}_{N}(dx)\int_{y\vee z}^{1}\hat{\mu}_{N}(dx)dydz}{\int_{\frac{j-1}{J}}^{\frac{j}{J}}\int_{\frac{m-1}{J}}^{\frac{m}{J}}\int_{0}^{y\wedge z}\hat{\mu}_{N}(dx)\int_{y\vee z}^{1}\hat{\mu}_{N}(dx)dydz}\\
 & \leq & \frac{\int_{\frac{j}{J}}^{\frac{j+1}{J}}\int_{\frac{m-1}{J}}^{\frac{m}{J}}\int_{0}^{y\wedge\frac{j+1}{J}}\hat{\mu}_{N}(dx)\int_{y\vee\frac{j}{J}}^{1}\hat{\mu}_{N}(dx)dydz}{\int_{\frac{j-1}{J}}^{\frac{j}{J}}\int_{\frac{m-1}{J}}^{\frac{m}{J}}\int_{0}^{y\wedge\frac{j-1}{J}}\hat{\mu}_{N}(dx)\int_{y\vee\frac{j}{J}}^{1}\hat{\mu}_{N}(dx)dydz}\\
 & \leq & \frac{\int_{\frac{m-1}{J}}^{\frac{m}{J}}\int_{0}^{y\wedge\frac{j+1}{J}}\hat{\mu}_{N}(dx)f(y)dy}{\int_{\frac{m-1}{J}}^{\frac{m}{J}}\int_{0}^{y\wedge\frac{j-1}{J}}\hat{\mu}_{N}(dx)f(y)dy}=1+\frac{\int_{\frac{m-1}{J}}^{\frac{m}{J}}\int_{y\wedge\frac{j-1}{J}}^{y\wedge\frac{j+1}{J}}\hat{\mu}_{N}(dx)f(y)dy}{\int_{\frac{m-1}{J}}^{\frac{m}{J}}\int_{0}^{y\wedge\frac{j-1}{J}}\hat{\mu}_{N}(dx)f(y)dy},
\end{eqnarray*}
where $f(y)=\int_{y\vee\frac{j}{J}}^{1}\hat{\mu}_{N}(dx)$. Consider
$m\leq j-1$. For $y\in[\frac{m-1}{J},\frac{m}{J}]$ we have $y=y\wedge\frac{j+1}{J}=y\wedge\frac{j-1}{J}$,
hence the error term is zero. Consider $m\geq j$. For $y\in[\frac{m-1}{J},\frac{m}{J}]$
we have $y\wedge\frac{j-1}{J}=\frac{j-1}{J}$. Consequently, using
(\ref{eq:Event R_alpha - bounds on the empirical measure}), we obtain
that for $j\geq J^{1/2}\sim\Delta^{-1/6}$
\begin{align*}
\frac{\int_{\frac{m-1}{J}}^{\frac{m}{J}}\int_{y\wedge\frac{j-1}{J}}^{y\wedge\frac{j+1}{J}}\hat{\mu}_{N}(dx)f(y)dy}{\int_{\frac{m-1}{J}}^{\frac{m}{J}}\int_{0}^{y\wedge\frac{j-1}{J}}\hat{\mu}_{N}(dx)f(y)dy} & \leq\frac{\int_{\frac{m-1}{J}}^{\frac{m}{J}}\int_{\frac{j-1}{J}}^{\frac{j+1}{J}}\hat{\mu}_{N}(dx)f(y)dy}{\int_{\frac{m-1}{J}}^{\frac{m}{J}}\int_{0}^{\frac{j-1}{J}}\hat{\mu}_{N}(dx)f(y)dy}\\
 & \sim\frac{2J^{-1}}{\frac{j-1}{J}}=\frac{2}{j-1}\lesssim\Delta^{1/6}.
\end{align*}
Finally, symmetric bound $1-\frac{\hat{M}_{j+1,m}}{\hat{M}_{j,m}}\lesssim\Delta^{1/6}$
can be obtained by similar calculations. 
\end{proof}
In previous proposition we have established uniform bounds on the
eigenfunction $\hat{w}_{1}$ . Next, we show that $\hat{w}_{1}$ is
a good approximation of $\hat{u}_{1}$. 
\begin{defn}
Let $\hat{L}$ be the matrix representation of the form $\hat{l}$
with respect to the algebraic basis $(\psi_{j}^{0})_{j}$ (see Definition
\ref{def:Basis psi_0  and matrices F_hat, M_hat}), i.e.
\[
\hat{L}_{i,j}:=\hat{l}(\psi_{i}^{0},\psi_{j}^{0})=\hat{l}(\psi_{i},\psi_{j}).
\]
On the event $\RG_{\alpha}$, for $\Delta$ sufficiently small, the
matrix $\hat{L}$ is symmetric tridiagonal. We want to bound the error
between the solutions of the generalized eigenproblems:
\[
\hat{M}(\hat{w}_{i})=\hat{\lambda}_{1}^{-1}\hat{F}(\hat{w}_{i})\quad\text{and}\quad\hat{M}(\hat{u}_{i})=\hat{\gamma}_{1}^{-1}\hat{L}(\hat{u}_{i}).
\]
\end{defn}
\begin{lem}
\label{lem:Norm of F-L}On the event $\RG_{\alpha}$ holds
\begin{equation}
\|\hat{F}-\hat{L}\|_{l^{2}}\lesssim\Delta^{1/2-3\alpha}.\label{eq:Norm of F-L}
\end{equation}
Furthermore matrix $\hat{L}$ is invertible and
\[
\|\hat{L}\|_{l^{2}},\|\hat{F}\|_{l^{2}},\|\hat{L}^{-1}\|_{l^{2}}^{-1},\|\hat{F}^{-1}\|_{l^{2}}^{-1}\sim\Delta^{1/3}.
\]
\end{lem}
\begin{proof}
Consider vector $(v_{j})_{j}\in\R^{J}$ with $\|(v_{j})_{j}\|_{l^{2}}=1$
and the corresponding function $v=\sum_{j=1}^{J}v_{j}\psi_{j}^{0}(x)\in V_{J}^{0}$.
Since
\begin{align*}
\|(\hat{F}-\hat{L})v\|_{l^{2}}^{2} & =\sum_{j=1}^{J}|\hat{f}(v,\psi_{j})-\hat{l}(v,\psi_{j})|^{2}\\
 & =\sum_{j=1}^{J}\big(v_{j-1}\hat{L}_{j-1,j}+v_{j}(\hat{F}_{j,j}-\hat{L}_{j,j})+v_{j+1}\hat{L}_{j+1,j}\big)^{2},
\end{align*}
to obtain (\ref{eq:Norm of F-L}), we just have to argue that $\hat{L}_{j-1,j},|\hat{F}_{j,j}-\hat{L}_{j,j}|$
and $\hat{L}_{j+1,j}$ are of order $\Delta^{1/2-3\alpha}.$ By the
definition of the forms $\hat{l}$ and $\hat{g}$ from the Eigenproblem
\ref{EigenProb:  l_hat g_hat}
\begin{eqnarray*}
2|\hat{L}_{j-1,j}| & = & \sum_{n=0}^{N-1}(\psi_{j-1}(X_{(n+1)\Delta})-\psi_{j-1}(X_{n\Delta}))(\psi_{j}(X_{(n+1)\Delta})-\psi_{j}(X_{n\Delta}))\\
 & = & \sum_{n=0}^{N-1}\I(X_{n\Delta}<{\textstyle \frac{j-1}{J}})\I(X_{(n+1)\Delta}>{\textstyle \frac{j-1}{J}})({\textstyle \frac{j-1}{J}}-X_{n\Delta})(X_{(n+1)\Delta}-{\textstyle \frac{j-1}{J}})\\
 &  & \qquad+\sum_{n=0}^{N-1}\I(X_{n\Delta}>{\textstyle \frac{j-1}{J}})\I(X_{(n+1)\Delta}<{\textstyle \frac{j-1}{J}})({\textstyle \frac{j-1}{J}-X_{(n+1)\Delta}})(X_{n\Delta}-{\textstyle \frac{j-1}{J}})\\
 & \lesssim & \sum_{n=0}^{N-1}\I(X_{n\Delta}<{\textstyle \frac{j-1}{J}})\I(X_{(n+1)\Delta}>{\textstyle \frac{j-1}{J}})({\textstyle X_{(n+1)\Delta}}-X_{n\Delta})^{2}\\
 &  & \qquad+\sum_{n=0}^{N-1}\I(X_{n\Delta}>{\textstyle \frac{j-1}{J}})\I(X_{(n+1)\Delta}<{\textstyle \frac{j-1}{J}})({\textstyle X_{(n+1)\Delta}}-X_{n\Delta})^{2}.
\end{eqnarray*}
Moreover
\begin{align*}
|\hat{F}_{j,j}-\hat{L}_{j,j}| & \leq\frac{1}{2}\sum_{n=0}^{N-1}|\I(X_{n\Delta}<{\textstyle \frac{j-1}{J}})-\I(X_{(n+1)\Delta}<{\textstyle \frac{j-1}{J}})|(X_{(n+1)\Delta}-X_{n\Delta})^{2}+\\
 & \qquad+\frac{1}{2}\sum_{n=0}^{N-1}|\I(X_{n\Delta}<{\textstyle \frac{j}{J}})-\I(X_{(n+1)\Delta}<{\textstyle \frac{j}{J}})|(X_{(n+1)\Delta}-X_{n\Delta})^{2}
\end{align*}
Hence, it suffices to show that for any $x\in(0,1)$
\begin{equation}
\sum_{n=0}^{N-1}\I(X_{n\Delta}<x)\I(X_{(n+1)\Delta}>x)(X_{(n+1)\Delta}-X_{n\Delta})^{2}\lesssim\Delta^{1/2-3\alpha}.\label{eq:Partial - crossings bound}
\end{equation}
By Definition \ref{def:Event R_alpha}.\ref{enu:Event R_alpha - bounds on omega_Delta},
on the event $\RG_{\alpha}$, we have
\begin{align*}
\sum_{n=0}^{N-1}\I(X_{n\Delta}<x)\I(X_{(n+1)\Delta} & >x)(X_{(n+1)\Delta}-X_{n\Delta})^{2}\\
 & \leq\Delta^{-2\alpha}\frac{1}{N}\sum_{n=0}^{N-1}\I(|X_{n\Delta}-x|\leq\Delta^{1/2-\alpha}).
\end{align*}
Arguing as in the proof of Lemma \ref{lem:mass of empirical measure on Ik},
we finally obtain
\[
\frac{1}{N}\sum_{n=0}^{N-1}\I(|X_{n\Delta}-x|\leq\Delta^{1/2-\alpha})\lesssim\int_{x-\Delta^{1/2-\alpha}}^{x+\Delta^{1/2-\alpha}}\mu_{1}(x)dx+\mc\|\mu_{1}\|_{\infty}\lesssim\Delta^{1/2-\alpha}.
\]

Since $\hat{F}$ is a diagonal matrix with diagonal entries of order
$\Delta^{1/3}$, we have $\|\hat{F}\|_{l^{2}},\|\hat{F}^{-1}\|_{l^{2}}^{-1}\sim\Delta^{1/3}.$
As argued above, on $\RG_{\alpha}$, the upper and lower diagonal
entries of $\hat{L}$ are of order $\Delta^{1/2-3\alpha}$. Since
for any $1\leq j\leq J$ holds $|\hat{L}_{j,j}-\hat{F}_{j,j}|\lesssim\Delta^{1/2-3\alpha}$,
matrix $\hat{L}$ is diagonally dominant with diagonal entries of
order $\Delta^{1/3}$. Hence it is invertible and $\|\hat{L}\|_{l^{2}},\|\hat{L}^{-1}\|_{l^{2}}^{-1}\sim\Delta^{1/3}.$\end{proof}
\begin{lem}
\label{lem:error between eigenvectors u and w}Eigenvectors $(\hat{w}_{1,j})$,
$(\hat{u}_{1,j})$, normalized so that $\|\hat{w}_{1}\|_{l^{2}}=\|\hat{u}_{1}\|_{l^{2}}=J^{1/2}$,
satisfy on $\RG_{\alpha}$ 
\[
\|(\hat{w}_{1,j})-(\hat{u}_{1,j})\|_{l^{2}}\lesssim\Delta^{-1/3}\|(\hat{F}-\hat{L})\hat{w}_{1}\|_{l^{2}}.
\]
\end{lem}
\begin{proof}
Recall that $(\hat{\lambda}_{j},\hat{w}_{j})_{j}$ are the eigenpairs
of the Eigenproblem \ref{EigenProb: f_hat g_hat}, with $\|(\hat{w}_{j})\|_{l^{2}}=\sqrt{J}$.
\citep[Theorem 26]{ChorowskiTrabs:2015} implies that there exists
an eigenpair $(\hat{\lambda}_{j_{0}},J^{-1/2}\hat{w}_{j_{0}})$ such
that
\begin{eqnarray*}
|\hat{\lambda}_{j_{0}}^{-1}-\hat{\gamma}_{1}^{-1}| & \lesssim & J^{-1/2}\|\hat{F}^{-1}\|_{l^{2}}\|(\hat{F}-\hat{L})\hat{w}_{1}\|_{l^{2}}\lesssim\|\hat{F}^{-1}\|_{l^{2}}\|\hat{F}-\hat{L}\|_{l^{2}},\\
\|(\hat{w}_{j_{0},j})-(\hat{u}_{1,j})\|_{l^{2}} & \lesssim & \delta^{-1}(\hat{\lambda}_{j_{0}}^{-1})\|\hat{F}^{-1}\|_{l^{2}}^{3/2}\|\hat{F}\|_{l^{2}}^{1/2}\|(\hat{F}-\hat{L})\hat{w}_{1}\|_{l^{2}},
\end{eqnarray*}
where $\delta(\hat{\lambda}_{j_{0}}^{-1})$ is the so called localizing
distance, i.e. $\delta(\hat{\lambda}_{j_{0}}^{-1})=\min_{j\neq j_{0}}|\hat{\lambda}_{j_{0}}^{-1}-\hat{\gamma}_{1}^{-1}|.$
From Lemma \ref{lem:Norm of F-L} we deduce
\[
|\hat{\lambda}_{j_{0}}^{-1}-\hat{\gamma}_{1}^{-1}|\lesssim\Delta^{1/6-3\alpha}.
\]
By \citet[Theorem 8.3]{Nakatsukasa:2011} for any $i=1,...,J$ we
have
\[
|\hat{\lambda}_{i}^{-1}-\hat{\gamma}_{i}^{-1}|\lesssim\|\hat{L}^{-1}\|_{l^{2}}\|\hat{\lambda}_{i}^{-1}(\hat{F}-\hat{L})\|_{l^{2}},
\]
which together with Lemmas \ref{lem:Uniform spectral gap of B_f}
and \ref{lem:Norm of F-L} imply
\begin{equation}
|\hat{\lambda}_{1}^{-1}-\hat{\gamma}_{1}^{-1}|\lesssim\Delta^{1/6-3\alpha}.\label{eq:First eigenvalues error bound}
\end{equation}
By Lemma \ref{lem:Uniform spectral gap of B_f} holds $|\hat{\lambda}_{1}^{-1}-\hat{\lambda}_{2}^{-1}|\gtrsim1$,
hence we must have $j_{0}=1$. Furthermore, from the same uniform
lower bound on the spectral gap follows
\[
\delta(\hat{\lambda}_{j_{0}}^{-1})=\delta(\hat{\lambda}_{1}^{-1})\gtrsim1.
\]
Since by Lemma \ref{lem:Norm of F-L} we have $\|\hat{F}^{-1}\|_{l^{2}}^{3/2}\|\hat{F}\|_{l^{2}}^{1/2}\lesssim\Delta^{-1/3}$,
we conclude that the claim holds. 
\end{proof}

\begin{proof}[Proof of Proposition \ref{prop:Properties of the eigenfunction u}]
Set $\epsilon>0$. By Remark \ref{rem:Properties of the event R_alphaM}
there exists $\alpha$ s.t. $\P(\L\setminus\RG_{\alpha})\leq\epsilon$.
Set
\[
\RG_{2}=\RG_{\alpha}\cap\big\{\|\hat{w}_{1}-\hat{u}_{1}\|_{l^{2}}^{2}\leq\Delta^{1/7-6\alpha}\big\}.
\]

\emph{Step 1.} We will show
\begin{equation}
\E\big[\I_{\RG_{\alpha}}\cdot\|(\hat{F}-\hat{L})\hat{w}_{1}\|_{l^{2}}^{2}\big]^{1/2}\lesssim\Delta^{5/12-3\alpha}.\label{eq:Mean error bound F - L on w}
\end{equation}
In the proof of Lemma \ref{lem:Norm of F-L} we argued that for any
$j=1,...,J$ holds
\begin{equation}
\hat{l}(\psi_{j},\psi_{j-1}),\hat{l}(\psi_{j},\psi_{j+1}),|\hat{l}(\psi_{j},\psi_{j})-\hat{f}(\psi_{j},\psi_{j})|\lesssim\Delta^{1/2-3\alpha}.\label{eq:Partial - bounds on terms l,f}
\end{equation}
Hence, using uniform bound (\ref{eq:Bound on ratios of w}), we deduce
\begin{equation}
|\hat{l}(\hat{w}_{1},\psi_{j})-\hat{f}(\hat{w}_{1},\psi_{j})|\lesssim\Delta^{1/2-3\alpha}.\label{eq:Partial - suboptimal bound}
\end{equation}
We will use the regularity of the eigenfunction $\hat{w}_{1}$ to
strengthen (\ref{eq:Partial - suboptimal bound}). Consider $J^{1/2}\leq j\leq J-J^{1/2}$.
By (\ref{eq:Bound on ratios of w})
\begin{align*}
|\hat{l}(\hat{w}_{1},\psi_{j})-\hat{f}(\hat{w}_{1},\psi_{j})| & \lesssim\hat{l}(\psi_{j-1},\psi_{j})\Big|\frac{\hat{w}_{1,j-1}}{\hat{w}_{1,j}}-1\Big|+\\
 & +\big|\hat{l}(I,\psi_{j})-\hat{f}(I,\psi_{j})\big|+\hat{l}(\psi_{j+1},\psi_{j})\Big|\frac{\hat{w}_{1,j+1}}{\hat{w}_{1,j}}-1\Big|.
\end{align*}
Inequalities (\ref{eq:Partial - bounds on terms l,f}) and (\ref{eq:error between w and unit vector})
imply
\[
\hat{l}(\psi_{j-1},\psi_{j})\Big(\frac{\hat{w}_{1,j-1}}{\hat{w}_{1,j}}-1\Big)+\hat{l}(\psi_{j+1},\psi_{j})\Big(\frac{\hat{w}_{1,j+1}}{\hat{w}_{1,j}}-1\Big)\Big)\lesssim\Delta^{2/3-4\alpha},
\]
while, since $\RG_{\alpha}\subset\RG_{1}$ , from Lemma \ref{lem:Error of l and f on the unit vector}
follows
\[
\E\Big[\I_{\RG_{\alpha}}\cdot\big|\hat{l}(I,\psi_{j})-\hat{f}(I,\psi_{j})\big|^{2}\Big]^{\frac{1}{2}}\lesssim\Delta^{\frac{2}{3}}.
\]
We conclude that for $J^{1/2}\leq j\leq J-J^{1/2}$
\begin{equation}
\E\Big[\I_{\RG_{\alpha}}\cdot|\hat{l}(\hat{w}_{1},\psi_{j})-\hat{f}(\hat{w}_{1},\psi_{j})|^{2}\Big]^{\frac{1}{2}}\lesssim\Delta^{\frac{2}{3}-4\alpha}.\label{eq:Partial - optimal bound}
\end{equation}
Since $\alpha<\frac{1}{12}$ inequalities (\ref{eq:Partial - suboptimal bound})
and (\ref{eq:Partial - optimal bound}) imply
\begin{align*}
\E\big[\I_{\RG_{\alpha}}\cdot\|(\hat{F}-\hat{L})\hat{w}_{1}\|_{l^{2}}^{2}\big] & =\sum_{j=1}^{J}\E\big[\I_{\RG_{\alpha}}\cdot|\hat{f}(\hat{w}_{1},\psi_{j})-\hat{l}(\hat{w}_{1},\psi_{j})|^{2}\big]\\
 & \lesssim J^{1/2}\Delta^{1-6\alpha}+J\Delta^{4/3-8\alpha}\lesssim\Delta^{5/6-6\alpha}.
\end{align*}

\emph{Step 2. }$\RG_{2}$ is a high probability event. Indeed, inequality
(\ref{eq:Mean error bound F - L on w}) and Lemma \ref{lem:error between eigenvectors u and w}
imply
\[
\E[\I_{\RG_{\alpha}}\cdot\|\hat{w}_{1}-\hat{u}_{1}\|_{l^{2}}^{2}]^{1/2}\lesssim\Delta^{1/12-3\alpha}.
\]
Hence, by Markov's inequality,
\[
\P(\L\setminus\RG_{2})\leq2\epsilon+\Delta^{-1/7+6\alpha}\E[\I_{\RG_{\alpha}}\cdot\|\hat{w}_{1}-\hat{u}_{1}\|_{l^{2}}^{2}]\leq2\epsilon+C_{\alpha}\Delta^{1/6-1/7}\leq3\epsilon
\]
for $\Delta$ sufficiently small.

\emph{Step 3. O}n the event $\RG_{2}$ holds
\[
\max_{i=1,...,J}|\hat{w}_{1,i}-\hat{u}_{1,i}|^{2}\leq\sum_{i=1}^{J}|\hat{w}_{1,i}-\hat{u}_{1,i}|^{2}=\|\hat{w}_{1}-\hat{u}_{1}\|_{l^{2}}^{2}\lesssim\Delta^{1/7-6\alpha}.
\]
Since $\alpha<1/42$ the eigenvector $(\hat{u}_{1,j})$ inherits the
uniform bounds of the eigenvector $(\hat{w}_{1,j})$. In particular,
for any $j=\left\lfloor aJ\right\rfloor -1,...,\left\lceil bJ\right\rceil +1,$
we have
\[
\hat{u}_{1,j}\sim1.
\]
Moreover, since for any $j=1,...,J$ holds $\hat{w}_{1,j}>0$, we
deduce that
\[
\sum_{j=1}^{J}\hat{u}_{1,j}^{2}\I(\hat{u}_{1,j}<0)\leq\|\hat{w}_{1}-\hat{u}_{1}\|_{l^{2}}^{2}\lesssim1.
\]
Finally, note that on the event $\RG_{2}$ the eigenvalue $\hat{\gamma}_{1}\sim1$
since on $\RG_{\alpha}$, by (\ref{eq:First eigenvalues error bound}),
holds $|\hat{\lambda}_{1}^{-1}-\hat{\gamma}_{1}^{-1}|\lesssim\Delta^{1/6-3\alpha}\lesssim1$
and $\hat{\lambda}_{1}^{-1}\sim1$ by Lemma \ref{lem:Uniform spectral gap of B_f}. 
\end{proof}

\subsection{\label{sub:Proof of HF error bound}Proof of Theorem \ref{thm:High Frequency Error}}

As announced in Section \ref{sub:Connection to FZ estimator}, we
will bound the approximation error of the spectral estimator and the
time symmetric Florens-Zmirou estimator by the difference of forms
$\hat{f}$ and $\hat{l}$. 
\begin{lem}
\label{lem: sigma_S minus sigma  bounded by forms a and f}On the
high probability event $\RG_{2}$ from Proposition \ref{prop:Properties of the eigenfunction u}
holds
\[
\|\tilde{\sigma}_{S}^{2}-\hat{\sigma}_{FZ}^{2}\|_{L^{1}([a,b])}\lesssim\sum_{j=\left\lfloor aJ\right\rfloor }^{\left\lceil bJ\right\rceil }|\hat{l}(\hat{u}_{1},\psi_{j})-\hat{f(}\hat{u}_{1},\psi_{j})|.
\]
\end{lem}
\begin{proof}
From representations (\ref{eq:Spectral estimator in HF form without eigenvalue})
and (\ref{eq:Florens-Zmirou estimator in spectral form}) follows
that
\[
\|\tilde{\sigma}_{S}^{2}-\hat{\sigma}_{FZ}^{2}\|_{L^{1}([a,b])}=\frac{1}{J}\sum_{j=\left\lfloor aJ\right\rfloor }^{\left\lceil bJ\right\rceil }|\tilde{\sigma}_{S,j}^{2}-\hat{\sigma}_{FZ,j}^{2}|\lesssim\frac{1}{J}\sum_{j=\left\lfloor aJ\right\rfloor }^{\left\lceil bJ\right\rceil }\frac{|\hat{l}(\hat{u}_{1},\psi_{j})-\hat{f(}\hat{u}_{1},\psi_{j})|}{\hat{u}_{1,j}\int_{\frac{j-1}{J}}^{\frac{j}{J}}\hat{\mu}_{N}(dx)}.
\]
By Proposition \ref{prop:Properties of the eigenfunction u}, for
$j=\left\lfloor aJ\right\rfloor -1\leq j\leq\left\lceil bJ\right\rceil +1$,
we have $\hat{u}_{1,j}\sim1$. Since, by Lemma \ref{lem:mass of empirical measure on Ik},
$J\int_{\frac{j-1}{J}}^{\frac{j}{J}}\hat{\mu}_{N}(dx)\sim1$, we conclude
that the claim holds.\end{proof}
\begin{prop}
\label{prop:Mean error between f and a on a general function}For
every function $v\in V_{J}^{0}$ and any $j=1,...,J$ we have
\[
\E\big[\I_{\RG_{1}}\cdot|\hat{f}(v,\psi_{j})-\hat{l}(v,\psi_{j})|^{2}\big]^{\frac{1}{2}}\lesssim(|v_{j-1}|^{2}+|v_{j}|^{2}+|v_{j+1}|^{2})^{\frac{1}{2}}\Delta^{\frac{1}{2}},
\]
where $v$ corresponds to the vector $(v_{j})_{j=1,...,J}$ and $v_{0},v_{J+1}=0.$\end{prop}
\begin{proof}
First, note that since for $i\neq j$ holds $\hat{f}(\psi_{i},\psi_{j})=0$
we have $\hat{f}(v,\psi_{j})=v_{j}\hat{f}(\psi_{j},\psi_{j})$. Moreover,
on the event $\RG_{1}$, for $\Delta$ sufficiently small, the increments
of the process $X$ are smaller than $J^{-1}$. Hence, for $|i-j|>1$,
holds $\hat{l}(\psi_{i},\psi_{j})=0$. Linearity implies
\begin{equation}
\hat{l}(v,\psi_{j})=v_{j-1}\hat{l}(\psi_{j-1},\psi_{j})+v_{j}\hat{l}(\psi_{j},\psi_{j})+v_{j+1}\hat{l}(\psi_{j+1},\psi_{j}).\label{eq:tridiagonal structure of l}
\end{equation}
Consequently, it is sufficient to show that
\begin{align}
\E\big[|\hat{l}(\psi_{j-1},\psi_{j})|^{2}\big]^{\frac{1}{2}} & +\E\big[|\hat{f}(\psi_{j},\psi_{j})-\hat{l}(\psi_{j},\psi_{j})|^{2}\big]^{\frac{1}{2}}+\nonumber \\
 & \qquad+\E\big[|\hat{l}(\psi_{j},\psi_{j-1})|^{2}\big]^{\frac{1}{2}}\lesssim\Delta^{\frac{1}{2}}.\label{eq:Diff f_hat  l_hat diag}
\end{align}
Decomposing the terms above like in Lemma \ref{lem:Norm of F-L},
we obtain that (\ref{eq:Diff f_hat  l_hat diag}) follows from Theorem
\ref{thm:Bound on mean crossings}. 
\end{proof}
We are now able to prove the suboptimal rate $\Delta^{1/6}$ for the
root mean squared $L^{2}([a,b])$ error of the spectral estimator
$\tilde{\sigma}_{S}$. 
\begin{prop}
\label{prop:Suboptimal rate for spectral and error between u and unit}For
every $\epsilon>0$ and $\Delta$ sufficiently small, there exists
an event $\RG_{3}=\RG_{3}(\epsilon)\subseteq\RG_{2}$, with $\P(\L\setminus\RG_{3})\leq\epsilon$,
such that for every $x\in(a,b)$
\begin{equation}
\E\big[\I_{\RG_{3}}\cdot|\tilde{\sigma}_{S}^{2}(x)-\sigma^{2}(x)|^{2}\big]^{\frac{1}{2}}\lesssim\Delta^{\frac{1}{6}}.\label{eq:Suboptimal rate for spectral}
\end{equation}
Furthermore, on $\RG_{3},$ for every $\left\lfloor aJ\right\rfloor \leq j\leq\left\lceil bJ\right\rceil $
we have
\begin{align}
 & \tilde{\sigma}_{S,j}^{2}\sim1,\label{eq:Uniform bounds on sigma_S}\\
 & \E\Big[\I_{\RG_{3}}\cdot\Big|\frac{\hat{u}_{1,j\pm1}}{\hat{u}_{1,j}}-1\Big|^{2}\Big]^{\frac{1}{2}}\lesssim\Delta^{\frac{1}{6}}.\label{eq:Mean error ratios of u and I}
\end{align}
\end{prop}
\begin{rem}
\label{rem:Regularity of the eigenfunction}Given the uniform lower
bound on the derivative $\hat{u}_{1,j}$, and since $\Delta^{1/6}\sim J^{-1/2}$,
inequality (\ref{eq:Mean error ratios of u and I}) can be reformulated
as
\[
\E\Big[\I_{\RG_{3}}\cdot\Big|\frac{\hat{u}_{1}'(\frac{j}{J}\pm\frac{1}{J})-\hat{u}_{1}'(\frac{j}{J})}{J^{-1/2}}\Big|^{2}\Big]^{\frac{1}{2}}\lesssim1.
\]
By means of Markov's inequality the latter can be interpreted as almost
$1/2-$Hölder regularity of $\hat{u}_{1}'$. In that sense Proposition
\ref{prop:Suboptimal rate for spectral and error between u and unit}
is a discrete time equivalent of Proposition \ref{prop:Properties of the eigenproblem for sigma and mu},
which states that the derivatives of the eigenfunctions inherit the
regularity of the design density, in the high-frequency case the regularity
of the local time.\end{rem}
\begin{proof}[Proof of Proposition \ref{prop:Suboptimal rate for spectral and error between u and unit}]
Fix $\epsilon>0.$ Let $\RG_{2}$ be the high probability event introduced
in Proposition \ref{prop:Properties of the eigenfunction u}. On $\RG_{2}$,
we choose the eigenfunction $\hat{u}_{1}$ s.t.
\begin{equation}
\sum_{j=1}^{J}\hat{u}_{1,j}^{2}=J\quad\text{and}\quad\hat{u}_{1,j}\sim1\quad\text{for every}\quad\left\lfloor aJ\right\rfloor -1\leq j\leq\left\lceil bJ\right\rceil +1.\label{eq:Bounds and norm of u1_hat}
\end{equation}

\emph{Step 1. Proof of (\ref{eq:Uniform bounds on sigma_S}).} On
the event $\RG_{1}$, for $\Delta$ sufficiently small, using the
representation (\ref{eq:Spectral estimator in HF form without eigenvalue})
together with (\ref{eq:tridiagonal structure of l}) and (\ref{eq:Bounds and norm of u1_hat})
we obtain that
\begin{equation}
\frac{\hat{l}(\psi_{j},\psi_{j})}{\int_{\frac{j-1}{J}}^{\frac{j}{J}}\hat{\mu}_{N}(dx)}\lesssim\tilde{\sigma}_{S,j}^{2}\lesssim\frac{\hat{l}(\psi_{j-1}+\psi_{j}+\psi_{j+1},\psi_{j})}{\int_{\frac{j-1}{J}}^{\frac{j}{J}}\hat{\mu}_{N}(dx)}\label{eq:Bounds on spectral estimator}
\end{equation}
holds for every $\left\lfloor aJ\right\rfloor \leq j\leq\left\lceil bJ\right\rceil $.
Since
\[
\hat{l}(\psi_{j-1}+\psi_{j}+\psi_{j+1},\psi_{j})\lesssim\sum_{n=0}^{N-1}(\I_{j}(X_{n\Delta})+\I_{j}(X_{(n+1)\Delta}))(X_{(n+1)\Delta}-X_{n\Delta})^{2},
\]
we deduce that $\tilde{\sigma}_{S,j}^{2}\lesssim\hat{\sigma}_{FZ,j}^{2}$.
Furthermore, since on $\RG_{2}$ holds
\begin{equation}
\hat{l}(\psi_{j},\psi_{j})\geq\frac{1}{2}\sum_{n=0}^{N-1}\I({\textstyle \frac{j-1}{J}}+\Delta^{5/11}\leq X_{n\Delta}\leq{\textstyle \frac{j-1}{J}}-\Delta^{5/11})(X_{(n+1)\Delta}-X_{n\Delta})^{2},\label{eq:lower bound on form l}
\end{equation}
the spectral estimator can be bounded from below by a time symmetric
Florens-Zmirou estimator with bandwidth $\frac{1}{2}\Delta^{1/3}-\Delta^{5/11}\sim\Delta^{1/3}$.
Arguing as in Corollary \ref{cor:Uniform FZ error}, we deduce that
there exists a high probability event $\RG_{3,1}$, such that on $\RG_{3,1}$,
bound $\tilde{\sigma}_{S}^{2}(x)\gtrsim1$ holds for any $x\in(a,b)$.
Set
\[
\RG_{3}=\RG_{2}\cap\RG_{3,1}.
\]

\emph{Step 2. Proof of (\ref{eq:Suboptimal rate for spectral}).}
Fix $x\in(a,b)$ and chose $j$ s.t. $\frac{j-1}{J}\leq x<\frac{j}{J}$.
Representations (\ref{eq:Spectral estimator in HF form without eigenvalue})
and (\ref{eq:Florens-Zmirou estimator in spectral form}), together
with Lemma \ref{lem:mass of empirical measure on Ik}, imply
\[
|\tilde{\sigma}_{S,j}^{2}-\hat{\sigma}_{FZ,j}^{2}|\lesssim\Delta^{-1/3}|\hat{l}(\hat{u}_{1},\psi_{j})-\hat{f(}\hat{u}_{1},\psi_{j})|.
\]
Hence, from Proposition \ref{prop:Mean error between f and a on a general function}
and (\ref{eq:Bounds and norm of u1_hat}) follows that
\[
\E\big[\I_{\RG_{3}}\cdot|\tilde{\sigma}_{S}^{2}(x)-\hat{\sigma}_{FZ}^{2}(x)|^{2}\big]^{\frac{1}{2}}\lesssim\Delta^{1/6}.
\]
By Theorem \ref{thm:FlorensZmirou L2 error} and Hölder regularity
of $\sigma^{2}$
\[
\E\big[\I_{\RG_{1}}\cdot\|\sigma^{2}-\hat{\sigma}_{FZ}^{2}\|_{\infty}^{2}\big]^{\frac{1}{2}}\lesssim\Delta^{1/6}.
\]
By the triangle inequality we conclude that (\ref{eq:Suboptimal rate for spectral})
holds.

\emph{Step 3. Proof of (\ref{eq:Mean error ratios of u and I}).}
Set $\left\lfloor aJ\right\rfloor \leq j\leq\left\lceil bJ\right\rceil .$
We will only prove
\begin{equation}
\E\Big[\I_{\RG_{3}}\cdot\Big|\frac{\hat{u}_{1,j+1}}{\hat{u}_{1,j}}-1\Big|^{2}\Big]^{\frac{1}{2}}\lesssim\Delta^{\frac{1}{6}},\label{eq:bound on derivative ration for j+1}
\end{equation}
as the symmetric bound on the second moment of $\I_{\RG_{3}}\cdot\Big|\frac{\hat{u}_{1,j-1}}{\hat{u}_{1,j}}-1\Big|$
can be obtained analogously. The general idea of the proof is similar
to the proof of (\ref{eq:error between w and unit vector}) in Proposition
\ref{prop:Bounds on the vector w_hat}. First, we will show that (\ref{eq:bound on derivative ration for j+1})
follows from
\begin{equation}
\E\Big[\I_{\RG_{3}}\cdot\Big|\frac{\hat{u}_{1,j+1}\tilde{\sigma}_{S,j+1}^{2}\int_{\frac{j}{J}}^{\frac{j+1}{J}}\hat{\mu}_{N}(dx)}{\hat{u}_{1,j}\tilde{\sigma}_{S,j}^{2}\int_{\frac{j-1}{J}}^{\frac{j}{J}}\hat{\mu}_{N}(dx)}-1\Big|^{2}\Big]^{\frac{1}{2}}\lesssim\Delta^{\frac{1}{6}}.\label{eq:mean error with sigma mu}
\end{equation}
To that purpose, by the triangle inequality and since on $\RG_{3}$
the derivatives $\hat{u}_{1,j},\hat{u}_{1,j+1}\sim1$, we have to
argue that
\begin{equation}
\E\Big[\I_{\RG_{3}}\cdot\Big|\frac{\tilde{\sigma}_{S,j+1}^{2}\int_{\frac{j}{J}}^{\frac{j+1}{J}}\hat{\mu}_{N}(dx)}{\tilde{\sigma}_{S,j}^{2}\int_{\frac{j-1}{J}}^{\frac{j}{J}}\hat{\mu}_{N}(dx)}-1\Big|^{2}\Big]^{\frac{1}{2}}\lesssim\Delta^{1/6}.\label{eq:Partial Step a}
\end{equation}
\emph{Step 3.1. Proof of (\ref{eq:Partial Step a}).} By Lemma \ref{lem:mass of empirical measure on Ik}
holds $J\int_{\frac{j-1}{J}}^{\frac{j}{J}}\hat{\mu}_{N}(dx),J\int_{\frac{j}{J}}^{\frac{j+1}{J}}\hat{\mu}_{N}(dx)\sim1$.
We defined above the event $\RG_{3}$ so that $\tilde{\sigma}_{S,j}^{2},\tilde{\sigma}_{S,j+1}^{2}\sim1.$
Hence, to prove (\ref{eq:Partial Step a}), it suffices to show
\begin{align}
 & \E\Big[\I_{\RG_{3}}\cdot\big|\tilde{\sigma}_{S,j+1}^{2}-\tilde{\sigma}_{S,j}^{2}\big|^{2}\Big]^{\frac{1}{2}}\lesssim\Delta^{1/6}\label{eq:sigma}\\
 & \E\Big[\I_{\RG_{3}}\cdot\big|\int_{\frac{j}{J}}^{\frac{j+1}{J}}\hat{\mu}_{N}(dx)-\int_{\frac{j-1}{J}}^{\frac{j}{J}}\hat{\mu}_{N}(dx)\big|^{2}\Big]^{\frac{1}{2}}\lesssim\Delta^{1/2}.\label{eq:mu}
\end{align}
(\ref{eq:sigma}) follows from (\ref{eq:Suboptimal rate for spectral})
and $1/2$ Hölder regularity of $\sigma^{2}$. Indeed
\begin{align*}
 & \E\Big[\I_{\RG_{3}}\cdot\big|\tilde{\sigma}_{S,j+1}^{2}-\tilde{\sigma}_{S,j}^{2}\big|^{2}\Big]^{\frac{1}{2}}\lesssim\E\Big[\I_{\RG_{3}}\cdot\big|\tilde{\sigma}_{S,j+1}^{2}-\sigma^{2}({\textstyle \frac{j+1/2}{J}})\big|^{2}\Big]^{\frac{1}{2}}+\\
 & +\E\Big[\I_{\RG_{3}}\cdot\big|\sigma^{2}({\textstyle \frac{j+1/2}{J}})-\sigma^{2}({\textstyle \frac{j-1/2}{J}})\big|^{2}\Big]^{\frac{1}{2}}+\E\Big[\I_{\RG_{3}}\cdot\big|\sigma^{2}({\textstyle \frac{j-1/2}{J}})-\tilde{\sigma}_{S,j}^{2}\big|^{2}\Big]^{\frac{1}{2}}\\
 & \qquad\lesssim\Delta^{1/6}.
\end{align*}
To prove (\ref{eq:mu}) let
\begin{align*}
 & \Big|\int_{\frac{j}{J}}^{\frac{j+1}{J}}\hat{\mu}_{N}(dx)-\int_{\frac{j-1}{J}}^{\frac{j}{J}}\hat{\mu}_{N}(dx)\Big|\leq\Big|\int_{\frac{j}{J}}^{\frac{j+1}{J}}\hat{\mu}_{N}(dx)-\int_{\frac{j}{J}}^{\frac{j+1}{J}}\mu_{1}(x)dx\Big|+\\
 & +\Big|\int_{\frac{j}{J}}^{\frac{j+1}{J}}\mu_{1}(x)dx-\int_{\frac{j-1}{J}}^{\frac{j}{J}}\mu_{1}(x)dx\Big|+\Big|\int_{\frac{j-1}{J}}^{\frac{j}{J}}\mu_{1}(x)dx-\int_{\frac{j-1}{J}}^{\frac{j}{J}}\hat{\mu}_{N}(dx)\Big|\\
 & \qquad:=E_{1}+E_{2}+E_{3}.
\end{align*}
By {[}Supplement A, Theorem 11{]} we have
\[
\E[E_{1}^{2}+E_{3}^{2}]^{\frac{1}{2}}\lesssim\Delta^{2/3},
\]
while the Cauchy-Schwarz inequality, together with {[}Supplement A,
Theorem 8{]} yield
\begin{alignat*}{1}
\E[E_{2}^{2}]^{\frac{1}{2}} & =\E\Big[\Big|\int_{0}^{J^{-1}}\mu_{1}({\textstyle \frac{j}{J}}+x)-\mu_{1}({\textstyle \frac{j-1}{J}}+x)dx\Big|^{2}\Big]^{\frac{1}{2}}\\
 & \leq\Big[\frac{1}{J}\int_{0}^{J^{-1}}\E\big[|\mu_{1}({\textstyle \frac{j}{J}}+x)-\mu_{1}({\textstyle \frac{j-1}{J}}+x)|^{2}\big]dx\Big]^{\frac{1}{2}}\lesssim\Delta^{\frac{1}{2}}.
\end{alignat*}

\emph{Step 3.2. Proof of (\ref{eq:mean error with sigma mu}).} The
representation (\ref{eq:Spectral estimator in HF form without eigenvalue}),
together with the eigenpair property of $(\hat{\gamma}_{1},\hat{u}_{1})$,
imply that
\[
\frac{\hat{u}_{1,j+1}\tilde{\sigma}_{S,j+1}^{2}\int_{\frac{j}{J}}^{\frac{j+1}{J}}\hat{\mu}_{N}(dx)}{\hat{u}_{1,j}\tilde{\sigma}_{S,j}^{2}\int_{\frac{j-1}{J}}^{\frac{j}{J}}\hat{\mu}_{N}(dx)}=\frac{\hat{l}(\hat{u}_{1},\psi_{j+1})}{\hat{l}(\hat{u}_{1},\psi_{j})}=\frac{\hat{g}(\hat{u}_{1},\psi_{j+1})}{\hat{g}(\hat{u}_{1},\psi_{j})}.
\]
In what follows we want to apply methods from the Perron-Frobenius
theory for nonnegative matrices. To that purpose recall the definition
of matrix $\hat{M}$ from Section \ref{sub:Properties of the eigenfunction u}
Eq. (\ref{eq:Formula for M}). We have
\[
\frac{\hat{g}(\hat{u}_{1},\psi_{j+1})}{\hat{g}(\hat{u}_{1},\psi_{j})}=\frac{\sum_{m=1}^{J}\hat{M}_{m,j+1}\hat{u}_{1,m}}{\sum_{m=1}^{J}\hat{M}_{m,j}\hat{u}_{1,m}}.
\]
To bound the above ratio we would like to proceed as in the proof
of inequality (\ref{eq:error between w and unit vector}) in Proposition
\ref{prop:Bounds on the vector w_hat}. Unfortunately, we can't, as
we don't know if the vector of derivatives $(\hat{u}_{1,j})$ is positive.
Still, using the inequality (\ref{eq:Frobenius bounds inequality})
and arguing as in the proof of (\ref{eq:error between w and unit vector}),
we obtain that
\[
\Big|\frac{\sum_{m=1}^{J}\hat{M}_{m,j+1}\hat{u}_{1,m}\I(\hat{u}_{1,m}>0)}{\sum_{m=1}^{J}\hat{M}_{m,j}\hat{u}_{1,m}\I(\hat{u}_{1,m}>0)}-1\Big|\lesssim\Delta^{1/6}.
\]
To finish the proof we need to show that the possible error due to
the negative derivative terms is small enough. On the event $\RG_{2}$
we have
\[
\hat{g}(\hat{u}_{1},\psi_{j})=\hat{\gamma}_{1}^{-1}\hat{l}(\hat{u}_{1},\psi_{j})\sim\hat{l}(\hat{u}_{1},\psi_{j})\geq\hat{u}_{1,j}\hat{l}(\psi_{j},\psi_{j})\sim\hat{l}(\psi_{j},\psi_{j}).
\]
Furthermore, on the event $\RG_{3}$ we have $\hat{l}(\psi_{j},\psi_{j})\gtrsim\int_{\frac{j-1}{J}}^{\frac{j}{J}}\hat{\mu}_{N}(dx)$;
indeed we defined $\RG_{3}$ such that the left hand side of (\ref{eq:Bounds on spectral estimator})
has a uniform lower bound. Thus, by Lemma \ref{lem:mass of empirical measure on Ik}
\[
\sum_{m=1}^{J}\hat{M}_{m,j}\hat{u}_{1,m}=\hat{g}(\hat{u}_{1},\psi_{j})\gtrsim\int_{\frac{j-1}{J}}^{\frac{j}{J}}\hat{\mu}_{N}(dx)\gtrsim\Delta^{1/3}.
\]
Consequently, we need to show that
\[
\sum_{m=1}^{J}(\hat{M}_{m,j}+\hat{M}_{m,j+1})|\hat{u}_{1,m}|\I(\hat{u}_{1,m}\leq0)\lesssim\Delta^{\frac{1}{2}}.
\]
From (\ref{eq:Formula for M}) follows $\hat{M}_{i,j}\lesssim J^{-2}$.
By the Cauchy-Schwarz inequality and Proposition \ref{prop:Properties of the eigenfunction u}
\begin{align*}
\sum_{m=1}^{J}(\hat{M}_{m,j} & +\hat{M}_{m,j+1})|\hat{u}_{1,m}|\I(\hat{u}_{1,m}\leq0)\lesssim\\
 & \lesssim J^{-3/2}\Big(\sum_{m=1}^{J}|\hat{u}_{1,m}|^{2}\I(\hat{u}_{1,m}\leq0)\Big)^{\frac{1}{2}}\lesssim\Delta^{\frac{1}{2}}.\qedhere
\end{align*}

\end{proof}
To obtain the suboptimal rate $\Delta^{1/6}$ we only used uniform
bounds on the derivatives vector $(\hat{u}_{1,j})_{j}$ together with
the general error bound from Proposition \ref{prop:Mean error between f and a on a general function}.
Having established the regularity of the eigenfunction $\hat{u}_{1}$,
we are now able to argue that the error $\E\big[\I_{\RG_{3}}\cdot|\hat{l}(\hat{u}_{1},\psi_{j})-\hat{f}(\hat{u}_{1},\psi_{j})|\big]$
is at most of order $\Delta^{2/3}.$ 
\begin{lem}
\label{lem:Error of l and f on the unit vector}Denote $I(x)=x-c_{0}$,
with $c_{0}$ such that $I\in V_{J}^{0}.$ For $\Delta$ sufficiently
small, for every $j=1,...,J$, it holds
\begin{equation}
\E\Big[\I_{\RG_{1}}\cdot|\hat{f}(I,\psi_{j})-\hat{l}(I,\psi_{j})|^{2}\Big]^{\frac{1}{2}}\lesssim\Delta^{2/3}.\label{eq:Error of f and l on I}
\end{equation}
\end{lem}
\begin{proof}
We will reduce (\ref{eq:Error of f and l on I}) to the term bounded
in Theorem \ref{thm:Crossings canceling}. By definition of the forms
$\hat{l}$, $\hat{f}$ and the representation (\ref{eq:Florens-Zmirou estimator in spectral form})
it holds
\begin{eqnarray*}
\hat{l}(I,\psi_{j}) & = & \frac{1}{2}\sum_{n=0}^{N-1}(X_{(n+1)\Delta}-X_{n\Delta})(\psi_{j}(X_{(n+1)\Delta})-\psi_{j}(X_{n\Delta})),\\
\hat{f}(I,\psi_{j}) & = & \frac{1}{4}\sum_{n=0}^{N-1}(\I_{j}(X_{n\Delta})+\I_{j}(X_{(n+1)\Delta}))(X_{(n+1)\Delta}-X_{n\Delta})^{2}.
\end{eqnarray*}
We will analyze the error contribution of a single summand. When $X_{n\Delta},X_{(n+1)\Delta}$
$\in[\frac{j-1}{J},\frac{j}{J}]$ both forms contribute by $\frac{1}{2}(X_{(n+1)\Delta}-X_{n\Delta})^{2}$,
hence cancel perfectly. When both $X_{n\Delta},X_{(n+1)\Delta}\notin[\frac{j-1}{J},\frac{j}{J}]$
neither of the forms contribute. Since on $\RG_{1}$, for $\Delta$
sufficiently small, the increment $|X_{(n+1)\Delta}-X_{n\Delta}|\leq1/J$
we deduce that the overall error $|\hat{f}(I,\psi_{j})-\hat{l}(I,\psi_{j})|$
is due only to summands with the increment $X_{n\Delta},X_{(n+1)\Delta}$
crossing the boundary of $[\frac{j-1}{J},\frac{j}{J}]$. In such case
the form $\hat{f}$ contributes by $\frac{1}{4}(X_{(n+1)\Delta}-X_{n\Delta})^{2}$,
while $\hat{l}$ by $\frac{1}{2}(X_{(n+1)\Delta}-X_{n\Delta})\beta$,
where
\[
\beta=\operatorname{sgn}(X_{(n+1)\Delta}-X_{n\Delta})\cdot\operatorname{length}\big([X_{n\Delta},X_{(n+1)\Delta}]\cap\big[{\textstyle \frac{j-1}{J}},{\textstyle \frac{j}{J}}\big]\big).
\]
Let $\gamma=X_{(n+1)\Delta}-X_{n\Delta}-\beta$. The contribution
of a single boundary crossing summand equals
\[
\frac{1}{4}(X_{(n+1)\Delta}-X_{n\Delta})^{2}-\frac{1}{2}(X_{(n+1)\Delta}-X_{n\Delta})\beta=\frac{1}{4}(\beta+\gamma)(\gamma-\beta)=\frac{\gamma^{2}-\beta^{2}}{4}.
\]
Considering all four possible crossing configurations, we obtain that
\begin{align*}
\hat{f}(I,\psi_{j})-\hat{l}(I,\psi_{j}) & =\sum_{n=0}^{N-1}\big(\I_{(\frac{j}{J},1]}(X_{(n+1)\Delta})-\I_{(\frac{j}{J},1]}(X_{n\Delta})\big)\cdot\\
 & \qquad\qquad\qquad\cdot\big((X_{(n+1)\Delta}-{\textstyle \frac{j}{J}})^{2}-(X_{n\Delta}-{\textstyle \frac{j}{J}})^{2}\big)\\
 & +\sum_{n=0}^{N-1}\big(\I_{(\frac{j-1}{J},1]}(X_{(n+1)\Delta})-\I_{(\frac{j-1}{J},1]}(X_{n\Delta})\big)\cdot\\
 & \qquad\qquad\qquad\cdot\big((X_{(n+1)\Delta}-{\textstyle \frac{j-1}{J}})^{2}-(X_{n\Delta}-{\textstyle \frac{j-1}{J}})^{2}\big).
\end{align*}
Thus, (\ref{eq:Error of f and l on I}) indeed follows from Theorem
\ref{thm:Crossings canceling}. 
\end{proof}

\begin{proof}[Proof of Theorem \ref{thm:High Frequency Error}]
Set $\epsilon>0$. Let $\RG_{3}$ be the high probability event introduced
in Proposition \ref{prop:Suboptimal rate for spectral and error between u and unit}.
In view of Remark \ref{rem:Spectral estimator without eigenvalue ratio},
it is enough to prove the claim for the estimator $\tilde{\sigma}_{S}^{2}.$
By Lemma \ref{lem: sigma_S minus sigma  bounded by forms a and f}
and since $J\sim\Delta^{-1/3}$, it is sufficient to show that for
any $\left\lfloor aJ\right\rfloor \leq j\leq\left\lceil bJ\right\rceil $
holds
\[
\E\big[\I_{R_{3}}\cdot|\hat{l}(\hat{u}_{1},\psi_{j})-\hat{f}(\hat{u}_{1},\psi_{j})|\big]\lesssim\Delta^{2/3}.
\]
By Definition \ref{def:Form f_hat} holds $\hat{f}(\hat{u}_{1},\psi_{j})=\hat{u}_{1,j}\hat{f}(\psi_{j},\psi_{j})=\hat{u}_{1,j}\hat{f}(I,\psi_{j})$.
Since on the event $\RG_{3}$, for $\Delta$ sufficiently small, the
increments $|X_{(n+1)\Delta}-X_{n\Delta}|\leq J^{-1}$, we have
\begin{eqnarray*}
\hat{l}(\hat{u}_{1},\psi_{j}) & = & \hat{u}_{1,j-1}\hat{l}(\psi_{j},\psi_{j-1})+\hat{u}_{1,j}\hat{l}(\psi_{j},\psi_{j})+\hat{u}_{1,j+1}\hat{l}(\psi_{j},\psi_{j+1}),\\
\hat{l}(I,\psi_{j}) & = & \hat{l}(\psi_{j},\psi_{j-1})+\hat{l}(\psi_{j},\psi_{j})+\hat{l}(\psi_{j},\psi_{j+1}).
\end{eqnarray*}
Consequently, since by Proposition \ref{prop:Properties of the eigenfunction u}
$\hat{u}_{1,j}\sim1$, we deduce that
\begin{align}
\hat{l}(\hat{u}_{1},\psi_{j})-\hat{f}(\hat{u}_{1},\psi_{j})\sim & \hat{l}(\psi_{j},\psi_{j-1})\Big(\frac{\hat{u}_{1,j-1}}{\hat{u}_{1,j}}-1\Big)+\hat{l}(I,\psi_{j})-\hat{f}(I,\psi_{j})+\nonumber \\
 & \qquad+\hat{l}(\psi_{j},\psi_{j+1})\Big(\frac{\hat{u}_{1,j+1}}{\hat{u}_{1,j}}-1\Big).\label{eq:Proof of HF error bounds}
\end{align}
By the Cauchy-Schwarz inequality together with Proposition \ref{prop:Mean error between f and a on a general function}
and the inequality (\ref{eq:Mean error ratios of u and I}) we can
uniformly bound the mean absolute value of the first and third term
by $\Delta^{2/3}.$ Since $\RG_{3}\subset\RG_{1}$ the mean absolute
value of the second term is bounded in Lemma \ref{lem:Error of l and f on the unit vector}. 
\end{proof}

\subsection{\label{sub:Mean crossing bounds}Technical results}

We devote this chapter to the proof of two technical results that
provide us with control over, properly rescaled, mean number of crossing
of a given level $\alpha.$ 
\begin{defn}
For $\alpha\in(0,1)$ and $n=0,...,N-1$ define
\[
\chi(n,\alpha)=\I_{[0,\alpha)}(X_{(n+1)\Delta})-\I_{[0,\alpha)}(X_{n\Delta}).
\]
The random variable $\chi$ codifies the event of the increment $X_{n\Delta},X_{(n+1)\Delta}$
crossing the level $\alpha.$ The sign of $\chi$ contains information
about the direction of the crossing. Since
\[
|\chi(n,\alpha)|\leq\I(|X_{n\Delta}-\alpha|\leq\mc),
\]
arguing as in the proof of Lemma (\ref{lem:mass of empirical measure on Ik})
we can show that
\[
\frac{1}{N}\sum_{n=0}^{N-1}|\chi(n,\alpha)|\leq4\mc\mu_{1}.
\]
Consequently, Theorem \ref{thm:Moments of the modulus of continuity}
implies that the mean number of crossings, rescaled by the sample
size, can be upper bounded by $\Delta^{1/2}\log(\Delta).$ Keeping
in mind that $(X_{(n+1)\Delta}-X_{n\Delta})^{2}$ is of the order
$\Delta=1/N,$ the next result is a refinement of the bound above.\end{defn}
\begin{thm}
\label{thm:Bound on mean crossings}For every $\alpha\in(0,1)$ we
have

\[
\E\Big[\Big(\sum_{n=0}^{N-1}|\chi(n,\alpha)|(X_{(n+1)\Delta}-X_{n\Delta})^{2}\Big)^{2}\Big]^{\frac{1}{2}}\lesssim\Delta^{1/2}.
\]
\end{thm}
\begin{proof}
Fix $\alpha\in(0,1).$ Since$|\chi(n,\alpha)|=1$ if and only if the
increment $(X_{n\Delta},X_{(n+1)\Delta})$ crosses the level $\alpha$,
the claim is equivalent to the inequalities:
\begin{eqnarray*}
\E\Big[\Big(\sum_{n=0}^{N-1}\I(X_{n\Delta}<\alpha)\I(X_{(n+1)\Delta}>\alpha)(X_{(n+1)\Delta}-X_{n\Delta})^{2}\Big)^{2}\Big]^{\frac{1}{2}} & \lesssim & \Delta^{1/2},\\
\E\Big[\Big(\sum_{n=0}^{N-1}\I(X_{n\Delta}>\alpha)\I(X_{(n+1)\Delta}<\alpha)(X_{(n+1)\Delta}-X_{n\Delta})^{2}\Big)^{2}\Big]^{\frac{1}{2}} & \lesssim & \Delta^{1/2}.
\end{eqnarray*}

Below, we only prove the first inequality. The second one can be obtained
in a similar way or by a time reversal argument. Denote 
\[
\eta_{n}=\I(X_{n\Delta}<\alpha)\I(X_{(n+1)\Delta}>\alpha)(X_{(n+1)\Delta}-X_{n\Delta})^{2}.
\]
We have
\[
\E\Big[\Big(\sum_{n=0}^{N-1}\I(X_{n\Delta}<\alpha)\I(X_{(n+1)\Delta}>\alpha)(X_{(n+1)\Delta}-X_{n\Delta})^{2}\Big)^{2}\Big]=\sum_{n=0}^{N-1}\E[\eta_{n}^{2}]+2\sum_{0\leq n<m}^{N-1}\E[\eta_{n}\eta_{m}].
\]
Denote by $p_{t}$ the transition kernel of the diffusion $X.$ Uniform
bounds on diffusion coefficients imply that
\begin{equation}
p_{t}(x,y)\leq M_{1}\frac{1}{\sqrt{t}}\exp\Big(-\frac{(x-y)^{2}}{M_{2}t}\Big),\label{eq:bound on transition kernel}
\end{equation}
with $M_{1},M_{2}$ positive constants uniform on $\Theta$, see \citep[Lemma 2]{Rozkosz:1992}.
From (\ref{eq:bound on transition kernel}) and the inequality \citep[Formula 7.1.13]{AbramowitzStegun:1972}:
\begin{equation}
\int_{x}^{\infty}e^{-z^{2}}dz\leq\frac{e^{-x^{2}}}{x+\sqrt{x^{2}+4/\pi}}\leq\frac{\sqrt{\pi}}{2}e^{-x^{2}},\label{eq:Uniform bound on gaussian tails}
\end{equation}
 follows that
\begin{eqnarray}
\int_{0}^{\alpha}\int_{\alpha}^{1}p_{\Delta}(x,y)(y-x)^{4}dydx & \lesssim & \int_{0}^{\alpha}\int_{\alpha}^{1}\frac{1}{\sqrt{\Delta}}e^{-\frac{(y-x)^{2}}{c\Delta}}(y-x)^{4}dydx\lesssim\Delta^{2}\int_{0}^{\alpha}\int_{\frac{\alpha-x}{\sqrt{\Delta c}}}^{\frac{1-x}{\sqrt{\Delta c}}}e^{-z^{2}}z^{4}dzdx\nonumber \\
 & \lesssim & \Delta^{2}\int_{0}^{\alpha}\int_{\frac{\alpha-x}{\sqrt{\Delta c}}}^{\frac{1-x}{\sqrt{\Delta c}}}e^{-\frac{z^{2}}{2}}dzdx\lesssim\Delta^{2}\int_{0}^{\alpha}e^{-\frac{(\alpha-x)^{2}}{2c\Delta}}dx\lesssim\Delta^{5/2}.\label{eq:trans kernel 4 moment}
\end{eqnarray}
Similarly
\begin{equation}
\int_{0}^{\alpha}\int_{\alpha}^{1}p_{\Delta}(x,y)(y-x)^{2}dydx\lesssim\Delta^{3/2}.\label{eq:trans kernel 2 moment}
\end{equation}
For simplicity we will use the stationarity of $X$, which is granted
by Assumption \ref{assu: Initial condition}. Using more elaborated
arguments the result could be obtained for an arbitrary initial condition.
By stationarity, for any $t$, the one dimensional margin $X_{t}$
is distributed with respect to the invariant measure $\mu(x)dx$.
Conditioning on $X_{n\Delta}$, from (\ref{eq:trans kernel 4 moment})
and uniform bounds on the density $\mu$ follows 
\[
\E\big[\eta_{n}^{2}\big]=\int_{0}^{\alpha}\int_{\alpha}^{1}p_{\Delta}(x,y)(y-x)^{4}dy\mu(x)dx\lesssim\Delta^{5/2}.
\]
Hence
\[
\sum_{n=0}^{N-1}\E[\eta_{n}^{2}]\lesssim N\Delta^{\frac{5}{2}}=\Delta^{\frac{3}{2}}.
\]
The Cauchy-Schwarz inequality implies 
\[
\sum_{n=0}^{N-2}\E[\eta_{n}\eta_{n+1}]\lesssim\sum_{n=0}^{N-2}\E[\eta_{n}^{2}]^{\frac{1}{2}}\E[\eta_{n+1}^{2}]^{\frac{1}{2}}\lesssim N\Delta^{\frac{5}{2}}\lesssim\Delta^{\frac{3}{2}}.
\]
Finally, using (\ref{eq:trans kernel 2 moment}), for $m>n+1$, we
obtain 
\begin{align*}
\E[\eta_{n}\eta_{m}] & =\int_{0}^{\alpha}\int_{\alpha}^{1}\int_{0}^{\alpha}\int_{\alpha}^{1}p_{\Delta}(x,y)(y-x)^{2}p_{(m-n-1)\Delta}(z,x)(z-w)^{2}p_{\Delta}(w,z)\mu(w)dydxdzdw\\
 & \lesssim\int_{0}^{\alpha}\int_{\alpha}^{1}p_{\Delta}(x,y)(y-x)^{2}dydx\frac{1}{\sqrt{(m-n-1)\Delta}}\int_{0}^{\alpha}\int_{\alpha}^{1}(z-w)^{2}p_{\Delta}(w,z)dzdw\\
 & \lesssim\Delta^{5/2}\frac{1}{\sqrt{m-n-1}}.
\end{align*}
Consequently
\begin{alignat*}{1}
\sum_{n=0}^{N-3}\sum_{m=n+2}^{N-1}\E[\eta_{n}\eta_{m}] & \lesssim\Delta^{5/2}\sum_{n=0}^{N-3}\sum_{k=1}^{N-n-2}\frac{1}{\sqrt{k}}\lesssim\Delta^{5/2}\sum_{n=0}^{N-3}\sqrt{N-n-2}\\
 & =\Delta^{5/2}\sum_{n=1}^{N-2}\sqrt{n}\lesssim\Delta^{5/2}N^{3/2}=\Delta.\tag*{{\qedhere}}
\end{alignat*}

\end{proof}
Note that the claim of Theorem \ref{thm:Bound on mean crossings}
still holds when we replace $(X_{(n+1)\Delta}-X_{n\Delta})^{2}$ by
$(X_{(n+1)\Delta}-\alpha)^{2}$ or $(X_{n\Delta}-\alpha)^{2}$. Next,
we show that, when considering the direction of the crossings, cancellations
occur that make the difference of $\sum_{n=0}^{N-1}\chi(n,\alpha)(X_{(n+1)\Delta}-\alpha)^{2}$
and $\sum_{n=0}^{N-1}\chi(n,\alpha)(X_{n\Delta}-\alpha)^{2}$ even
smaller. 
\begin{thm}
\label{thm:Crossings canceling}For any $\alpha\in[\frac{1}{J},1-\frac{1}{J}]$
we have 
\[
\E\Big[\I_{\RG_{1}}\cdot\Big|\sum_{n=0}^{N-1}\chi(n,\alpha)\big((X_{(n+1)\Delta}-\alpha)^{2}-(X_{n\Delta}-\alpha)^{2}\big)\Big|^{2}\Big]^{\frac{1}{2}}\lesssim\Delta^{2/3}.
\]

\end{thm}
Due to the sign of the terms the proof of the next theorem cannot
be done in a similar way as the previous result. In what follows we
show that on the event $\RG_{1}$ 
\[
\sum_{n=0}^{N-1}\chi(n,\alpha)\big((X_{(n+1)\Delta}-\alpha)^{2}-(X_{n\Delta}-\alpha)^{2}\big)=\int_{0}^{1}\I(X_{s}<\alpha)ds-\frac{1}{N}\sum_{n=0}^{N-1}(\I(X_{n\Delta}<\alpha)+R,
\]
where the remainder term is of the right order. Thus we are left with
showing that
\begin{equation}
\E\Big[\Big|\int_{0}^{1}\I(X_{s}<\alpha)ds-\frac{1}{N}\sum_{n=0}^{N-1}(\I(X_{n\Delta}<\alpha)\Big|^{2}\Big]^{\frac{1}{2}}\lesssim\Delta^{2/3}.\label{eq:estimating occupation time}
\end{equation}
Note that $\frac{1}{N}\sum_{n=0}^{N-1}(\I(X_{n\Delta}<\alpha)$ is
a Riemann type estimator of the occupation time of the interval $[0,\alpha)$.
The problem of establishing the rate of convergence was recently considered
in \citep{NgoOgawa:2011,KohatsuMakhloufNgo:2014}. Although obtained
results do not apply as they require higher smoothness of the coefficients,
they suggest an ever better rate $\Delta^{3/4}$. Indeed, in the case
of reflected diffusion with bounded coefficients, we can show that

\[
\E\Big[\Big|\int_{0}^{1}f(X_{s})ds-\frac{1}{N}\sum_{n=0}^{N-1}f(X_{n\Delta})\Big|^{2}\Big]^{\frac{1}{2}}\lesssim\Delta^{\frac{1+s}{2}}\|f\|_{H^{s}},
\]
for any cadlag function $f$ with Sobolev regularity $0\leq s\leq1$,
see \citep{ChorowskiAltmeyer:2016}.
\begin{proof}
Fix $\alpha\in[\frac{1}{J},1-\frac{1}{J}].$ On the event $\RG_{1}$,
whenever $\I_{[0,\alpha)}(X_{(n+1)\Delta})-\I_{[0,\alpha)}(X_{n\Delta})\neq0$
we must have $|X_{n\Delta}-\alpha|,|X_{(n+1)\Delta}-\alpha|\leq\mc<\Delta^{4/9}$.
Consider function $d:[0,1]\to\R$ given by
\[
d(x)=(x-\alpha)^{2}\I(|x-\alpha|\leq\Delta^{4/9}).
\]
We have 
\begin{align*}
 & \big(\I_{[0,\alpha)}(X_{(n+1)\Delta})-\I_{[0,\alpha)}(X_{n\Delta})\big)\big((X_{(n+1)\Delta}-\alpha)^{2}-(X_{n\Delta}-\alpha)^{2}\big)=\\
 & \hspace{3cm}=\big(\I_{[0,\alpha)}(X_{(n+1)\Delta})-\I_{[0,\alpha)}(X_{n\Delta})\big)\big(d(X_{(n+1)\Delta})-d(X_{n\Delta})\big).
\end{align*}
\emph{Step 1. }We will first show that 
\begin{equation}
\E\Big[\I_{\RG_{1}}\cdot\Big|\sum_{n=0}^{N-1}\I_{[0,\alpha)}(X_{n\Delta})\big(d(X_{(n+1)\Delta})-d(X_{n\Delta})\big)\Big|^{2}\Big]^{\frac{1}{2}}\lesssim\Delta^{2/3}.\label{eq:increment canceling Partial}
\end{equation}
Note that
\begin{eqnarray*}
d'(x) & = & 2(x-\alpha)\I(|x-\alpha|\leq\Delta^{4/9}),\\
\frac{1}{2}d''(x) & = & -\Delta^{4/9}\delta_{\{\alpha-\Delta^{4/9}\}}+\I(|x-\alpha|\leq\Delta^{4/9})-\Delta^{4/9}\delta_{\{\alpha+\Delta^{4/9}\}},
\end{eqnarray*}
where the second derivative must be understood in the distributional
sense. Since we fixed $\alpha$ separated from the boundaries, $d'(0)=d'(1)=0$
for $\Delta$ small enough. Denote by 
\[
L_{s,t}(x):=L_{t}(x)-L_{s}(x),
\]
the local time of the path fragment $(X_{u},s\leq u\leq t)$. From
the Itô-Tanaka formula \citep[Chapter VI, Theorem 1.5]{RevuzYor:1999}
follows that 
\begin{align*}
 & d(X_{(n+1)\Delta})-d(X_{n\Delta})=\int_{n\Delta}^{(n+1)\Delta}d'(X_{s})\sigma(X_{s})dW_{t}+\int_{n\Delta}^{(n+1)\Delta}d'(X_{s})b(X_{s})ds+\\
 & \qquad\qquad+\int_{n\Delta}^{(n+1)\Delta}\sigma^{2}(X_{s})\I(|X_{s}-\alpha|\leq\Delta^{4/9})ds-\Delta^{4/9}L_{n\Delta,(n+1)\Delta}(\alpha-\Delta^{4/9})\\
 & \qquad\qquad-\Delta^{4/9}L_{n\Delta,(n+1)\Delta}(\alpha+\Delta^{4/9}):=\int_{n\Delta}^{(n+1)\Delta}d'(X_{s})\sigma(X_{s})dW_{t}+D_{n}.
\end{align*}
First, we will bound the sum of the martingale terms. Since martingale
increments are uncorrelated, using Itô isometry, we obtain that
\begin{alignat*}{1}
 & \E\Big[\Big|\sum_{n=0}^{N-1}\I_{[0,\alpha)}(X_{n\Delta})\int_{n\Delta}^{(n+1)\Delta}d'(X_{s})\sigma(X_{s})dW_{t}\Big|^{2}\Big]=\\
 & \qquad=\sum_{n=0}^{N-1}\E\Big[\I_{[0,\alpha)}(X_{n\Delta})\int_{n\Delta}^{(n+1)\Delta}(d'(X_{s})\sigma(X_{s}))^{2}ds\Big]\\
 & \qquad\lesssim\Delta^{\frac{8}{9}}\E\Big[\int_{0}^{1}\I(|X_{s}-\alpha|\leq\Delta^{\frac{4}{9}})ds\Big]=\Delta^{\frac{8}{9}}\int_{\alpha-\Delta^{\frac{4}{9}}}^{\alpha+\Delta^{\frac{4}{9}}}\E[\mu_{1}(x)]dx\lesssim\Delta^{\frac{4}{3}},
\end{alignat*}
 where the last inequality follows from (\ref{eq:bound on moments of mu_T}).
Now, we will bound the sum of the finite variation terms: $\sum_{n=0}^{N-1}\I_{[0,\alpha)}(X_{n\Delta})D_{n}$.
Note first, that since $b$ is uniformly bounded, we have 
\[
\sum_{n=0}^{N-1}\I_{[0,\alpha)}(X_{n\Delta})\big|\int_{n\Delta}^{(n+1)\Delta}d'(X_{s})b(X_{s})ds\big|\lesssim\Delta^{4/9}\int_{0}^{1}\I(|x-\alpha|\leq\Delta^{4/9})\mu_{1}(x)dx\lesssim\Delta^{8/9}\|\mu_{1}\|_{\infty}.
\]
Since by the inequality (\ref{eq:bound on moments of mu_T}) $\|\mu_{1}\|_{\infty}$
has all moments finite, the root mean squared value of this sum is
of smaller order than $\Delta^{2/3}.$ Now, note that since on the
event $\RG_{1}$ $\mc<\Delta^{4/9}$, condition $X_{n\Delta}<\alpha$
implies $L_{n\Delta,(n+1)\Delta}(\alpha+\Delta^{4/9})=0.$ On the
other hand, whenever $L_{n\Delta,(n+1)\Delta}(\alpha-\Delta^{4/9})\neq0$
we must have $X_{n\Delta}<\alpha$ . Hence
\[
\sum_{n=0}^{N-1}\I_{[0,\alpha)}(X_{n\Delta})(\Delta^{4/9}L_{n\Delta,(n+1)\Delta}(\alpha-\Delta^{4/9})+\Delta^{4/9}L_{n\Delta,(n+1)\Delta}(\alpha+\Delta^{4/9}))=\Delta^{4/9}L_{1}(\alpha-\Delta^{4/9}).
\]
Using first the Cauchy-Schwarz inequality and then the regularity
of the local time (see \citep[Chapter VI, Corollary 1.8 and the remark before]{RevuzYor:1999})
we obtain 
\begin{alignat*}{1}
\E\Big[\big|\Delta^{4/9}L_{1}(\alpha-\Delta^{4/9})-\int_{\alpha-\Delta^{4/9}}^{\alpha}L_{1}(x)dx\big|^{2}\Big] & \leq\Delta^{4/9}\int_{\alpha-\Delta^{4/9}}^{\alpha}\E[|L_{1}(x)-L_{1}(\alpha-\Delta^{4/9})|^{2}]dx\\
 & \lesssim\Delta^{4/9}\int_{\alpha-\Delta^{4/9}}^{\alpha}|x-(\alpha-\Delta^{4/9})|dx\lesssim\Delta^{4/3}.
\end{alignat*}
Consequently, to prove (\ref{eq:increment canceling Partial}) we
just have to argue that the root mean squared error of
\begin{alignat}{1}
 & \int_{\alpha-\Delta^{4/9}}^{\alpha}L_{1}(x)dx-\sum_{n=0}^{N-1}\I_{[0,\alpha)}(X_{n\Delta})\int_{n\Delta}^{(n+1)\Delta}\sigma^{2}(X_{s})\I(|X_{s}-\alpha|\leq\Delta^{4/9})ds\nonumber \\
 & \qquad=\sum_{n=0}^{N-1}\int_{n\Delta}^{(n+1)\Delta}(\I(X_{s}<\alpha)-\I(X_{n\Delta}<\alpha))\sigma^{2}(X_{s})\I(|X_{s}-\alpha|\leq\Delta^{4/9})ds\nonumber \\
 & \qquad=\sum_{n=0}^{N-1}\int_{n\Delta}^{(n+1)\Delta}(\I(X_{s}<\alpha)-\I(X_{n\Delta}<\alpha))\sigma^{2}(X_{s})ds\label{eq:Partial step. Error with sigma}
\end{alignat}
is of order $\Delta^{2/3}.$ From the Lipschitz property of $\sigma^{2}$
follows that
\begin{align*}
 & \Big|\sum_{n=0}^{N-1}\int_{n\Delta}^{(n+1)\Delta}(\I(X_{s}<\alpha)-\I(X_{n\Delta}<\alpha))(\sigma^{2}(X_{s})-\sigma^{2}(\alpha))ds\Big|\lesssim\\
 & \qquad\qquad\lesssim\sum_{n=0}^{N-1}\int_{n\Delta}^{(n+1)\Delta}\I(|X_{s}-\alpha|\leq\Delta^{\frac{4}{9}})\Delta^{\frac{4}{9}}ds\lesssim\Delta^{\frac{4}{9}}\int_{\alpha-\Delta^{4/9}}^{\alpha+\Delta^{4/9}}\mu_{1}(dx)\lesssim\Delta^{\frac{8}{9}}\|\mu_{1}\|_{\infty}.
\end{align*}
Thus, by (\ref{eq:bound on moments of mu_T}), we reduced (\ref{eq:Partial step. Error with sigma})
to 
\[
\int_{0}^{1}\I(X_{s}<\alpha)ds-\frac{1}{N}\sum_{n=0}^{N-1}(\I(X_{n\Delta}<\alpha),
\]
which is of the right order by (\ref{eq:estimating occupation time}).
We conclude that (\ref{eq:increment canceling Partial}) holds.

\emph{Step 2.} Consider the time reversed process $Y_{t}=X_{1-t}$
. Since $X$ is reversible, the process $Y$, under the measure $\P,$
has the same law as $X$. Furthermore, the occupation density and
the modulus of continuity of processes $Y$ and $X$ are identical,
hence $\RG_{1}$ is a ``good'' event also for $Y$. Inequality (\ref{eq:increment canceling Partial})
is equivalent to
\[
\E\Big[\I_{\RG_{1}}\cdot\Big|\sum_{m=0}^{N-1}\I_{[0,\alpha)}(Y_{m\Delta})(d(Y_{(m+1)\Delta})-d(Y_{m\Delta}))\Big|^{2}\Big]^{\frac{1}{2}}\lesssim\Delta^{\frac{2}{3}}.
\]
Substituting $n=N-m$ we obtain
\[
\sum_{m=0}^{N-1}\I_{[0,\alpha)}(Y_{m\Delta})(d(Y_{(m+1)\Delta})-d(Y_{m\Delta}))=-\sum_{n=0}^{N-1}\I_{[0,\alpha)}(X_{(n+1)\Delta})(d(X_{(n+1)\Delta})-d(X_{n\Delta})).\tag*{{\qedhere}}
\]

\end{proof}

\section{\label{sec:proof in LF}Low-frequency analysis}

\subsection{Spectral estimation method}

In 1998 \citet{HansenScheinkmanTouzi1998:} explained how the coefficients
of a diffusion process are related to the spectral properties of its
infinitesimal generator. In this section we want to shortly introduce
the main idea of their method.

The generator $L$ of the reflected diffusion $X$ is an unbounded
operator on $L^{2}$ with 
\begin{eqnarray*}
\operatorname{dom}(L) & = & \{f\in H^{2}:f'(0)=f'(1)=0\}\\
Lf(x) & = & \mu^{-1}(x)\big(\frac{1}{2}\sigma^{2}(x)\mu(x)f'(x)\big)',\text{ for f}\in\operatorname{dom}(T).
\end{eqnarray*}
Spectral properties of $L$ are discussed in the appendix Section
\ref{sec:Eigenvalue problem for form a}. Seen as an operator on the
equivalent Hilbert space $L^{2}(\mu)$, the generator $L$ is elliptic,
self-adjoint and has a compact resolvent operator. Consequently, the
eigenproblem 
\begin{eigenproblem}
\label{EigenProb: Generator L}Find $(\zeta,u)\in\R\times L^{2}$,
with $u\neq0$, such that 
\[
Lu=\zeta u.
\]

\end{eigenproblem}
has countably many non-positive eigenvalues $0=\zeta_{0}>\zeta_{1}>\zeta_{2}\geq...,$
with $\mu-$orthogonal eigenfunctions $(u_{i})_{i=0,...}$. The eigenvalue
$\zeta_{1}$ is simple and the corresponding eigenfunction $u_{1}$
is strictly monotone, see Proposition \ref{prop:Properties of the eigenproblem for sigma and mu}.
The main idea of the spectral estimation method is that the diffusion
coefficient $\sigma^{2}$ can be expressed in terms of the invariant
density $\mu$ and the eigenpair $(\zeta_{1},u_{1})$ (c.f. \citet[Eq. 5.2]{HansenScheinkmanTouzi1998:}):
\begin{equation}
\sigma^{2}(x)=\frac{2\zeta_{1}\int_{0}^{x}u_{1}(y)\mu(y)dy}{u_{1}'(x)\mu(x)}.\label{eq:LF volatility from eigenpair formula}
\end{equation}

\subsection{Estimation error of the invariant measure}

From now on we take the Assumptions \ref{assu: Initial condition}
and \ref{assu: Set Theta} as granted. Fix $\Delta>0$ and $0<a<b<1$.
Set $J\sim N^{1/5}.$ Since the generator $L$ has a spectral gap,
diffusion $X$ is geometrically ergodic. Below we state general bounds
on the variance of integrals with respect to the empirical measure
$\hat{\mu}_{N}$, which are due to the mixing property of the observed
sample $(X_{n\Delta})_{n=0,...,N}$. For the proof we refer to \citet[Lemma 9]{ChorowskiTrabs:2015}. 
\begin{lem}
\label{lem:Uniform variance bounds for low frequency}For any $v,u\in L^{2}([0,1])$
we have 
\begin{align*}
\Var\Big[\int_{0}^{1}v(x)\hat{\mu}_{N}(dx)\Big] & \lesssim N^{-1}\|v\|_{L^{2}}^{2},\\
\Var\Big[\frac{1}{N}\sum_{n=0}^{N-1}v(X_{n\Delta})u(X_{(n+1)\Delta})\Big] & \lesssim N^{-1}\|v\cdot P_{\Delta}u\|_{L^{2}}^{2}.
\end{align*}
\end{lem}
\begin{cor}
\label{cor:Uniform bounds on the estimator of invariant measure}There
exists a high probability event $\TG_{1}$, with $\P(\Omega\setminus\TG_{1})\lesssim N^{-1}J^{2}$,
such that, for any $1\leq j\leq J$, on $\TG_{1}$ holds 
\[
J\int_{\frac{j-1}{J}}^{\frac{j}{J}}\hat{\mu}_{N}(dx)\sim1.
\]
\end{cor}
\begin{proof}
Since the invariant density $\mu$ is uniformly bounded on $\Theta$,
there exist constants $0<c<C$ s.t. $c\leq J\int_{\frac{j-1}{J}}^{\frac{j}{J}}\mu(x)dx\leq C$.
Let 
\[
\TG_{1}=\Big\{\forall j=1,...,J\text{ holds }\Big|\int_{0}^{1}\psi_{j}'(x)\hat{\mu}_{N}(dx)-\int_{0}^{1}\psi_{j}'(x)\mu(x)dx\Big|\leq\frac{c}{2J}\Big\}.
\]
Using first the Markov inequality and then $\|\psi_{j}'\|_{L^{2}}^{2}=J^{-1}$
with $\|\psi_{j}'\|_{L^{2}}^{2}=J^{-1}$ we conclude that the claim
holds. 
\end{proof}

\subsection{Proof of Theorem \ref{thm:Low Frequency Error}}

First, we state the approximation properties of the spaces $V_{J}$. 
\begin{defn}
Denote by $\pi_{J}$ and $\pi_{J}^{\mu}$ the $L^{2}$ and $L^{2}(\mu)-$orthogonal
projections on $V_{J}$ respectively. 
\end{defn}
Since $V_{J}$ is the space of linear spline functions with regular
knots at $\{0,\frac{1}{J},\frac{2}{J},...,\frac{J-1}{J},1\}$, it
satisfies the following Jackson and Bernstein type inequalities: 
\begin{align}
\|(I-\pi_{J})f\|_{H^{k}} & \lesssim J^{-(2-k)\alpha}\|f\|_{C^{1,\alpha}}\text{ for }f\in C^{1,\alpha}([0,1])\text{ and }k=0,1,\label{eq:Jackson's inq.}\\
\|v\|_{H^{1}} & \lesssim J\|v\|_{L^{2}}\text{ for }v\in V_{J}.\label{eq:Bernstein's inq.}
\end{align}

\begin{defn}
Denote by $(\phi_{j})_{j=0,...,J}$ the Franklin system on $[0,1]$,
i.e. the $L^{2}-$orthogonal basis of $V_{J}$, obtained from the
Schauder algebraic basis by the Gram\textendash Schmidt orthonormalization
procedure. 
\end{defn}
For construction and properties of the Franklin system we refer to
\citet{Ciesielski:1963}. In particular, basis functions $(\phi_{j})_{j}$
satisfy the following uniform bound (cf. \citet[Theorem 5]{Ciesielski:1963}):
\begin{equation}
\Big\|\sum_{j=0}^{J}\phi_{j}^{2}\Big\|_{\infty}\lesssim J.\label{eq:Summation property of Franklin system}
\end{equation}

\begin{proof}[Proof of Theorem \ref{thm:Low Frequency Error}]
As noted in Section \ref{sub:Connection to GHR estimator}, the estimator
$(\hat{\zeta}_{1},\hat{u}_{1})$ is constructed in the exactly same
way as the eigenpair estimator in \citet{GobetHoffmannReiss:2004,ChorowskiTrabs:2015}.
Given the properties of the Franklin system, arguing as in \citet[Corollary 18]{ChorowskiTrabs:2015},
we obtain that there exists a high probability event $\TG_{2}$, with
$\P(\Omega\setminus\TG_{2})\lesssim N^{-2/5}$, such that 
\begin{align}
\E\big[\I_{\TG_{2}}\cdot(|\gamma_{1}-\hat{\zeta}_{1}|^{2}+\|u_{1}-\hat{u}_{1}\|_{H^{1}}^{2})\big]^{\frac{1}{2}}\lesssim N^{-1/5}.\label{eq:eigenpair error}
\end{align}
Furthermore, on the event $\TG_{2}$, we have $|\hat{v}_{1}|\sim1$
and $\|\hat{u}_{1}\|_{H^{1}}\lesssim1$.

Before we can prove the upper bound on the estimation error, we need
to face one more technical difficulty. Since the estimator $\hat{u}_{1}$
converges to the eigenfunction $u_{1}$ in the sense of mean $H^{1}$
norm, we can't postulate a uniform positive lower bound on $\inf_{x\in[a,b]}\hat{u}_{1}'(x)$.
Following \citet[Lemma 19]{ChorowskiTrabs:2015}, this difficulty
can be overcome by applying the threshold $\hat{\sigma}_{S}^{2}\wedge D$.
We conclude that there exists a high probability event $\TG_{3}\subset\TG_{2}\cap\TG_{1}$,
with $\P(\Omega\setminus\TG_{3})\lesssim N^{-2/5}$, such that on
$\TG_{3}$, for $j$ s.t. $[\frac{j-1}{J},\frac{j}{J}]\subset(a,b),$
we have 
\begin{align}
\hat{\sigma}_{S,j}^{2}\wedge D=\frac{-2\hat{\zeta}_{1}\int_{0}^{1}\psi_{j}(x)\hat{u}_{1}(x)\hat{\mu}_{N}(dx)}{(\hat{u}_{1,j}\vee c_{a,b})\int_{0}^{1}\psi_{j}'(x)\hat{\mu}_{N}(dx)}\wedge D,\label{eq:sigma_tilde}
\end{align}
for a deterministic constant $c{}_{a,b}>0$ satisfying $c_{a,b}\leq\inf_{x\in[a,b]}u'_{1}(x)$.

Having established (\ref{eq:eigenpair error}) and (\ref{eq:sigma_tilde}),
the plug-in error can be bounded by similar considerations as in the
proof of \citet[Theorem 7]{ChorowskiTrabs:2015}. 
\end{proof}
\appendix

\section{\label{sec:Construction and properties of X}Construction and properties
of a scalar diffusion with two reflecting barriers}

In this section, we construct a weak solution of the SDE (\ref{eq:SDE for X}).
In the next sections we will use the presented construction to generalize
properties of scalar diffusions to reflected processes. The main idea
of the following reasoning is to extend the diffusion coefficients
$b$ and $\sigma$ to the whole real line, apply general SDEs theory
to obtain a solution on $\R$ and finally project this solution to
the interval $[0,1]$ in a way that corresponds to the instantaneous
reflection. We refer the reader to \citep[I.23]{GihmanSkorohod:1972}
for a very similar construction of a diffusion on$[-1,1]$ with two
reflecting barriers.
\begin{defn}
\label{def:Construction of reflected dif.}Define $f:\R\to[0,1]$
by
\[
f(x)=\begin{cases}
x-2n & :2n\leq x<2n+1\\
2(n+1)-x & :2n+1\leq x<2n+2
\end{cases},\text{ for }n\in\N.
\]
Function $f$ is almost everywhere differentiable with the derivative
\[
f'(x)=\begin{cases}
1 & :2n<x\leq2n+1\\
-1 & :2n+1<x\leq2n+2
\end{cases}.
\]
For $\sigma,b:[0,1]\to\R$ we define the extended coefficients $\tilde{\sigma},\tilde{b}:\R\to\R$
by
\begin{eqnarray*}
\tilde{b}(x) & = & f'(x)\cdot b\circ f(x)\\
\tilde{\sigma}(x) & = & \sigma\circ f(x).
\end{eqnarray*}
\end{defn}
\begin{thm}
\label{thm:Strong solution of SDE for X}Grant Assumption \ref{assu: Set Theta}.
For every initial condition $x_{0}\in[0,1]$ that is independent of
the driving Brownian motion $W$ the SDE
\begin{eqnarray}
dY_{t} & = & \tilde{b}(Y_{t})dt+\tilde{\sigma}(Y_{t})dW_{t},\label{eq:SDE for Y}\\
Y_{0} & = & x_{0},\nonumber 
\end{eqnarray}
has a non-exploding unique strong solution. Define
\[
X_{t}=f(Y_{t}).
\]
The process $(X_{t},t\geq0)$ is a weak solution of the SDE (\ref{eq:SDE for X}).\end{thm}
\begin{proof}
$\tilde{b}$ is bounded and $\tilde{\sigma}'\in L_{loc}^{2}(\R)$.
Hence, the existence of a unique strong solution $(Y_{t},t\geq0)$
of the SDE (\ref{eq:SDE for Y}) follows from \citep[Theorem 4]{Veretennikov:1979}.
As discussed in the proof of \citep[Chapter 5, Proposition 5.17]{KaratzasShreve:1991}
the boundedness of $\tilde{b}$ prevents the explosion of the solution.
Process $Y$ is a continuous semimartingale, hence by \citep[Chapter VI Theorem 1.2]{RevuzYor:1999}
it admits a local time process $(L_{t}^{Y},t\geq0)$. By the Itô-Tanaka
formula (\citep[Chapter VI Theorem 1.5]{RevuzYor:1999}) process $X$
satisfies
\begin{eqnarray*}
X_{t} & = & x_{0}+\int_{0}^{t}\tilde{b}(Y_{s})f'(Y_{s})ds+\int_{0}^{t}\tilde{\sigma}(Y_{s})f'(Y_{s})dW_{s}+\sum_{n\in\Z}L_{t}^{Y}(2n)-\sum_{n\in\Z}L_{t}^{Y}(2n+1)\\
 & = & x_{0}+\int_{0}^{t}b(X_{s})ds+\int_{0}^{t}\sigma(X_{s})dB_{s}+K_{t},
\end{eqnarray*}
where $B_{t}=\int_{0}^{t}f'(Y_{s})dW_{s}$ and $K_{t}=\sum_{n\in\Z}L_{t}^{Y}(2n)-\sum_{n\in\Z}L_{t}^{Y}(2n+1)$.
Note that for any $T>0$ the path $(X_{t},0\leq t\leq T)$ is bounded,
hence $K$ is well defined. Process $B$ is a martingale with quadratic
variation
\[
\langle B\rangle_{t}=\int_{0}^{t}(f'(Y_{s}))^{2}ds=t.
\]
Hence, Lévy's characterization theorem implies that $B$ is a standard
Brownian motion. From the properties of the local time $L_{t}^{Y}$
follows that $K$ is an adapted continuous process with finite variation,
starting from zero and varying on the set $\bigcup_{n\in\Z}\{Y_{t}=2n\}\cup\{Y_{t}=2n+1\}\subseteq\{X_{t}\in\{0,1\}\}$.
Consequently, $X$ satisfies the SDE (\ref{eq:SDE for X}).
\end{proof}
Next, we use the above construction of a reflected diffusion process
to prove Brownian bounds on the moments of the modulus of continuity
of $X.$
\begin{proof}[Proof of Theorem \ref{thm:Moments of the modulus of continuity}]
 \citet{FischerNappo:2010} proved claimed upper bound for the standard
Brownian motion. We will now generalize their result to diffusions
with boundary reflection. 

\emph{Step 1.} Consider a martingale $M$ satisfying $dM_{t}=\sigma(X_{t})dW_{t}.$
By the Dambis, Dubins-Schwarz theorem, $M_{t}=B_{\int_{0}^{t}\sigma^{2}(X_{u})du}$
for some Brownian motion $B$. Consequently,
\[
|M_{t}-M_{s}|=\big|B_{\int_{0}^{t}\sigma^{2}(X_{u})du}-B_{\int_{0}^{s}\sigma^{2}(X_{u})du}\big|\leq\omega^{B}(|t-s|\|\sigma^{2}\|_{\infty}),
\]
where $\omega^{B}$ is the modulus of continuity of $B$. Thus, (\ref{eq:Moments of the modulus of continuity})
holds for the martingale $M$, with a constant that depends only on
the uniform upper bound on the volatility $\sigma.$

\emph{Step 2.} Consider a semimartingale $Y$ satisfying $dY_{t}=b(X_{t})dt+dM_{t}.$
Then
\[
|Y_{t}-Y_{s}|\leq\big|\int_{0}^{t}b(Y_{u})du-\int_{0}^{s}b(Y_{u})du\big|+|M_{t}-M_{s}|\leq|t-s|\|b\|_{\infty}+\omega^{M}(|t-s|).
\]
Consequently, (\ref{eq:Moments of the modulus of continuity}) holds
for the semimartingale $Y$, with a constant that depends only on
the upper bounds on $\sigma$ and $b$.

\emph{Step 3.} For $(\sigma,b)\in\Theta$ consider the reflected diffusion
process $X$ satisfying the SDE (\ref{eq:SDE for X}). Let
\begin{eqnarray*}
dY_{t} & = & \tilde{b}(Y_{t})dt+\tilde{\sigma}(Y_{t})dW_{t},\\
X_{t} & = & f(Y_{t}),
\end{eqnarray*}
where $\tilde{b},\tilde{\sigma}$ and $f$ are as in Definition \ref{def:Construction of reflected dif.}.
From Step 2, it follows that (\ref{eq:Moments of the modulus of continuity})
holds for the semimartingale $Y$ with a uniform constant on $\Theta$.
By the construction of the reflected process $X$ we have $|X_{s}-X_{t}|\leq|Y_{s}-Y_{t}|$.
We conclude that $\omega^{X}\leq\omega^{Y}$, hence the claim holds
for the reflected diffusion $X$.
\end{proof}

\section{Bilinear coercive form\label{sec:Eigenvalue problem for form a}}

Recall that $H^{1},H^{2}$ denote the $L^{2}-$Sobolev spaces on $[0,1]$
of order 1 and 2 respectively. For differentiable, strictly positive
functions $\sigma$ and $\mu$ consider an elliptic operator $T$
on $L^{2}([0,1]),$ with Neumann type domain $\operatorname{dom}(T)=\{v\in H^{2}:v'(0)=v'(1)=0\}$,
given in the divergence form by
\begin{equation}
Tv(x)=-\frac{(\sigma^{2}(x)\mu(x)v'(x))'}{2\mu(x)},\text{ for }v\in\operatorname{dom}(T).\label{eq:Divergence form for T}
\end{equation}
Note that the operator $-T$ is an infinitesimal generator of the
diffusion process on $[0,1]$ with instantaneous reflection at the
boundaries, volatility function $\sigma$ and an invariant measure
with density $\mu$. We want to analyze the eigenvalue problem for
$T$, i.e.
\begin{eigenproblem}
\label{EigenProb: Operator T}Find $(\lambda,w)\in\R\times\operatorname{dom}(T)$,
with $w\neq0,$ such that
\[
Tw=\lambda w.
\]

\end{eigenproblem}
Integrating by parts, one can check, that the eigenpairs of the Eigenproblem
\ref{EigenProb: Operator T} solve
\begin{eigenproblem}
\label{EigenProb: Operator T in weak form}Find $(\lambda,w)\in\R\times H^{1},$
with $w\neq0,$ such that
\begin{equation}
\int_{0}^{1}w'(x)v'(x)\sigma^{2}(x)\mu(x)dx=2\lambda\int_{0}^{1}w(x)v(x)\mu(x)dx\text{ for all }v\in H^{1}.\label{eq:Eigenproblem for T in H1}
\end{equation}

\end{eigenproblem}
Eigenproblem \ref{EigenProb: Operator T in weak form} is a weak formulation
of the Eigenproblem \ref{EigenProb: Operator T} for the associated
Dirichlet form $l(u,v)=\langle Tu,v\rangle_{\mu}$. The biggest advantage
of the weak formulation is that the Eigenproblem \ref{EigenProb: Operator T in weak form}
makes sense for any, not necessarily regular, functions $\mu$. When
$\mu$ is not differentiable, the Eigenproblem \ref{EigenProb: Operator T}
has no longer probabilistic interpretation in terms of the infinitesimal
generator. Nevertheless, such problems arise naturally when one considers
spectral estimation method with fixed time horizon, when the role
of the invariant measure is taken by the non differentiable occupation
density.

In what follows, we want to generalize the results of \citep{GobetHoffmannReiss:2004}
on the spectral properties of an infinitesimal generator, to the solutions
of the Eigenproblem \ref{EigenProb: Operator T in weak form} with
a Hölder regular function $\mu$. For $0<\alpha\leq1$ denote by $C^{\alpha}$the
space of $\alpha-$Hölder regular functions on $[0,1].$ Furthermore,
for $k\in\N$ let $C^{k,\alpha}$ be the space of $k-$times differentiable
functions with $k^{th}$derivative in $C^{\alpha}.$
\begin{defn}
For any given $0<d<D$ let 
\begin{align*}
\varTheta_{\alpha}:=\Big\{(\sigma,\mu)\in H^{1}([0,1])\times C^{\alpha}([0,1]) & :\|\sigma\|_{H^{1}},\|\mu\|_{C^{\alpha}}\leq D,\\
 & \qquad\inf_{x\in[0,1]}(\sigma(x)\wedge\mu(x))\geq d,\int_{0}^{1}\mu(x)dx=1\Big\}
\end{align*}

\end{defn}
Eigenproblem \ref{EigenProb: Operator T in weak form} is a conforming
eigenvalue problem for a bilinear coercive form on the Hilbert space
$L^{2}(\mu)$. \citep{Chatelin:1983} is a standard reference.
\begin{prop}
\label{prop:Properties of the eigenproblem for sigma and mu}Let $(\sigma,\mu)\in\varTheta_{\alpha}$.
The Eigenproblem \ref{EigenProb: Operator T in weak form} has countably
many solutions $(\lambda_{i},w_{i})_{i},$ with real nonnegative eigenvalues
$0=\lambda_{0}<\lambda_{1}\leq\lambda_{2}\leq...$ and $\mu-$orthogonal
eigenfunctions, satisfying Neumann boundary conditions $w_{i}'(0)=w_{i}'(1)=0$.
The smallest positive eigenvalue $\lambda_{1}$ is simple, the derivative
$w_{1}'$ of the corresponding eigenfunction is $1/2\wedge\alpha$
Hölder continuous and strictly monotone.\end{prop}
\begin{proof}
It is easy to check that for any $(\sigma,\mu)$ $\lambda_{0}=0$
and $w_{0}\equiv1$ form an eigenpair. Let $L_{0}^{2}(\mu)=\{v\in L^{2}(\mu):\int_{0}^{1}v(x)\mu(x)dx=0\}$
and $H_{0}^{1}(\mu)=L_{0}^{2}(\mu)\cap H^{1}$. $L_{0}^{2}(\mu)$
with the $L^{2}(\mu)$ inner product and $H_{0}^{1}(\mu)$ with $\langle u,v\rangle_{H^{1}(\mu)}=\langle u,v\rangle_{L_{0}^{2}(\mu)}+\int_{0}^{1}u'(x)v'(x)\mu(x)dx$
are Hilbert spaces. The identity embedding $I:H_{0}^{1}(\mu)\to L_{0}^{2}(\mu)$
is compact. 

For $u,v\in H_{0}^{1}(\mu)$ let
\[
l(u,v)=\int_{0}^{1}u'(x)v'(x)\sigma^{2}(x)\mu(x)dx.
\]
$l$ is a symmetric positive-definite bilinear form on $H_{0}^{1}(\mu)\times H_{0}^{1}(\mu).$
Furthermore, for any $u\in H_{0}^{1}(\mu)$ holds
\begin{equation}
c\|u\|_{H_{0}^{1}(\mu)}^{2}\leq l(u,u)\leq C\|u\|_{H_{0}^{1}(\mu)}^{2},\label{eq:Coercivity of form a}
\end{equation}
for some constants $0<c<C$ that depend only on $d,D$. Indeed, since
$\sigma$ and $\mu$ are uniformly bounded, we only have to show that
$\int_{0}^{1}u^{2}(x)dx\leq\int_{0}^{1}(u'(x))^{2}dx$. Consider $u\in C^{1}([0,1])\cap H_{0}^{1}(\mu)$.
Since $u$ is continuous and integrates to zero, there exists $x_{0}\in[0,1]$
s.t. $u(x_{0})=0$. Since $u(x)=\int_{x_{0}}^{x}u'(y)dy,$ the upper
bound $\|u\|_{L^{2}}\leq\|u'\|_{L^{2}}$ follows from the Cauchy-Schwarz
inequality. As continuous functions are dense in $H^{1}$, we conclude
that (\ref{eq:Coercivity of form a}) holds.

$l$ is the Dirichlet form of an unbounded operator $T$ on $L_{0}^{2}(\mu)$.
Define $D=\operatorname{dom}(T)$ as these $u\in H_{0}^{1}(\mu),$
that the functional $v\mapsto l(u,v)$ is continuous on $H_{0}^{1}(\mu)$
with norm $\|\cdot\|{}_{L^{2}(\mu)}$. By the definition of the weak
differentiability, domain $D=\{u:H_{0}^{1}(\mu):u'\sigma^{2}\mu\in H^{1}\}$.
Furthermore, $D$ is dense in $L_{0}^{2}(\mu)$ (see \citep[Exercise 4.51]{Chatelin:1983}).
For $u\in D,$ we define $Tu$ via the Riesz representation theorem
by $l(u,v)=\langle Tu,v\rangle_{L^{2}(\mu)}.$ Such defined $T$ is
an elliptic, densely defined, self-adjoint operator with compact resolvent
(see \citep[Proposition 4.17]{Chatelin:1983}). Consequently, $T$
has a discrete spectrum $(\lambda_{i})_{i=1,...}$, with all eigenvalues
positive and corresponding eigenfunctions $\mu-$orthogonal.

Integrating by parts the right hand side of (\ref{eq:Eigenproblem for T in H1}),
we obtain
\[
\int_{0}^{1}w_{i}'(x)\sigma^{2}(x)\mu(x)v'(x)dx=-2\lambda_{i}\int_{0}^{1}\int_{0}^{x}w_{i}(y)\mu(y)dyv'(x)dx\text{ for all }v\in H^{1}.
\]
Since $\{v':v\in H^{1}\}$ is dense in $L^{2}$, it follows that
\begin{equation}
w_{i}'(x)=\frac{2\lambda_{i}\int_{0}^{x}w_{i}(y)\mu(y)dy}{\sigma^{2}(x)\mu(x)}.\label{eq:Formula for eigenfunctions of a}
\end{equation}
By Sobolev embedding $\sigma^{2}$ is $1/2-$Hölder regular. Consequently
$w_{i}'$ lies in $C^{1/2\wedge\alpha}$. Since the eigenfunctions
$\mu-$integrate to zero, we deduce that $w_{i}'(0)=w_{i}'(1)=0$. 

Finally, we need to show that $\lambda_{1}$ is simple and that $w_{1}$
is strictly monotone. By the variational formula for the eigenpairs
of a self-adjoint operator
\begin{equation}
2\lambda_{1}=\inf_{u\in H_{0}^{1}(\mu)}\frac{\int_{0}^{1}(u'(x))^{2}\sigma^{2}(x)\mu(x)dx}{\int_{0}^{1}u^{2}(x)\mu(x)dx}.\label{eq:Formula for eigenvalue Lambda_one of a}
\end{equation}
Arguing as in \citep[Lemma 6.1]{GobetHoffmannReiss:2004}, we obtain
that $\int_{0}^{1}u^{2}(x)\mu(x)dx=\int_{0}^{1}\int_{0}^{1}m(y,z)u'(y)u'(z)dydz$
with $m(y,z)=\int_{0}^{y\wedge z}\mu(x)dx\int_{y\vee z}^{1}\mu(x)dx$.
We deduce that the derivative of the eigenfunction $w_{1}$ must have
a constant sign, otherwise we could reduce the ratio in (\ref{eq:Formula for eigenvalue Lambda_one of a})
by considering 
\[
\tilde{w}_{1}=w_{1}\I(w_{1}'\geq0)-w_{1}\I(w_{1}'\leq0).
\]
Hence, the set $\{x:w_{1}'(x)=0\}$ has zero Lebesgue measure. From
(\ref{eq:Formula for eigenfunctions of a}) follows that $w_{1}'(x)=0$
only for $x=0,1$, meaning that $w_{1}$ is strictly monotone on $(0,1)$.
Consequently, for any two eigenfunctions $w_{1}$ and $\bar{w}_{1}$,
which correspond to $\lambda_{1}$, the scalar product
\[
\int_{0}^{1}w_{1}(x)\bar{w}_{1}(x)\mu(x)dx=\int_{0}^{1}\int_{0}^{1}m(y,z)w'_{1}(y)\bar{w}_{1}'(z)dydz\neq0,
\]
hence the eigenspace corresponding to $\lambda_{1}$ is one dimensional.\end{proof}
\begin{prop}
\label{prop:Uniform bounds on eigenpairs on theta}The eigenvalues
$\lambda_{1},\lambda_{2}$ and the norm ratio $\|w_{1}\|_{C^{1,1/2\wedge\alpha}}/\|w_{1}\|_{L^{2}(\mu)}$
are uniformly bounded for all $(\sigma,\mu)\in\varTheta_{\alpha}$.
Furthermore, for every $0<a<b<1$, $\inf_{x\in[a,b]}|w_{1}'(x)|$
and the spectral gap $\lambda_{2}-\lambda_{1}$ have uniform lower
bounds on $\varTheta_{\alpha}$.\end{prop}
\begin{proof}
We adapt the notation from the proof of Proposition \ref{prop:Properties of the eigenproblem for sigma and mu}.
Choose $w_{1}$ normalized s.t. $\|w_{1}\|_{L^{2}(\mu)}=1$. We will
first argue that $\lambda_{1},\lambda_{2}$ and $\|w_{1}\|_{C^{1,1/2\wedge\alpha}}$
are uniformly bounded on $\Theta_{\alpha}$. From (\ref{eq:Coercivity of form a})
we imply that 
\[
\lambda_{1}=l(w_{1},w_{1})\geq c\|w_{1}\|_{H^{1}(\mu)}^{2}\geq c,
\]
 with $c>0$ depending only on the bounds on $\sigma$ and $\mu$.
It follows that the eigenvalues are uniformly separated from zero.
By the variational formula
\[
2\lambda_{2}=\inf_{\begin{array}[t]{c}
{\scriptstyle {\scriptstyle S\subset H^{1}}}\\
{\scriptstyle \text{dim}(S)=3}
\end{array}}\sup_{u\in S}\frac{\int_{0}^{1}(u'(x))^{2}\sigma^{2}(x)\mu(x)dx}{\int_{0}^{1}u^{2}(x)\mu(x)dx}\leq\inf_{\begin{array}[t]{c}
{\scriptstyle {\scriptstyle S\subset H^{1}}}\\
{\scriptstyle \text{dim}(S)=3}
\end{array}}\sup_{u\in S}\frac{D^{3}\int_{0}^{1}(u'(x))^{2}dx}{d\int_{0}^{1}u^{2}(x)dx}\leq4\pi^{2}\frac{D^{3}}{d},
\]
since $4\pi^{2}$ is the third eigenvalue of the negative Laplace
operator on $L^{2}([0,1])$ with Neumann boundary conditions. We conclude
that the eigenvalues $\lambda_{1}$ and $\lambda_{2}$ are uniformly
bounded. The uniform bound on $\|w_{1}\|_{C^{1,1/2\wedge\alpha}}$
follows from the representation (\ref{eq:Formula for eigenfunctions of a}).

We will now prove a uniform lower bound on the spectral gap $\lambda_{2}-\lambda_{1}$.
Assume by contradiction that for some sequence of coefficients $(\sigma_{n},\mu_{n})\in\varTheta_{\alpha}$
the corresponding spectral gaps $(\lambda_{n,2}-\lambda_{n,1})$ converge
to zero. Since $\varTheta_{\alpha}$ is compact in the uniform convergence
metric, we can assume that $(\sigma_{n},\mu_{n})$ converges uniformly
to some $(\sigma,\mu)\in\varTheta_{\alpha}$. We will argue that the
uniform convergence of the coefficients leads to convergence of the
eigenvalues, hence contradicts Proposition \ref{prop:Properties of the eigenproblem for sigma and mu}
(cf. \citep[proof of Proposition 6.5]{GobetHoffmannReiss:2004}).
However, since the function $\mu$ is embedded in the definition of
spaces $L_{0}^{2}(\mu)$ and $H_{0}^{1}(\mu)$, we need first to reduce
the Eigenproblem \ref{EigenProb: Operator T in weak form} to a universal
function space.

Let $U(x)=\int_{0}^{x}\mu(y)dy$ be the distribution function of $\mu$.
Substituting $U(x)=y,$ we find that the Eigenproblem \ref{EigenProb: Operator T in weak form}
is equivalent to
\begin{eqnarray*}
\int_{0}^{1}\tilde{w}'(x)\tilde{v}'(x)\tilde{\sigma}^{2}dx & = & 2\lambda\int_{0}^{1}\tilde{w}'(x)\tilde{v}'(x)dx\text{ for all }\tilde{v}\in H^{1}\\
\tilde{w} & = & w\circ U^{-1},
\end{eqnarray*}
 with $\tilde{\sigma}=(\sigma\mu)\circ U^{-1}$. Consider $(\tilde{\sigma}_{n})_{n}$
and $\tilde{\sigma}$ corresponding to $(\sigma_{n},\mu_{n})$ and
$(\sigma,\mu)$ respectively. Note that $\tilde{\sigma}_{n}$ converges
to $\tilde{\sigma}$ in the uniform norm. Denote $L_{0}^{2}=L_{0}^{2}(1)$
and $H_{0}^{1}:=H_{0}^{1}(1)$. For $u,v\in H_{0}^{1}$ denote 
\[
\tilde{l}_{n}(u,v)=\int_{0}^{1}u'(x)v'(x)\tilde{\sigma}_{n}(x)^{2}dx
\]
 and by $\tilde{T}_{n}$ the corresponding operators on $L_{0}^{2}$.
Recall that the operators $\tilde{T}_{n}$ are unbounded and self-adjoint
on $L_{0}^{2}$, with dense domains $\tilde{D}_{n}$. Domains $\tilde{D}_{n}$
do not have to possess a common core, which is needed to study the
convergence of the sequence $(\tilde{T}_{n})_{n}$. We circumvent
this difficulty by introducing inverse operators $\tilde{R}_{n}=\tilde{T}_{n}^{-1}$.
Using the divergence formula (\ref{eq:Divergence form for T}) for
$\tilde{T}_{n}$, we check that for $u\in L_{0}^{2}$ 
\begin{equation}
\tilde{R}_{n}u(x)=-2\int_{0}^{x}\tilde{\sigma}_{n}^{-2}(y)\int_{0}^{y}u(z)dz+c_{n}(u),\label{eq:Formula for resolvent of T_n}
\end{equation}
where $c_{n}(u)\in\R$ is such that $\int_{0}^{1}\tilde{R}_{n}u(x)dx=0.$
The convergence $\tilde{\sigma}_{n}\to\tilde{\sigma}$ in $C^{1}[(0,1)]$
implies that operators $\tilde{R}_{n}$ converge to $\tilde{R}$ in
the operator norm on $L_{0}^{2}$. By \citep[Proposition 5.28]{Chatelin:1983}
this entails the regular convergence, which, by \citep[Theorem 5.20]{Chatelin:1983},
is equivalent to the strongly stable convergence. Finally, \citep[Proposition 5.6]{Chatelin:1983}
ensures the convergence of the eigenvalues with preservation of their
multiplicities.

Set $0<a<b<1$. We finally have to prove the uniform lower bound on
$\inf_{x\in[a,b]}|w_{1}'(x)|$. We will use the same indirect arguments
as when bounding the spectral gap. Assume that for some sequence $(\sigma_{n},\mu_{n})\in\varTheta_{\alpha}$,
with $(\sigma_{n},\mu_{n})$ converging in the uniform norm to $(\sigma,\mu)\in\varTheta_{\alpha}$,
the corresponding eigenfunctions $w_{1,n}$ satisfy $\inf_{n}\inf_{x\in[a,b]}|w_{1,n}'(x)|=0$.
Arguing as for the spectral gap, we reduce the problem to bounded
operators $(\tilde{R}_{n})_{n}$ and $\tilde{R}$. From formula (\ref{eq:Formula for resolvent of T_n})
we deduce that the uniform convergence of coefficients implies $\tilde{R}_{n}\to\tilde{R}$
in the operator norm on $C([0,1])$. We conclude, that the eigenfunctions
converge in the uniform norm, which contradicts Proposition \ref{prop:Properties of the eigenproblem for sigma and mu}.\end{proof}
\begin{eigenproblem}
\label{EigenProb: restricted to V_J}Let $V_{J}$ be a finite dimensional
subspace of $L^{2}$. Find $(\lambda_{J},w_{J})\in\R\times V_{J},$
with $w_{J}\neq0$ such that
\[
\int_{0}^{1}w'(x)v'(x)\sigma^{2}(x)\mu(x)dx=\lambda\int_{0}^{1}w(x)v(x)\mu(x)dx\text{ for any }v\in V_{J}.
\]
\end{eigenproblem}
\begin{prop}
\label{prop:Properties of eigenvalue problem for sigma mu on V_J}Let
$(V_{J})_{J=1,...}$ be a sequence of approximation spaces satisfying
the following Jackson's type inequality:
\[
\|(I-\pi_{J})v\|_{H^{1}}\leq CJ^{-\alpha}\|v\|_{C^{1,\alpha}}\text{ for }v\in C^{1,\alpha},
\]
where $\pi_{J}$ is the $L^{2}-$orthogonal projection on $V_{J}$
and $C>0$ some universal constant . Furthermore, assume that every
$V_{J}$ contains constant functions.

For $(\sigma,\mu)\in\varTheta_{\alpha}$ the Eigenproblem \ref{EigenProb: restricted to V_J}
has $\operatorname{dim}(V_{J})$ solutions $(\lambda_{J,i},w_{J,i})_{i}$
with real eigenvalues $0=\lambda_{J,0}<\lambda_{J,1}<\lambda_{J,2}\leq...\leq\lambda_{J,\operatorname{dim}(V_{J})-1}$.
For $J$ big enough, the eigenvalue $\lambda_{J,1}$ and the spectral
gap $\lambda_{J,2}-\lambda_{J,1}$ are uniformly bounded on $\varTheta_{\alpha}.$ \end{prop}
\begin{proof}
We adapt the notation from the proof of Proposition \ref{prop:Properties of the eigenproblem for sigma and mu}.
By the Lax-Milgram theorem, there exists an isomorphism $S_{l}:H_{0}^{1}(\mu)\to H_{0}^{1}(\mu)$
such that
\[
l(S_{l}v,u)=\langle v,u\rangle_{H^{1}(\mu)},\text{ for all }v,u\in H_{0}^{1}(\mu).
\]
Note that since for any $v\in L_{0}^{2}(\mu)$ the functional $H_{0}^{1}(\mu)\ni u\longmapsto\langle v,u\rangle_{L^{2}(\mu)}\in\R$
is continuous on $H_{0}^{1}(\mu),$ by the Riesz representation theorem
there exists a continuous operator $J:L_{0}^{2}(\mu)\to H_{0}^{1}(\mu)$
such that
\[
\langle v,u\rangle_{L^{2}(\mu)}=\langle Jv,u\rangle_{H^{1}(\mu)}.
\]
Define the operator $B_{l}=S_{l}\circ J\circ I$, where $I$ is the
identity embedding of $H_{0}^{1}(\mu)$ into $L_{0}^{2}(\mu).$ By
(\ref{eq:Coercivity of form a}), the form $l$ defines an equivalent
norm on $H_{0}^{1}(\mu)$. Note that $B_{l}$ is a self-adjoint and
compact operator on the Hilbert space $H_{0}^{1}(\mu)$ with $l-$induced
inner product. Consider $(\lambda_{i},w_{i})$, a solution of the
Eigenproblem \ref{EigenProb: Operator T in weak form}. For any $v\in H_{0}^{1}(\mu)$
we have
\[
l(w_{i},v)=\lambda_{i}\langle w_{i},v\rangle_{L^{2}(\mu)}=\lambda_{i}\langle Jw_{i},v\rangle_{H^{1}(\mu)}=\lambda_{i}l(S_{l}Jw_{i},v)=l(\lambda_{i}B_{l}w_{i},v),
\]
hence $(\lambda_{i}^{-1},w_{i})$ is an eigenpair of the operator
$B_{l}$. In particular, Proposition \ref{prop:Properties of the eigenproblem for sigma and mu}
implies that the biggest eigenvalue $\lambda_{1}^{-1}$ is simple.

Denote by $\pi_{J}^{l}$ the $l-$orthogonal projection on the subspace
$V_{J}$. Define the operator $B_{l,J}=\pi_{J}^{l}B_{l}\pi_{J}^{l}$.
Since $B_{l,J}$ is a self-adjoint operator on $V_{J},$ with the
$l-$induced inner product, it has $\operatorname{dim}(V_{J})-1$
solutions $(\lambda_{J,i}^{-1},w_{J,i})_{i},$ with the eigenvalues
$\lambda_{J,1}^{-1}\geq\lambda_{J,2}^{-1}\geq...\geq\lambda_{J,\operatorname{dim}(V_{J})-1}^{-1}$.
Analogously as for the operator $B_{l}$, we check that $(\lambda_{J,i},w_{J,i})$
are solutions of the finite dimensional Eigenproblem \ref{EigenProb: restricted to V_J}.
From (\ref{eq:Coercivity of form a}) together with the uniform bound
on $\mu$ follows that 
\[
\|(I-\pi_{n}^{l})w_{1}\|_{l}\leq\|(I-\pi_{n}^{l})(I-\pi_{J})w_{1}\|_{l}\leq2\|(I-\pi_{J})w_{1}\|_{l}\leq C\|(I-\pi_{J})w_{1}\|_{H^{1}},
\]
for some, uniform on $\varTheta_{\alpha}$, constant $C$. Using Jackson's
inequality, the uniform bound on the Hölder norm of $w_{1}$ and uniform
bounds on the eigenvalues $\lambda_{1},\lambda_{2},$ we conclude
that, for $J$ large enough, 
\[
\|(I-\pi_{n}^{l})w_{1}\|_{l}<\frac{\lambda_{1}^{-1}-\lambda_{2}^{-1}}{6\lambda_{1}^{-1}}.
\]
The claim follows from \citep[Theorem 25]{ChorowskiTrabs:2015}.
\end{proof}
\bibliographystyle{apalike}
\bibliography{bibliography}

\end{document}